\documentclass[12pt]{article}

\usepackage{amssymb}
\usepackage{amsfonts}
\usepackage[fleqn]{amsmath}
\usepackage{geometry}
\usepackage{graphicx}
\usepackage{subcaption}
\usepackage{natbib}
\usepackage[utf8x]{inputenc}
\usepackage{float}
\usepackage{enumitem}
\usepackage[font=scriptsize,justification=centering]{caption}
\usepackage{setspace}
\usepackage{hyperref}
\usepackage{bigints}
\usepackage[usenames, dvipsnames]{xcolor}
\usepackage{ed}      	
\usepackage{pgf}
\usepackage{tikz}
\usetikzlibrary{patterns}
\hypersetup{ pdftitle= Multidimensional Taxation, pdfauthor=Job Boerma and Aleh Tsyvinski and Alexander Zimin,
colorlinks, citecolor=black,
filecolor=black,                                                linkcolor=Blue,                                                                                                pdftitle=title,                                                pdfauthor=author,                                                bookmarks,
citecolor= Blue,                                                pdftex}

\usepackage{amsthm}
\usepackage{mathrsfs}
\usepackage{algorithm2e}
\theoremstyle{definition}
\newtheorem{theorem}{Theorem}
\newtheorem*{theorem*}{Theorem}

\newtheorem{claim}{Claim}

\newtheorem{corollary}[theorem]{Corollary}
\newtheorem*{corollary*}{Corollary}

\newtheorem*{definition*}{Definition}

\newtheorem{lemma}{Lemma}

\newtheorem{proposition}[theorem]{Proposition}

\usepackage[noabbrev]{cleveref}

\usetikzlibrary{arrows,decorations.markings}

\begin{document}

\title{Bunching and Taxing Multidimensional Skills\thanks{%
\baselineskip12.5pt We thank Hector Chade, Paolo Martellini, Chris Moser, Emmanuel Saez, Florian Scheuer, Andrew Shephard, Chris Taber and especially Jean-Charles Rochet for detailed insightful comments.}}
\author{ Job Boerma \\ 
{\small { \hspace{2 cm} University of Wisconsin-Madison  \hspace{2 cm} } }\\ \vspace{-0.5 cm}
\and  Aleh Tsyvinski \\
{\small { \hspace{4 cm}  Yale University \hspace{4 cm} } }\\ \vspace{-0.5 cm}
\and  Alexander P. Zimin \\
{\small { \hspace{4 cm}  MIT \hspace{4 cm} } }\\
}
\date{\vspace{0.2 cm} March 2025 \vspace{0.2 cm}}
\maketitle

\vspace{-0.4cm}
\begin{abstract}
\fontsize{12.0pt}{18.0pt} \selectfont

\noindent We characterize optimal policy in a multidimensional nonlinear
taxation model with bunching. We develop an empirically relevant model
with cognitive and manual skills, firm heterogeneity, and labor market
sorting. We first derive two  conditions for the optimality of taxes that take
into account bunching. The first condition $-$ a stochastic dominance optimal tax condition $-$
shows that at the optimum the schedule of benefits
dominates the schedule of distortions in terms of second-order
stochastic dominance. The second condition $-$ a global optimal tax formula $-$ provides a representation
that balances the local costs and benefits of optimal taxation while
explicitly accounting for global incentive constraints. Second, we
use Legendre transformations to represent our problem as a linear program.
This linearization allows us to solve the model quantitatively and
to precisely characterize bunching. At
an optimum, 10 percent of workers is bunched.
We introduce two notions of bunching -- blunt bunching and targeted
bunching. Blunt bunching constitutes 30 percent of all bunching, occurs
at the lowest regions of cognitive and manual skills, and lumps the
allocations of these workers resulting in a significant distortion.
Targeted bunching constitutes 70 percent of all bunching and recognizes
the workers' comparative advantage. The planner separates workers
on their dominant skill and bunches them on their weaker skill, thus
mitigating distortions along the dominant skill dimension. 


\end{abstract}

\renewcommand{\thefootnote}{\fnsymbol{footnote}} \renewcommand{%
\thefootnote}{\arabic{footnote}} 
\thispagestyle{empty} \setcounter{page}{0}

\fontsize{11.5pt}{21.0pt} \selectfont


\newpage

\section{Introduction}

We make significant progress in analyzing multidimensional optimal
nonlinear income taxation problems with bunching. This is one of the
important open questions in the theory and practice of optimal taxation.
Our paper is the first to solve for optimal multidimensional taxation with bunching
in an empirically relevant model of wage determination.


The key difficulty of analyzing multidimensional optimal tax
problems lies in characterizing regions of bunching. Bunching occurs
when workers of different types receive identical allocations. \citet{Kleven:2009}
establish the importance of bunching in a model of couples taxation
in which one partner makes only an extensive margin labor supply choice.
Little is known about optimal taxation and the nature of bunching
in more general settings. At the same time, a large literature in
labor economics emphasizes the importance of multidimensional skills
and labor market sorting in determining wage dispersion.

We develop an empirically relevant model that incorporates three important
elements of wage dispersion. Workers differ in both manual and in
cognitive skills, firms differ in productivity, and workers' output
depends on the firms they work for and coworkers they work with. We
characterize equilibrium for the positive model in closed form and use this closed-form solution to identify the underlying
multidimensional skill distribution. 

The characterization of optimal taxes in our model is based on two
main theoretical insights. First, we derive two conditions for optimal
taxes that take into account bunching. The first condition $-$ a stochastic
dominance optimal taxation condition $-$ shows that at the optimum the
benefits and the costs are not necessarily equated for each skill
level but rather the entire schedule of benefits dominates the entire
schedule of distortions in terms of second-order stochastic dominance.
The second condition $-$ a global optimal tax formula $-$ provides
a representation that balances the local costs and benefits of optimal
taxation while explicitly accounting for global incentive constraints.
These optimal tax conditions generalize the classic unidimensional
optimal taxation conditions in \citet{Mirrlees:1971}, \citet{Diamond:1998}, and \citet{Saez:2001} to a multidimensional optimal taxation problem, accounting for global
incentive constraints and bunching. Second, we use Legendre transforms
to represent our problem as a linear program. Legendre transforms
are a powerful tool from convex analysis that allow to represent a
convex function by a family of its tangent lines. This linearization
enables us to fully solve the model quantitatively and, in particular,
precisely characterize the patterns of bunching.

We find that 10 percent of all workers are bunched at the optimal
allocation for our empirically estimated economy. We show that workers
are bunched with other workers who are better in one dimension but
worse in the other dimension. Moreover, a sizable portion of bunching
is nonlocal.

We introduce two notions of bunching: blunt bunching and targeted
bunching. Blunt bunching occurs at the lowest regions of cognitive
and manual skills. The planner does not distinguish workers' cognitive skills
from their manual skills and lumps their allocations together by an
index of their skills. This is a blunt tool for providing incentives
as it creates significant distortions leading to high marginal taxes $-$ 30 percent of all bunching is blunt. Targeted bunching recognizes
the workers' comparative advantage. The planner separates workers
on their dominant skill while bunching them on their weaker skill $-$ 70 percent of all bunching is targeted.


We now discuss our model and results in more detail. We consider the optimal policy problem which we formulate as a mechanism
design problem \citep{Mirrlees:1971}. The planner chooses consumption
allocations, allocations of cognitive and manual tasks, and the assignment
of workers to coworkers and firms subject to incentive constraints
that workers truthfully report their type. The primary difficulty
in analyzing multidimensional optimal taxation problems is in characterizing
regions of bunching. \citet{Rochet:1998} shows that bunching is generic
in the multidimensional multiproduct monopolist problem that is closely
related to the multidimensional optimal taxation problem. An important
paper by \citet{Kleven:2009} solves a multidimensional optimal taxation
model under the restriction that one of the allocations is binary
and argues for the importance of bunching in their setup. In our model,
both the allocations of the manual and cognitive tasks are instead
unrestricted. For this general setup, no characterization of optimal
tax policy is known.

Our first theoretical result is the derivation of two optimality conditions
that take into account the regions of bunching. First, we show that,
at the optimum, the utility and revenue benefits from the entire schedule
of taxes second-order stochastically dominate the costs of distortions
they induce. Without bunching, this tradeoff is made locally and leads
the planner to equalize the marginal benefits and costs of taxes as
in a unidimensional problem. We show that, when there is bunching,
the planner no longer equates the benefits and the costs of the taxes
at each worker skill level but instead requires that the entire schedule
of the benefits of taxes second-order stochastically dominate the
entire schedule of distortions, making this tradeoff nonlocal. Second,
we derive an optimal tax formula for the multidimensional taxation
problem that holds with equality. We show that the local tradeoffs
have to be augmented with an additional term that accounts for global
incentive constraints. Specifically, the additional term modifies
the social welfare weights through a convexity correction.\footnote{When there is no bunching, a classic pointwise optimality condition
holds that can be rewritten in terms of a multidimensional ABC taxation
formula similar to the unidimensional tax formula in \citet{Diamond:1998}
and \citet{Saez:2001}. The absence of bunching is equivalent to the
indirect utility function being strongly convex and the first-order
approach being valid. \citet[p. 23]{Kleven:2006} derive such a multidimensional
ABC formula without bunching.} Both formulas are new to the literature on optimal taxation.


Our next main theoretical insight is to transform the planner problem
to a linear program. This is a key step that enables computation
of the bunching regions. Legendre transformations linearize a convex
function by replacing it with the upper envelope of all its tangent
lines. The Legendre transform allows us to translate the multidimensional
optimal taxation problem into a linear program that can be  analyzed
quantitatively with high precision.  Numerical precision is not merely
a technical curiosity but is essential to identify the regions and
nature of bunching. 


A parallel starting point of our analysis is a characterization of
the equilibrium in a positive economy. In our positive economy, workers
choose the amount of cognitive and manual tasks to deliver, coworkers
to work with, and firms to work for. This problem integrates endogenous labor supply decisions with the assignment
of multidimensional workers to teams and to heterogenous firms. Our first result for the positive
economy is that workers sort with identical coworkers (self-sorting)
and that better teams work on more valuable projects (positive sorting).
The resulting assignment is qualitatively identical to the assignment
under the optimal policy problem but the exact assignment differs
because of differences in labor supply due to incentive constraints.

We use the dual formulation of the equilibrium assignment problem
to characterize equilibrium wages. We show that wages are a convex
function of an index of the worker's task inputs rather than a function of
each task individually. We further establish an exact mapping between
curvature in the wage schedule and the distribution of firm productivity.
By choosing a parametric convex function, we can then infer a distribution
of firm projects such that this convex function is the equilibrium
wage schedule.

Having developed the theory to characterize both positive and optimal
allocations, we bring our theory to the data. In order to quantify
cognitive and manual skill heterogeneity in the U.S. population, we
use earnings information for all workers between 2000 and 2019 in
the American Community Survey (ACS). We combine the earnings data
from the ACS with data from O{*}NET on the manual and cognitive task
intensity for every occupation (\citet{Acemoglu:2011}).

We use our closed-form characterization for wages in the positive
economy to pointwise identify the level of manual and cognitive tasks
completed by each worker. We then use the worker's optimality condition
for each task together with these inferred task levels to identify
levels of cognitive and manual skill that deliver each worker's wage
and relative task intensity in the cross-sectional data as a model
outcome. For the unidimensional taxation problem, an important contribution
of \citet{Saez:2001} was to infer the underlying productivity distribution
using earnings data which then becomes a central input to determine optimal taxes. Our identification generalizes these findings and delivers
the distribution of manual and cognitive skills in a model accounting
for multiple drivers of earnings (multidimensional skills, coworkers,
firms). Our identification resembles \citet{BK2:2020,BK1:2021} who
use explicit solutions for home production models to identify productivity
at home and market productivity using data on consumption, home and market hours.

We next turn to the quantitative characterization of the optimum using
the inferred skill distribution. In order to understand our quantitative
characterization, we first describe a benchmark without incentive
constraints. Due to separability of preferences and technology over
tasks, the efforts in a given task depend exclusively on the worker’s
skills in this task and, hence, there is no cross-dependence between
tasks. Trivially, there is no bunching and there are no distortions.

In sharp contrast to the benchmark, optimal task intensity in our
model depends positively on both of the worker's skills. Workers with
high manual skills also deliver high levels of cognitive tasks.
This codependence is lower at the top end of the skill distribution.
In the limit, workers face zero distortion in their manual task allocation
at the top of the manual skill distribution, meaning there is no dependence
of their manual task intensity on their cognitive skills as in the
benchmark. At the lower end of the skill distribution, the distortion
from this positive codependence is high.


A central part of our contribution is to characterize patterns of
bunching. We first show that 10.4 percent of all workers is
bunched at the optimum. Workers bunch with other workers both near
and afar. Moreover, workers exclusively bunch with workers that are
better in one skill dimension but worse in another. Workers do not
bunch with workers over whom they have an absolute advantage nor with
workers who have an absolute advantage over them.

We introduce two types of bunching: blunt bunching and targeted bunching.
In the blunt bunching region, the planner requires all workers with
the same effective skill index to deliver identical cognitive and
manual tasks, and thus bunches workers that vastly differ in their
skills. This is a blunt way to provide incentives and comes at a cost
of significant output distortions. In the targeted bunching region,
the planner recognizes the increasing efficiency costs of distorting
higher skill workers. When workers have a higher relative level of,
for example, manual skills they are separated along this dimension
but are bunched on their relatively low cognitive skill. The planner
thus separates according to workers' comparative advantage and bunches
workers by comparative disadvantage. In contrast to the blunt bunching
region, targeted bunching occurs with workers who are more similar
in skills: not too far away yet still nonlocally. In the region without
bunching, the planner distorts allocations as in the unidimensional
case.

We summarize the bunching patterns by describing the tax wedges they
induce. In particular, we find that the level of tax wedges is high
in the bunching regions. The tax wedges are particularly high for
low skill workers who are bluntly bunched and are also high along
the dimension of comparative disadvantage for the more skilled workers
in the targeted bunching region. We further show that the optimum is
implementable by a tax function that is only a function of earnings
and line of work.

\vspace{0.4cm}
\noindent \textbf{Literature}. We now describe related literature. \citet{Kleven:2009}
is the first paper that analyzed optimal multidimensional taxation
with bunching. They model a binary labor supply choice for the secondary
earner along with continuous labor supply choice for the primary earner.
\citet{Judd:2017} consider numerically some cases of optimal taxation
with multiple dimensions of heterogeneity (up to five dimensions of
heterogeneity with five individual types) and find that some non-local
constraints bind. The most ambitious attempt to date to solve a multidimensional
policy problem with bunching is \citet{Moser:2019} for a model where
workers are heterogeneous in two dimensions but only one dimension
of heterogeneity enters the planner’s objective. Their key ingredient
is paternalistic preferences, which delivers bunching due to disagreement
between the planner and workers. In their environment bunching is
optimal and, in fact, an essential feature even for the unidimensional
problem. The fact that the planner cares only about one dimension
of heterogeneity significantly reduces the complexity of deviations
patterns. They characterize the model theoretically with the continuous
skill distributions and also compute the model with six impatient
types in one dimension and a large number of types in the second dimension.
In our paper and, more broadly, for multidimensional optimal nonlinear
taxation problems the planner cares about heterogeneity in several
dimensions and, hence, the deviations and bunching patterns are significantly
more complicated and nuanced, especially, when a large number of types
within each skill dimension is analyzed. \citet{Heathcote:2021b}
comprehensively analyze computational performance of different algorithms
for unidimensional optimal taxation. They show that the number of
skill types is not just a technical detail but has an important impact
on policy prescriptions. In our settings, the need for fine skill
differentiation in both dimensions of heterogeneity is additionally
important to recover the nuanced patterns of bunching and deviating,
especially in the regions of targeted bunching. More broadly, there
is a vast literature on multidimensional mechanisms (e.g., \citet{McAfee:1988},
\citet{Armstrong:1996} and \citet{Rochet:1998}) that also emphasizes
the complexity, as well as the central role, of bunching for the optimal
solutions.

An important strand of papers in \citet{Scheuer:2014},
\citet{Rothschild:2013,Rothschild:2014,Rothschild:2016} analyze nonlinear
optimal taxation with multidimensional heterogeneity. These papers
achieve tractability by transforming the multidimensional problem
into a unidimensional screening problem with an endogenous wage distribution.
Moreover, \citet{Rothschild:2014,Rothschild:2016} in the multidimensional
case and \citet{Scheuer:2017} also emphasize the importance of labor
market sorting. \citet{Lehmann:2021} and \citet{Golosov:2022} use
a first-order approach to theoretically and numerically study multidimensional
optimal taxation when there is no bunching.

A complementary approach is to analyze optimal policy in economies
with multidimensional heterogeneity by restricting taxes to parametric
families. The most comprehensive recent analysis using this approach
is \citet{Blundell:2012} on optimal taxation of low-income families
and \citet{Gayle:2019} on optimal taxation of couples. Notable papers
that use such a parametric approach in a variety of other areas of
optimal taxes are, for example, \citet{Benabou:2002}, \citet{Conesa:2009},
\citet{HSV:2017}. \citet{Heathcote:2021a} synthesize the Mirrleesian
approach and the parametric approach to optimal taxation.

Our positive wage determination model relates to a growing literature
in labor economics that adopts a task approach to understand the contribution
of multidimensional skills to labor market outcomes. Recent prominent
examples in this area include \citet{Yamaguchi:2012}, \citet{Sanders:2012},
\citet{Lindenlaub:2017}, \citet{Deming:2017}, \citet{Guvenen:2020},
\citet{Lise:2020}, \citet{Roys:2022} and \citet{Lindenlaub:2020}.
Differently from these papers, we combine multidimensional skill heterogeneity
with sorting into worker teams and sorting with heterogeneous firms.


\section{Environment}

We consider an economy with two-dimensional worker skill heterogeneity and heterogeneous firms. Worker output depends not only on own cognitive and manual efforts but also on the coworker it works with and the firm it works for $-$ as emphasized in the modern literature on wage determination.

\vspace{0.4 cm}
\noindent \textbf{Workers}. The economy is populated by a measure two of workers who differ in two unobservable characteristics. Workers are endowed with a bundle of cognitive and manual skills $\alpha = (\alpha_c , \alpha_m) \in A$. The distribution over types $\alpha$ is denoted by $\Phi$. 


Workers have preferences over consumption $c$ and experience disutility from effort in cognitive and manual activities $\ell = (\ell_c,\ell_m)$:
\begin{equation}
U(c,\ell) = u ( c ) - v(\ell_c) - v(\ell_m) , \label{e:individual_prefs}
\end{equation}
where consumption $c$ and leisure $\ell$ are positive, utility is increasing and concave in consumption, and decreasing and strictly concave in cognitive and manual efforts. We further assume disutility has the form: 
\begin{equation}
v ( \ell ) = \kappa \ell^\rho , \label{e:disutility}
\end{equation}
with $\rho > 2, \kappa >0$.


\vspace{0.4 cm}
\noindent \textbf{Technology}. Cognitive and manual production input $(x_c,x_m) \in \mathbf{X}$ for a worker are the product of their skills and their efforts:
\begin{equation}
x_s = \alpha_s \ell_s , \label{e:worker_tech}
\end{equation}
for all tasks $s \in S = \{ c, m\}$. The worker's skill is given by $\alpha$, while their effort is given by $\ell$. 

The economy is populated by a unit mass of heterogeneous firms that produce a single output by organizing two workers into a team to work on a project $z$. Firm production is represented by $y$. We use a bilinear team technology together with a multiplicative firm technology:\footnote{The bilinear technology is also used in \citet{Lindenlaub:2017}, \citet{Lise:2020}, and \citet{Lindenlaub:2020}, among others.}
\begin{equation}
y (x_1,x_2,z) = z \left( x_{1c} x_{2c} + x_{1m} x_{2m} \right) . \label{e:firm_tech}
\end{equation}


\vspace{0.2 cm}
\noindent \textbf{Assignment}. An assignment pairs workers with coworkers and projects. Formally, an assignment is a probability measure $\gamma$ over workers, coworkers, and firms. Given a distribution of worker inputs $F_x$, a distribution of coworker inputs $F_x$, and a distribution of firms $F_z$, the set of feasible assignments is $\Gamma := \Gamma(F_x,F_x,F_z)$. This is the set of probability measures on the product space $\mathbf{X} \times \mathbf{X} \times Z$ such that the marginal distributions of $\gamma$ onto the set of workers and coworkers $\mathbf{X}$ are $F_x$, and the marginal distribution of $\gamma$ onto the set of firms $Z$ is $F_z$.  The assignment function captures the measure of workers that work together on a project as $\gamma(x_1,x_2,z)$. Given a feasible assignment total output is $\int y(x_1,x_2,z) \text{d} \gamma$.

\vspace{0.4 cm}
\noindent \textbf{Resources}. Aggregate output and external resources $R$ are allocated to workers to consume:
\begin{equation}
\int y(x_1,x_2,z) \text{d} \gamma + R \geq \int c(\alpha) \text{d} \Phi , \label{e:resources_original}
\end{equation} 
where $\int c(\alpha) \text{d} \Phi$ is aggregate consumption.

\section{Planning Problem} \label{s:planner}

In this section, we formulate a planner problem and characterize optimal sorting. The planner problem is to choose an allocation $\{ (c(\alpha),x_c(\alpha),x_m(\alpha)) \}_{\alpha \in A}$ and an assignment $\gamma \in \Gamma$ to minimize the resource cost of providing welfare $\mathcal{U}$:
\begin{equation}
\int c(\alpha) \text{d} \Phi  - \int y (x_{1},x_{2}, z ) \text{d} \gamma, \label{e:resources_original}
\end{equation}
subject to incentive constraints for all workers $\alpha \in A$, so that workers do not gain by misreporting types to be $\hat{\alpha}=(\hat{\alpha}_{c},\hat{\alpha}_{m})$: 
\begin{equation}
U (c(\alpha), x_c(\alpha) / \alpha_c, x_m(\alpha) / \alpha_m ) = \max\limits_{\hat{\alpha} \in A} \; U (c(\hat{\alpha}), x_c(\hat{\alpha}) / \alpha_c, x_m(\hat{\alpha}) / \alpha_m ) , \label{e:incentives_original}
\end{equation}
and the promise keeping condition:
\begin{equation}
\int U (c(\alpha), x_c(\alpha) / \alpha_c, x_m(\alpha) / \alpha_m ) \text{d} \Phi \geq \mathcal{U} \label{e:promise_keeping},
\end{equation}
which requires that aggregate welfare exceeds promised value $\mathcal{U}$.\footnote{The planning problem is equivalent to maximizing utilitarian welfare function subject to the resource constraint (\ref{e:resources_original}) and the incentive constraints (\ref{e:incentives_original}). There are no incentive constraints for firms since we assume firm output $y$ and inputs $x_1$ and $x_2$ are observed by the planner. Hence, the firm productivity $z$ is not private information.} 


\subsection{Assignment} \label{s:planner_assignment}

The planning problem contains an assignment problem. The planner pairs worker and coworker inputs with firm projects to maximize output given a distribution of worker inputs and firm projects. We show the planner optimally chooses a self-sorted assignment, meaning that workers are paired with identical coworkers, and also show that the planner pairs better teams with more valuable projects.\footnote{This assignment problem falls into a class of multimarginal, multidimensional optimal transportation problems. Multidimensional skill and the dependence of worker output on coworkers are central to recent advances in sorting models that utilize optimal transport theory to characterize
equilibrium \citep{Chiappori:2010,Dupuy:2014,Lindenlaub:2017,Chiappori:2017,Galichon:2021b}.}

The assignment problem embedded in the planning problem is to choose an assignment to maximize production given the distribution of workers tasks $F_x$ and the project distribution $F_z$:
\begin{equation}
\max\limits_{\gamma \in \Gamma}\; \int y (x_1, x_2, z ) \text{d} \gamma. \label{e:assignment}
\end{equation}
\noindent We construct an assignment $\gamma$ that self-sorts workers and coworkers to obtain a unidimensional distribution for team quality, or effective worker skill, $X = x_c^2 + x_m^2$. The assignment $\gamma$ combines self-sorting of workers with positive sorting between the effective worker skill $X$ and projects $z$.\footnote{Self-sorting is defined with respect to distribution of effective task inputs that the workers supply, not necessarily with respect to the underlying worker skills $\alpha$. In the presence of bunching, multiple workers $\alpha$ supply the same task levels $x$ and hence self-sorting of effective tasks may imply matching different $\alpha$.} This assignment $\gamma$ solves the assignment problem (\ref{e:assignment}).

\begin{proposition}{\textit{Optimal Sorting}.} \label{prop:p_assignment}
The planner assignment $\gamma_*$ is characterized by self-sorting of workers and by positive sorting between team quality and project values. 
\end{proposition}
\noindent The proof is in Appendix \ref{pf:p_assignment}.

We now develop the intuition for Proposition \ref{prop:p_assignment}. Given a firm project, and since the technology for each task in equation (\ref{e:firm_tech}) is supermodular, the planner optimally wants to positively sort both cognitive and manual inputs to project $z$. In our economy with multidimensional worker skills, positive sorting within each task is attained by self-sorting. An optimal assignment thus features self-sorting of workers with coworkers within projects $z$. Given that workers are optimally self-sorted, the planner remains to sort self-sorted workers with effective skill $X$ to firms $z$. Since the effective production technology is supermodular in team quality $X$ and project value $z$, the optimal assignment features positive sorting between teams and project values.\footnote{Positive sorting of effective skill $X$ with project values  $z$ follows the classical Beckerian analysis \citep{Becker:1973}.}

Given that the assignment features self-sorting, the output per worker produced by a team of two workers supplying task inputs $(x_c,x_m)$ is $\frac{1}{2} z ( x^2_{c} + x^2_{m})$. Aggregate output is $\int y (x_{1},x_{2}, z ) \text{d} \gamma = \frac{1}{2} \int z ( x^2_{c} + x^2_{m}) \text{d} \Phi$, and the resource cost (\ref{e:resources_original}) can be written as:
\begin{equation}
\int \Big( c(\alpha) - \frac{1}{2} z (\alpha) \big( x^2_{c}(\alpha) + x^2_{m}(\alpha) \big) \Big) \text{d} \Phi.\label{e:resources_original2}
\end{equation}

\subsection{Utility Allocations} \label{ss:change_of_var}

We next transform the planner problem from choosing consumption and task allocations to choosing consumption utility and labor disutility allocations. 

For each task  $s \in S$, we define the skill parameter $p_s = \kappa \alpha_s^{-\rho}$ so the skill parameter is inversely related to the underlying skill $\alpha_s$. The implied distribution function for the skill parameter vector $p$ is denoted $\pi$, and the corresponding assignment is denoted $z(p)$. We use this skill transformation to define a worker's utility from consumption as a function of their skill vector as $c(p) := u(c(\alpha))$. Following this definition, the resource cost of consumption utility is $\mathcal{C}(c(p)) = u^{-1}(c(p))$. Since the utility from consumption is strictly increasing and concave in the consumption allocation, the resource cost is strictly increasing and convex in consumption utility. Similarly, we define labor disutility in each task $s \in \mathcal{S}$ as a function of the transformed skill parameter $p$ as $x_s(p) := x_s(\alpha)^\rho$. The resource cost of providing disutility $\mathcal{X}(x_s(p)) := - \frac{1}{2} x_s(p)^{\frac{2}{\rho}}$ is decreasing and strictly convex in labor disutility for $\rho > 2$. 


Given the introduction of the skill parameter $p$ and the transformation of the choice variables from allocations to utilities, the planner chooses $\{ ( c(p),x_c(p),x_m(p) ) \}$ to minimize the resource cost of providing welfare: 
\begin{equation}
\int \left( \mathcal{C} ( c (p) ) + z(p) \big( \mathcal{X} ( x_c (p) ) + \mathcal{X} ( x_m (p) ) \big) \right) \pi (p) \text{d} p ,  \label{e:resources_original3}
\end{equation}
subject to linear incentive constraints:
\begin{equation}
c (p) - p_c x_c(p) - p_m x_m(p) \geq c (q) - p_c x_c(q) - p_m x_m(q), \label{e:linear_ic}
\end{equation}
for all workers $(p,q)$, and a linear promise keeping condition:
\begin{equation}
\int ( c (p) - p_c x_c(p) - p_m x_m(p) ) \pi (p) \text{d} p \geq \mathcal{U} \label{e:promise_keeping_linear}.
\end{equation}


\subsection{Incentive Compatibility}


We show that the indirect utility for workers is convex and decreasing in type $p = (p_c,p_m)$. The indirect utility is defined as:
\begin{equation}
u(p) =  c (p) - p_c x_c(p) - p_m x_m(p) \label{e:indirect_utility}, 
\end{equation}
which implies that for any incentive compatible allocation $\nabla u(p) = - x(p) = - ( x_c(p), x_m (p))^T$ almost everywhere. Using the indirect utility function, the incentive constraints (\ref{e:linear_ic}) are $u(p) \geq u(q) - (p_c - q_c) x_c(q) - (p_m - q_m) x_m(q)$ or, equivalently in notation of scalar product $\langle \cdot, \cdot \rangle$:
\begin{equation}
u(p) - u(q) \geq \langle p - q, - x (q) \rangle =  \langle p - q, \nabla u(q) \rangle, \label{e:linear_ic2}
\end{equation}
for the incentive constraint where worker type $p$ does not want to report to be of type $q$.

A differentiable function $u$ on a convex domain is convex if and only if $u(p) \geq u(q) + \langle p - q, \nabla u(q) \rangle$. This implies that an incentive compatible indirect utility function is necessarily convex. Since the gradient of the indirect utility function is the negative of a worker's production disutility, and production disutility is positive, the indirect utility function  decreases in $p$, or $\nabla u(p) = - x(p) \leq 0$.\footnote{The indirect utility function thus increases in skill $\alpha$. In Appendix \ref{a:ic}, we discuss differentiability of the indirect utility function in more detail, and we also establish which incentive compatibility constraints are redundant.} 

\vspace{0.1 cm}
\begin{lemma}\label{lemma:p_convexity}
Any indirect utility function (\ref{e:indirect_utility}) that is incentive compatible is convex and decreasing in worker type $p$.
\end{lemma}

\noindent We denote the set of utility allocations that satisfy the set of incentive constraints by $\mathcal{I}$, which we refer to as feasible allocations.


\subsection{Bunching} \label{s:bunching_theory}

We refer to bunching as different workers being assigned the same labor allocation $x$, and therefore the same consumption allocation $c$. We label the set of bunching points by $\mathcal{B}$.\footnote{Alternatively, one could define a worker $p$ being bunched when there exists another worker $p'$ in its neighborhood such that $x(p) = x(p')$. Our definition of bunching is the closure of this set. While these definitions are economically equivalent, our definition facilitates the presentation of Proposition \ref{p:bunch}.}

\begin{definition*}
Worker $p$ is bunched, $p \in \mathcal{B}$, if and only if in any neighborhood around this worker there exists two other workers $p'$ and $p''$ with identical allocations $x(p') = x(p'')$.
\end{definition*}

We now state the notions of convexity and strong convexity. Assume that the indirect utility is twice continuously differentiable in the neighborhood of a type $p$. An indirect utility function $u$ is convex if and only if the Hessian matrix $H(u)$ is positive semidefinite. The indirect utility function is strongly convex if $H(u) - \alpha_I I$ is positive semidefinite for some strictly positive $\alpha_I$, where $I$ is the identity matrix.

\begin{lemma}\label{l:strconvex}
If the indirect utility (\ref{e:indirect_utility}) is strongly convex, then there is no bunching. If the indirect utility is not strongly convex at all points in the neighborhood of type $p$, then worker $p$ is bunched.
\end{lemma}

\noindent The proof is in Appendix \ref{pf:strconvex}. 


\subsection{Taxation} 

In order to describe optimal distortions, we define tax wedges for each task. The tax wedge captures the difference between a worker's marginal rate of substitution between task $x_s$ and consumption $c$, $v'\big( \frac{x_s(\alpha)}{\alpha_s} \big) \frac{1}{\alpha_s} \big/ u'(c(\alpha))$, and the marginal rate of transformation, $z x_s (\alpha)$. We define the tax wedge as:
\begin{equation}
1 - \tau_s := \frac{v'\big( \frac{x_s(\alpha)}{\alpha_s} \big) \frac{1}{\alpha_s}}{u'(c(\alpha))} \bigg/ \big( z x_s (\alpha) \big) = - \frac{p_s}{ z} \frac{\mathcal{C}'(c(p))}{\mathcal{X}'(x_s(p))} \label{e:optimal_wedges} ,
\end{equation}
where it follows from the inverse function theorem that $\mathcal{C}'(c(p)) = 1 / u'(c(\alpha))$. A positive wedge plays a role of an implicit tax on marginal income on task $s$.\footnote{Using the definition for the tax wedge, we write $\frac{ \tau_s}{1 - \tau_s} = - \frac{z \mathcal{X}'(x_s(p)) + p_s \mathcal{C}'(c(p))}{p_s \mathcal{C}'(c(p))}$. \label{eq:footnotewedge}} 

\section{Characterization} \label{s:characterization}

We next derive an optimality condition for the multidimensional tax problem that incorporates bunching. 

\subsection{Implementability Condition}

The planner chooses consumption utility and labor disutility $(c,x)$ to minimize the Lagrangian:
\begin{equation}
\mathcal{L}( c,x ) = \int \Big(  \mathcal{C} ( c ) + z \big( \mathcal{X} ( x_c ) + \mathcal{X} ( x_m ) \big) - \lambda \big( c - p_c x_c -p_m x_m - \mathcal{U} \big) \Big) \pi  \text{d}p . \label{e:lagrange_basic}
\end{equation}
subject to the incentive constraints (\ref{e:linear_ic}), where $\lambda$ denotes the multiplier on the promise keeping constraint (\ref{e:promise_keeping_linear}).\footnote{We suppress the dependence on $p$ to streamline notation.} 


\begin{proposition}{\textit{Implementability Condition}}. \label{p:implementable}
Let $( c,x )$ denote a solution to the planner problem, then the implementability condition:
\begin{equation}
\int \big( \mathcal{C}'( c ) \hat{c} + z \big( \mathcal{X}' ( x_c ) \hat{x}_c + \mathcal{X}' ( x_m ) \hat{x}_m \big) \big) \pi \text{d}p  \geq \lambda \int \big( \hat{c} - p_c \hat{x}_c - p_m \hat{x}_m \big) \pi \text{d}p \label{e:stoch_dominance_weak_basic1}
\end{equation}
holds for any feasible allocation $(\hat{c},\hat{x}) \in \mathcal{I}$. At a solution $(\hat{c},\hat{x}) = (c,x)$, (\ref{e:stoch_dominance_weak_basic1}) holds with equality.
\end{proposition}

\noindent The proof is in Appendix \ref {pf:implementable}. Proposition \ref{p:implementable} states that for any feasible allocation $(\hat{c},\hat{x})$, the implementability condition is necessarily satisfied, where the marginal resource costs of providing consumption utility $\mathcal{C}'( c )$, as well as the marginal resource costs of providing disutility from work $(\mathcal{X}' ( x_c ),\mathcal{X}' ( x_m ))$, are evaluated at an optimum. Thus, the implementability condition places restrictions on the optimal $( c,x )$ that need to satisfy (\ref{e:stoch_dominance_weak_basic1}) for any feasible allocation $(\hat{c},\hat{x})$.

Proposition \ref{p:implementable} combines two variational arguments. First, consider a small proportional change in consumption utility and labor disutility. This variation is feasible. Since this scaling is unrestricted, meaning that it can either increase or decrease the utility allocations, it implies that (\ref{e:stoch_dominance_weak_basic1}) holds with equality at the optimal allocation $(c,x)$. Second, consider a convex combination of an optimal allocation and any other feasible allocation with a small weight. The convex combination is equivalent to scaling down the optimal allocation and adding a small positive perturbation. By the previous argument, rescaling does not change the Lagrangian at the optimum allocation. The positive perturbation should not decrease the Lagrangian. Since this perturbation is positive it gives an inequality condition. 

Proposition \ref{p:implementable} presents an implementability constraint for an incentive constrained economy. The implementability conditions are more common in the Ramsey taxation literature where they represent the distortions to allocations introduced by pre-specified taxes. In our model, we do not impose direct restrictions on the permissible taxes and, instead, an information friction endogenously restricts the set of allocations. Importantly, our implementability condition holds with inequality which, as we show, is essential for characterizing the bunching regions. 

\subsection{Optimal Tax Condition as Stochastic Dominance}

We use Proposition \ref{p:implementable} to derive an optimality condition for our multidimensional taxation problem in terms of a stochastic dominance condition. 

We first use the indirect utility (\ref{e:indirect_utility}) for a feasible allocation $(\hat{c},\hat{x})$ to write the implementability condition (\ref{e:stoch_dominance_weak_basic1}) as:  
\begin{equation}
\int \big( \mathcal{C}'( c ) \big( \hat{u} - \nabla \hat{u} \cdot p \big) - z \mathcal{X}' ( x ) \cdot \nabla \hat{u} \big) \pi  \text{d}p  - \lambda \int \hat{u} \pi \text{d}p \geq 0 \label{e:stoch_dominance_weak_basic},
\end{equation}
for any nonnegative, decreasing and convex indirect utility function $\hat{u}$. By Proposition \ref{p:implementable} it follows that (\ref{e:stoch_dominance_weak_basic}) holds with equality for an optimal indirect utility function. Integrating implementability condition (\ref{e:stoch_dominance_weak_basic}) by parts we obtain:
\begin{equation}
\int \big( \partial_{p_c} \big( \pi (p_c \mathcal{C}'( c )+ z  \mathcal{X}' ( x_c )) \big) + \partial_{p_m} \hspace{-0.07 cm} \left( \pi (p_m \mathcal{C}'( c )+ z  \mathcal{X}' ( x_m )) \right) \hspace{-0.10 cm} \big) \hat{u} \text{d}p \geq \int \pi (\lambda - C'(c)) \hat{u} \text{d}p + \Xi(\hat{u}) \label{e:generalized_el}, 
\end{equation}
for any nonnegative, decreasing and convex indirect utility function $\hat{u}$, where boundary conditions act on $\hat{u}$ as $\Xi(\hat{u})= \int^{\bar{p}_m}_{\underline{p}_m} \pi (p_c \mathcal{C}'( c )+ z  \mathcal{X}' ( x_c )) \hat{u}\big\vert^{\bar{p}_c}_{\underline{p}_c} \text{d} p_m + \int^{\bar{p}_c}_{\underline{p}_c} \pi (p_m \mathcal{C}'( c )+ z  \mathcal{X}' ( x_m )) \hat{u} \big\vert^{\bar{p}_m}_{\underline{p}_m} \text{d} p_c$. 

We now define second-order stochastic dominance \citep{Shaked:2007}:
\begin{definition*}
The measure $\mu$ \underline{second-order stochastically dominates} the measure $\nu$, or $\mu \succeq \nu$, if and only if for any nonnegative, decreasing and convex function $\hat{u}$: 
\begin{equation}
\int \hat{u}(p) \text{d} \mu \geq \int \hat{u}(p) \text{d} \nu . \label{e:stoch_dom}
\end{equation}
\end{definition*}
\noindent Second-order stochastic dominance states that equation (\ref{e:stoch_dom}) holds for any nonnegative, decreasing, and convex function $\hat{u}$. These conditions exactly correspond to the indirect utility being feasible (\Cref{lemma:p_convexity}). Applying the definition for second-order stochastic dominance to equation (\ref{e:generalized_el}) we obtain the following theorem.

\begin{theorem*}{\textit{Optimal Tax Condition as Stochastic Dominance}}. \label{s:p_stochastic}
Suppose that the optimal allocation $(c,x)$, density function, and assignment are all continuously differentiable. Then,
\begin{equation}
\partial_{p_c} \big( \pi (p_c \mathcal{C}'( c )+ z  \mathcal{X}' ( x_c )) \big) + \partial_{p_m} \hspace{-0.07 cm} \left( \pi (p_m \mathcal{C}'( c )+ z  \mathcal{X}' ( x_m )) \right) \succeq \pi (\lambda - \mathcal{C}'(c)) + \Xi . \label{e:general_abc}
\end{equation}
\end{theorem*}

\noindent This theorem derives the optimality condition for the multidimensional taxation economy that incorporates bunching. This condition shows that, at the optimum, the measure over marginal tax revenues, $\pi \left( 1 / u'(\mathcal{C}(c)) - \lambda \right)$, 
second-order stochastically dominates the measure over marginal tax distortions, 
\begin{equation}
\partial_{p_c} \Big( \frac{\pi}{u'(\mathcal{C}(c))} p_c \frac{\tau_c}{1 - \tau_c}  \Big) + \partial_{p_m} \Big(  \frac{\pi}{u'(\mathcal{C}(c))} p_m \frac{\tau_m}{1 - \tau_m} \Big) + \Xi,
\end{equation} 
where we use the definition of the labor skill wedge (\ref{e:optimal_wedges}) and footnote \ref{eq:footnotewedge}.


Comparing the costs and the benefits of taxes is the key insight of the classic ABC formula and the analysis of \citet{Diamond:1998} and \citet{Saez:2001}. In the classic unidimensional case, these costs and benefits are exactly equated for each of the skill levels. Our theorem shows that for the multidimensional tax case with bunching the logic of the ABC formula applies as the costs and the benefits of the taxes are compared. However, those are not necessarily equated at each skill level. Instead, our optimal tax condition (\ref{e:stoch_dom}) considers the benefits and the costs of the entire schedule of taxes and states that the entire schedule of benefits of taxes should second-order stochastically dominate the entire schedule of distortions $-$ showing the non-local nature of the problem with multidimensional skills in the regions with bunching. Our formula applies both to the regions with and without bunching and, in the latter case reduces to equating the costs and the benefits of distortions at each skill level $-$ thus making it a local problem for the regions without bunching.\footnote{In Appendix \ref{a:ssd}, we develop the connection between our general optimal tax conditions and the classic ABC formula. Our optimal taxation condition as stochastic dominance is also related to the sweeping operator in \citet{Rochet:1998}. More specifically, the existence of a version of the sweeping operator can be established by using a variation of the Strassen Theorem (see \citet{Shaked:2007}, Theorem 4.A.5). The optimal taxation formula goes beyond existence of such an operator by further relating the entire schedule of costs to the entire schedule of benefits of optimal taxes.}  


\subsection{Global Optimal Tax Formula}

We next provide an optimal tax formula for the multidimensional taxation problem as an equality. This representation also connects to the optimal tax formulas in unidimensional taxation problems which are derived as equality measuring the marginal costs and benefits of taxation \citep{Mirrlees:1971,Diamond:1998,Saez:2001}. The main difference with these results is that the optimal tax formula in our multidimensional taxation problem explicitly accounts for global incentive constraints.  

\begin{theorem*}{\textit{Global Optimal Tax Formula}}. \label{s:p_equality}
Suppose the optimal allocation $(c,x)$, density function, and assignment are all continuously differentiable. Then,
\begin{equation}
\partial_{p_c} \Big( \frac{\pi}{u'(\mathcal{C}(c))} p_c \frac{\tau_c}{1 - \tau_c}  \Big) + \partial_{p_m} \Big(  \frac{\pi}{u'(\mathcal{C}(c))} p_m \frac{\tau_m}{1 - \tau_m} \Big) = \pi \left( \frac{1}{u'(\mathcal{C}(c))} - \lambda \right) - \Delta M(p), \label{e:general_abceq}
\end{equation}
where $M(p)$ is a positive semidefinite matrix that enforces the convexity of the indirect utility function, and $\Delta M(p) = \sum\limits_{i,j} \frac{\partial^2}{\partial p_{i}\partial p_{j}}M_{ij}(p)$.
\end{theorem*}

\vspace{0.3 cm}
\noindent The full proof is in Appendix \ref{a:globaloptimaltax} and here we provide a sketch of the proof. First, we use the definition of the indirect utility function (\ref{e:indirect_utility}) to reformulate the planner problem as directly choosing an indirect utility function to minimize the resource cost of providing welfare. For the indirect utility function to be globally incentive compatible, the reformulated planning problem is constrained by the condition that the indirect utility function is convex and decreasing in worker type $p$ following the characterization in Lemma \ref{lemma:p_convexity}. 

Second, the indirect utility function $u(p)$ being convex is equivalent to its Hessian being positive semidefinite, $H(u) \succeq 0$ for all worker types $p$, which in turn is equivalent to:
\begin{equation}
v^T H(u) v \geq 0 , 
\end{equation}
for all vectors $v \in \mathbb{R}^2$. These inequalities are an infinite series of constraints parameterized by the vectors $v$ for each worker type $p$. For each of these constraints, we introduce a multiplier $\lambda(v, p) \geq 0$ and include these constraints into the Lagrangian for the planning problem. 

Third, we establish (Lemma \ref{l:convexity}) that one can represent the constraint that the indirect utility function has to be convex as a matrix condition by introducing a positive semidefinite Kuhn-Tucker matrix $M(p)$ for each worker $p \in P$. Instead of considering the infinite series of constraints for each worker $p$, a single positive semidefinite matrix $M(p)$ induces convexity of the indirect utility function. Upon integration by parts, this restriction appears as a modified social welfare weight $\omega(p) = 1+\frac{\Delta M(p)}{\lambda\pi}$. In summary, the main contribution of the convexity constraint to the planning problem is modifying the social welfare weights through a convexity correction. Finally, we derive in Appendix \ref{a:ocp2} and Appendix \ref{a:ocp3} how this modified social welfare weight translates into the optimal tax condition (\ref{e:general_abceq}).

\vspace{0.4 cm}
\noindent We now discuss in more detail how the optimal tax condition applies in regions without bunching. Specifically, we consider the domain where the indirect utility function is strongly convex and, therefore, there is no bunching.  

The main difficulty in analyzing bunching in the multidimensional case is that the possible indirect utility perturbations $\hat{u}$ are required to be convex. The convexity of perturbations thus acts as an additional constraint on the entire tax schedule. Without bunching, the perturbation argument is straightforward to construct and leads to equating of cost and benefits of taxes at each skill level. Intuitively, if the underlying utility function is strongly convex, a small enough additive perturbation preserves convexity. As a result, the optimal tax condition (\ref{e:general_abceq}) in Theorem \ref{s:p_equality} applies with the convexity correction $\Delta M=0$ at the types where there is no bunching.

\begin{corollary}{\textit{Multidimensional Optimal Tax Formula without Bunching}.}\label{p:el}
If the indirect utility function is strongly convex for a worker $p$, then:
\begin{equation}
\pi \left( \frac{1}{u'(\mathcal{C}(c))} - \lambda \right) = \partial_{p_c} \left( \frac{\pi}{u'(\mathcal{C}(c))} p_c \frac{\tau_c}{1 - \tau_c}  \right) + \partial_{p_m} \left(  \frac{\pi}{u'(\mathcal{C}(c))} p_m \frac{\tau_m}{1 - \tau_m} \right) .\label{e:euler}
\end{equation}
\end{corollary}

\vspace{0.4 cm}
\noindent The proof is in Appendix \ref{a:el}. In order to provide intuition for Corollary \ref{p:el}, and to connect our expression to the existing literature, we also write this condition in the original worker type coordinates $\alpha$:
\begin{equation}
\phi(\alpha) \left( \lambda - \frac{1}{u'(c(\alpha))} \right) = \frac{1}{\rho} \partial_{\alpha_c}\left( \frac{\phi(\alpha)}{u'(c(\alpha))} \alpha_c \frac{\tau_c}{1 - \tau_c}  \right) + \frac{1}{\rho} \partial_{\alpha_m} \left( \frac{\phi(\alpha)}{u'(c(\alpha))} \alpha_m \frac{\tau_m}{1 - \tau_m} \right) \label{e:eulera},
\end{equation}
which is the same form as derived in \citet[p. 23]{Kleven:2006}, \citet{Lehmann:2021}, and \citet{Golosov:2022}. The left-hand side captures the marginal benefit of increasing taxes, lowering the resource cost by taxing worker $\alpha$ at the cost $\lambda$ of tightening the promise keeping condition, and where $\phi$ denotes the density function over the types $\alpha$. At an optimum, the marginal benefit of increasing taxes is equated to the marginal distortionary cost of increasing taxes, which is given by the right-hand side. The right-hand side captures the change in labor distortions inversely weighted by the marginal utility of consumption. Distortionary costs of taxation scale with the elasticity of labor supply, which is governed by $\rho$. When the supply of tasks is elastic (low $\rho$), marginal distortionary costs are large. When the supply of tasks is inelastic (high $\rho$), marginal distortionary costs are small. All else equal, if the marginal utility from consumption is low, $\lambda < 1 / u'(c(\alpha))$, for high-skill workers, the labor skill distortion decreases with an increase in either cognitive or manual skills. When more workers are affected by a change in the skill distortions, or when the promise keeping constraint is tight, marginal labor distortions change more rapidly.

Finally, we provide a converse to Corollary \ref{p:el} that allows to determine the regions of bunching.

\begin{proposition}{\textit{Identifying Bunching}.} \label{p:bunch}
If equation (\ref{e:euler}) does not hold for a worker type $p$, then this worker is bunched.
\end{proposition}

\noindent Proposition \ref{p:bunch} thus provides a test to identify bunching. Whenever equation (\ref{e:euler}) is violated, the worker is bunched. We prove Proposition \ref{p:bunch} in Appendix \ref{pf:bunch}. By the contrapositive to Corollary \ref{p:el} it follows that when equation (\ref{e:euler}) does not hold, the indirect utility function is not strongly convex, meaning that the Hessian matrix is degenerate for worker $p$. We show that the Hessian matrix is also degenerate for all workers within the neighborhood of $p$, which we show is equivalent to worker $p$ being bunched.

\subsection{Legendre Linearization} \label{s:legendre}

In this section, we discuss the main technique that enables the numerical solution of our problem. Specifically, we transform our problem into a linear problem using Legendre transformations for convex functions that translates convex functions into the upper envelopes of their tangent lines. In order to explain the Legendre transform, and show its importance, we use the resource cost of providing consumption utility $\mathcal{C}$ as an example. 


A convex function exceeds all tangent lines. For any consumption utility $c$, and for any point of tangency $a$: 
\begin{equation}
\mathcal{C}(c) \geq \mathcal{C}(a) + (c - a) \mathcal{C}'(a) = \varphi c - \mathcal{C}^*(\varphi),  \label{e:c_constraint_inequality} 
\end{equation}
where the equality follows by parameterizing the tangent lines with their slope $\varphi := \mathcal{C}'(a)$ and by letting $\mathcal{C}^*(\varphi) := - \mathcal{C}(a) + a \mathcal{C}'(a)$ for $a = {\mathcal{C}'}^{-1} (\varphi)$. The function $\mathcal{C}^*$ is the Legendre transform for the resource cost of providing consumption utility $\mathcal{C}$. Since a convex function exceeds all its tangent lines, and since the function value equals the value of the tangent line at the point of tangency:
\begin{equation}
\mathcal{C}(c) = \max_{\varphi \geq 0} \; \varphi c - \mathcal{C}^*(\varphi). \label{e:lt_c} 
\end{equation}
The Legendre transformation converts the convex resource cost of providing consumption utility on the left side of (\ref{e:lt_c}) into a family of linear constraints on the right. The family of linear constraints is parameterized by the slopes of the tangent lines of the cost function. Since the resource cost increases with consumption utility, the slopes of the tangent lines are positive, or $\varphi \geq 0$. 




The previous steps apply for any convex function, allowing us to use the same argument to transform the resource cost of providing work disutility into a family of linear constraints:
\begin{equation}
\mathcal{X}(x_s) = \max_{\psi_s \leq 0} \; \psi_s x_s - \mathcal{X}^*(\psi_s) \label{e:lt_x},
\end{equation}
for each skill $s \in \mathcal{S}$. An increase in production disutility increases production and therefore lowers resource costs. The resource cost of production disutility is decreasing, implying negative slopes of the tangent lines, or $\psi_s \leq 0$.


To summarize, the transformed planning problem is to minimize the resource cost of providing utilitarian welfare $\mathcal{U}$: 
\begin{equation}
\int \Big( \max_{\varphi(p) \geq 0} \big( \varphi(p) c(p) - \mathcal{C}^*(\varphi(p)) \big) + z(p) \sum_{s} \max_{\psi_s(p) \leq 0} \big( \psi_s(p) x_s(p) - \mathcal{X}^*(\psi_s(p)) \big) \Big) \pi (p) \text{d} p \label{e:resources_original4}
\end{equation}
subject to incentive constraints (\ref{e:linear_ic}) for all workers $(p,q) \in P \times P$, and the linear promise keeping condition (\ref{e:promise_keeping_linear}).\footnote{In \Cref{pf:planner_duality}, we show this problem is equivalent to maximizing utilitarian welfare subject to the resource constraint, and the incentive constraints. In \Cref{s:transformed_planner} we establish how to derive the stochastic dominance condition and the general optimal tax formula directly from the transformed problem.}

\vspace{0.4 cm}
\noindent \textbf{Numerical Approach}. The main insight of this analysis is that Legendre transform enables us to translate the planning problem into a linear problem (see Appendix \ref{a:numerical_approach} for more detail). This is the reason why we are able to solve the model for a total of 40 thousand worker types, with 200 types in both the cognitive and the manual dimension, and a total of 1.6 billion incentive constraints. Importantly, a large number of types and numerical precision is not merely a technical and computational curiosity, it is essential to characterize the regions and nature of bunching. In addition, we use two other significant steps to reduce the number of effective incentive constraints.

First, we consider only a small set of incentive constraints by adding incentive constraints between two worker types only if the distance between them is small.\footnote{\citet{Oberman:2013} shows that the solution to the problem with only local constraints provides a reasonable initial guess.} We then use an iterative procedure to update the set of incentive constraints. On each step, we add all violated incentive constraints to the problem.\footnote{After the final step, the candidate solution satisfies all constraints to the strictly convex optimization problem and hence is the unique solution. In practice, we always obtain the same solution for different initial conditions.} With 40 thousand types, this procedure allows to reduce the number of incentive constraints to about 4 million constraints instead of 1.6 billion. Second, an important step that helps us reduce the number of incentive constraints is that we do not need to consider reducible incentive constraints (see \Cref{a:ic}). This observation additionally reduces the number of constraints by a factor of two. In Appendix \ref{a:numerical_approach} we further prove the accuracy of the approximate planner problem and describe the algorithm that we use to characterize the numerical solution. We finally note that without introducing Legendre transforms the objective is nonlinear. Currently, even the state-of-the-art nonlinear solvers cannot handle the characterization of the solution even for small numbers of types.

\section{Positive Economy} \label{s:positive}

We describe and characterize an equilibrium in a positive model of workers with multidimensional skills sorting with heterogeneous firms.

\vspace{0.4 cm}
\noindent Every firm $z$ takes wage schedule $w$ as given and chooses two workers to solve:
\begin{equation}
\Omega(z) = \max_{x_{1},x_{2}} \; y (x_1,x_2,z) - w(x_1) - w(x_2) . \label{e:firm_problem}
\end{equation}
We define the surplus $S$ as output minus payments to the workers and the firm: 
\begin{equation}
S(x_1,x_2,z) = y (x_1,x_2,z) - w(x_1) - w(x_2) - \Omega(z) .
\end{equation}
Firm output cannot exceed payments to its workers and owner, that is, $S(x_1,x_2,z) \leq 0$ for any triplet $(x_1,x_2,z)$.

Every worker takes the wage schedule $w$ as given and chooses their cognitive and manual task inputs $x$ to solve: 
\begin{equation}
\max_{x_c,x_m} \; u(c) - v \Big( \frac{x_c}{\alpha_c} \Big) - v \Big( \frac{x_m}{\alpha_m} \Big) \label{e:worker_problem}
\end{equation}
subject to the budget constraint $c = ( 1 - \tau ) w(x)$, where $w(x) = w(x_c,x_m)$ is the wage as a function of cognitive and manual inputs, and the disutility from work is given by (\ref{e:disutility}). The government taxes earnings at a rate $\tau$ to finance public expenditures $G$ that are not valued by workers.

The resource constraint is given by:
\begin{equation}
\int y(x_1,x_2,z) \text{d} \gamma(x_1,x_2,z)  = \int c(\alpha) \text{d} \Phi(\alpha) + \int \Omega(z) \text{d} F_z(z) + G . \label{e:resource_constraint}
\end{equation}
Total production, $\int y(x_1,x_2,z) \text{d} \gamma(x_1,x_2,z)$, equals output distributed to workers, $\int c(\alpha) \text{d} \Phi(\alpha)$, to firms $\int \Omega(z) \text{d} F_z(z)$, and to public expenditures $G$.

\vspace{0.4 cm}
\noindent \textbf{Equilibrium}. Given fiscal policy $(\tau,G)$, an equilibrium is a firm value function $\Omega$, a wage schedule $w$, a worker input distribution $F_x$, a feasible assignment $\gamma$, and an allocation $\{ (c(\alpha),x_c(\alpha),x_m(\alpha) )\}$ such that firms solve their profit maximization problem (\ref{e:firm_problem}), workers solve the worker's problem (\ref{e:worker_problem}), the government budget constraint is satisfied $G = \tau \int w(x) \text{d} \Phi(\alpha)$, and the resource constraint (\ref{e:resource_constraint}) is satisfied.

\vspace{0.4 cm}
\noindent The equilibrium assignment is the assignment that maximizes aggregate output, that is, solves the primal problem, while the equilibrium wages $w$ and firm value function $\Omega$ solve the corresponding dual problem. The characterization of the equilibrium assignment, wage schedule, and firm value function through primal and dual problems is  discussed for completeness in Appendix \ref{p:equilibrium_transport}.

\subsection{Characterizing Equilibrium} 

We note that solving for the equilibrium assignment in the positive economy follows the same steps as solving for the planner assignment (\ref{e:assignment}) in Section \ref{s:planner_assignment}. It follows from Proposition \ref{prop:p_assignment} that the equilibrium features self-sorting between workers and coworkers, and positive sorting between team quality and firm project values. 

In order to characterize wages and firm values, we solve the dual transport problem. Since the surplus is negative for any triplet in equilibrium, $S(x_1,x_2,z) \leq 0$, and since the aggregate resource constraint (\ref{e:resource_constraint}), the government budget constraint and the household budget constraints hold in equilibrium, the surplus equals zero almost everywhere with respect to the equilibrium assignment, so $w(x_1) + w(x_2) + \Omega(z) = y(x_1,x_2,z)$. Output is distributed to the owner and to the workers. We use this condition to establish further properties of the firm value function and the wage schedule in \Cref{a:wage_effective}.

In \Cref{a:wage_effective}, we first note that wages are only a function of effective worker skills $X = x^2_c + x^2_m$, and we define $h(X)$, the firm's total wage bill, as $h(X) =2w(x)$. By applying standard arguments from optimal transport, wages are convex in effective worker skill $X$. In other words, small differences in effective worker skill translate into increasingly large differences in earnings.\footnote{The hedonic pricing condition $z = h'(X)$ delivers superstar effects in our model as well as a number of other assignment models (see, for example, \citet{Rosen:1981}, \citet{Gabaix:2008}, \citet{Tervio:2008}, \citet{Scheuer:2017}, and \citet{BTZ:2021}).} Moreover, the firm value function is the Legendre transform of the wage bill, $\Omega = h^*$. As a result, $h(X) + h^*(z) = z X$. 

\vspace{0.4 cm}
\noindent In our quantitative analysis, we infer the distribution of project values $F_z$ using earnings data. The key is to show that there exists a firm project $z$ such that $h(X) + h^*(z) = z X$ for any pairing $(z,X)$. When the wage bill $h$ is continuously differentiable $h(X) + h^*(z) = z X$ implies $z = h'(X)$. That is, the derivative of the firm's wage bill equals its project value. Given an increasing and convex wage bill $h$, and effective skills $X$, this condition identifies increasing values for firm productivity $z$. 

We apply this logic to the parametric continuously differentiable function $h(X) = X^\eta + 2 \zeta$ where $\eta \geq 1$ governs the convexity of wages and $\zeta$ captures the lowest wage per worker. Using the derived fact that $z = h'(X)$, we can relate the distribution of firm projects $z$ to the convexity parameter $\eta$ of the wage bill. If $\eta=1$, there is no dispersion in firm productivity. We formalize this in Lemma \ref{prop:equilibrium_positive}.


\begin{lemma}\label{prop:equilibrium_positive}
For some firm distribution $F_z$ there exists an equilibrium with ($i$) a self-sorted assignment, and ($ii$) a wage function:  
\begin{equation}
w(x) = \frac{1}{2} \big( x^2_{c} + x^2_{m} \big)^\eta + \zeta . \label{e:eq_wages}
\end{equation}
\end{lemma}

\noindent The proof is in Appendix \ref{proof:equilibrium_positive}. The idea is to show there is a firm distribution $F_z$ so that given wage schedule (\ref{e:eq_wages}), workers and firms both optimize in a self-sorting equilibrium. Given the firm technology (\ref{e:firm_tech}) and the wage equation (\ref{e:eq_wages}), firm profits decrease in the difference between their workers' skills. In order to minimize output losses, firms thus hire pairs of identical workers. Given wage equation (\ref{e:eq_wages}), the worker problem (\ref{e:worker_problem}) has a unique solution, so that the distribution of worker inputs $F_x$ is uniquely determined by the worker problem. Finally, we map the firm distribution that induces (\ref{e:eq_wages}) as an equilibrium wage equation using $z = h'(X)$. We use these steps to pointwise identify the worker skill distribution as we show in Section \ref{s:quant}.

\section{Quantitative Analysis}\label{s:quant}

In this section we infer the distribution of cognitive and manual talents $\Phi$. The inference of the underlying distributions of skills, a central input for the calculation of the optimal tax formula, generalizes the approach of the unidimensional skills in \citet{Saez:2001} to a labor market model with multidimensional skills, coworker and firm effects. We also calibrate the parameter $\rho$ that governs the curvature of disutility with respect to effort.

\subsection{Data Sources}

We use data from the American Community Survey (ACS). We consider individuals between 25 and 60 years of age. The final sample from the ACS includes almost 16 million individuals between 2000 and 2019. For all our results, we use sample weights provided by the survey.  Our measure of labor income is wage and salary income before taxes over the past 12 months.\footnote{This measure includes wages, salaries, commissions, cash bonuses, tips, and other money income received from an employer. We drop individuals with earnings below a threshold to focus on workers who are attached to the labor market. This minimum is one-half of the federal minimum wage times 13 weeks at 40 hours per week (as in \citet{Guvenen:2014}).} 

The ACS contains occupational information for every worker. We combine a worker with the task intensity for their occupation using O*NET task measures from \citet{Acemoglu:2011}. Our cognitive measure is the average of their cognitive measures, and our manual measure is the average of their manual measures. Our resulting scores are approximately normally distributed across occupations.

For identification, we first construct a measure of relative task intensity by occupation. To obtain aggregated task production levels we use a Cobb-Douglas technology to map worker subtasks into final task production similar to \citet{Kremer:1993}, \citet{Acemoglu:2011} and \citet{Deming:2017}:  
\begin{equation}
q_s = \exp \bigg( \frac{1}{|\mathcal{V}|}\sum\limits_{\nu \in \mathcal{V}} \log q_{s \nu} \bigg) . \label{e:cobb_task}
\end{equation}
Letting $\log q_{s \nu}$ be the $Z$-score by subtask $\nu$, we obtain cognitive and manual task production levels. Since our aggregated cognitive and manual measure are approximately normally distributed, task production levels are approximately lognormal. We now make an identification assumption that the relative task input is equal to the relative task production level, $x_m/x_c = q_m/q_c$, which is hence also approximately lognormally distributed across occupations.

   \begin{figure}[!t]
    \begin{subfigure}{0.34\linewidth}
    \centering
        \includegraphics[trim=0.0cm 0.0cm 0.0cm 0.0cm, width=1\textwidth,height=0.22\textheight]{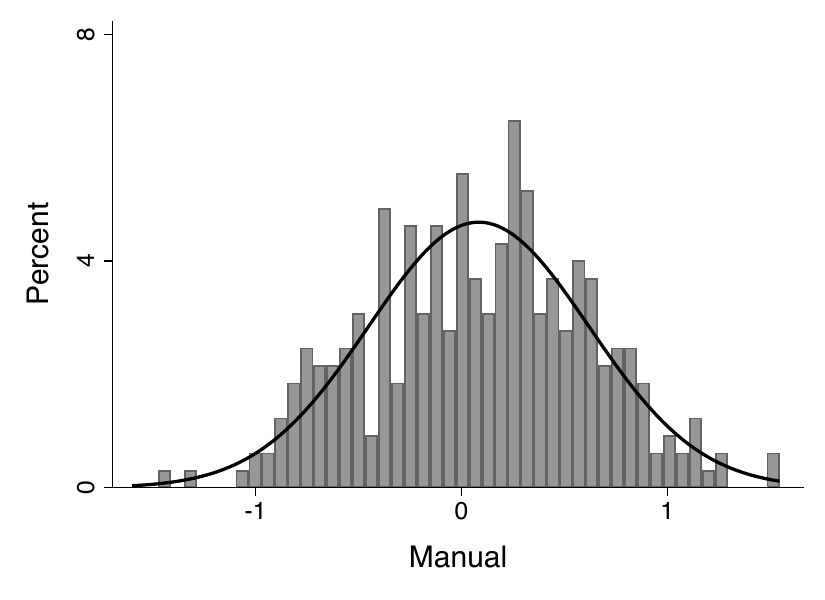}
    \end{subfigure}%
    \begin{subfigure}{0.34\linewidth}
    \centering
        \includegraphics[trim=0.0cm 0.0cm 0.0cm 0.0cm, width=1\textwidth,height=0.22\textheight]{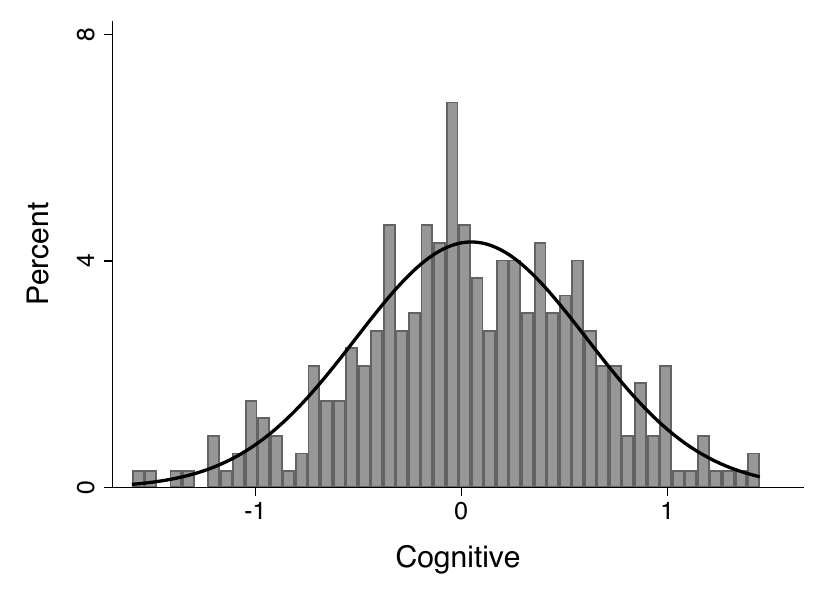}
    \end{subfigure}
    \begin{subfigure}{0.34\linewidth}
    \centering
        \includegraphics[trim=0.0cm 0.0cm 0.0cm 0.0cm, width=1\textwidth,height=0.22\textheight]{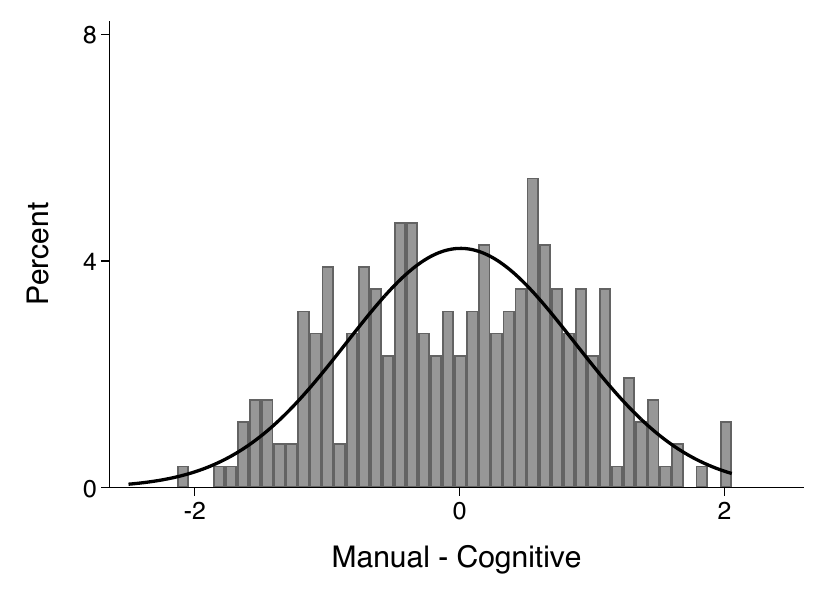}
    \end{subfigure}
    \vspace{-0.2 cm}
    \caption{Task Intensity Across Occupations} \label{f:task_intensity}
{\scriptsize \vspace{.2 cm} Figure \ref{f:task_intensity} shows the distribution of manual and cognitive task production levels across occupations in logs (left and center panel) together with the relative distribution of manual and cognitive task intensity (right panel). Each distribution is well-approximated by a lognormal distribution.}
    \label{fig:manmade}
    \end{figure}

Figure \ref{f:task_intensity} shows the distribution of manual and cognitive task production levels across occupations in logs together with the relative distribution of manual and cognitive task intensity. The first two panels show that the distribution of manual task production levels and the distribution of cognitive task production levels can be described by a lognormal distribution. The right panel shows that the same holds for the relative manual task intensity.

Figure \ref{f:relative_wages} displays the relation between relative task intensity and average earnings across occupations. Earnings are low for occupations with high manual task intensity, such as gardeners and truck drivers, while earnings are high for occupations with high cognitive task intensity such as software developers and actuaries. Moving from the 25th percentile to the 75th percentile in relative manual task intensity decreases earnings from 62 to 35 thousand dollars. 

   \begin{figure}[!t]
    \begin{subfigure}{1\linewidth}
    \centering
        \includegraphics[trim=0.0cm 0.0cm 0.0cm 0.0cm, width=0.62\textwidth,height=0.31\textheight]{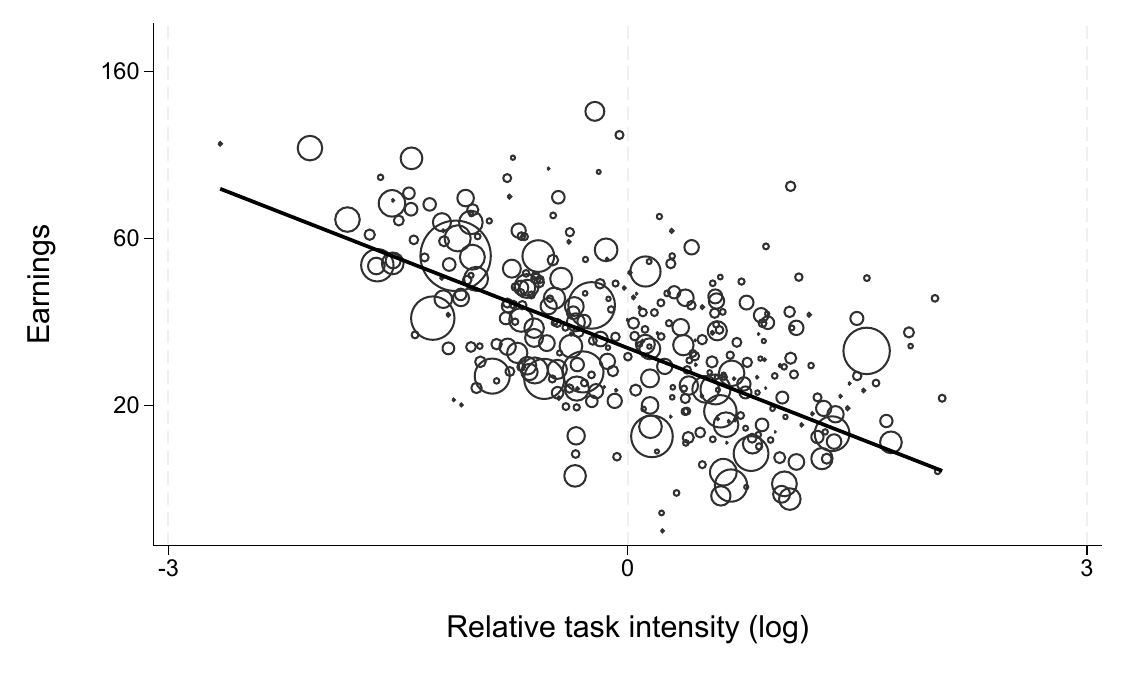}
    \end{subfigure}%
    \vspace{-0.2 cm}
    \caption{Earnings and Relative Task Intensity} \label{f:relative_wages}
{\scriptsize \vspace{.2 cm} \Cref{f:relative_wages} show the relation between average earnings (y-axis, logarithmic scale) and relative task intensity across occupations. Average earnings are decreasing in the relative manual task intensity. The size of each circle corresponds to the occupation's employment share.}
    \end{figure}


\subsection{Calibration}

We now calibrate the positive model. We parameterize fiscal policy and preferences, and infer the underlying multidimensional skill distribution.

The government taxes labor earnings to finance expenditures $G$. If pre-tax earnings are $w$, then taxes are given by $
T(w) = \tau w$. After-tax earnings are thus $( 1 - \tau ) w$, we set $\tau = 0.3$.

Firm heterogeneity governs the convexity of the wage schedule (see Lemma \ref{prop:equilibrium_positive}). We set the curvature parameter for the wage schedule $\eta$ to align the added variation in log wages due to firm heterogeneity with evidence from the literature on variation in log wages due to firm effects. Using the wage equation (\ref{e:eq_wages}), the variation in firm projects multiplies the underlying variation across workers by $\eta^2$. We set $\eta = 1.1$ to attribute 17 percent of the added variation in wages to firm effects. Our target of 17 percent is in line with estimates from the literature.\footnote{For example, \citet{Abowd:2003} find that firm variation makes up 17 percent of the variance in wages while \citet{Song:2019} instead report that firm variation makes up between 8 percent and 12 percent.}




We next discuss the calibration of worker preferences. We use linear preferences with respect to consumption goods, $u(c) = c$, and estimate the parameter governing the curvature of the disutility function to efforts in each task $\rho$. We set $\rho$ such that a regression of log market hours on hourly wages, holding constant the marginal value of wealth, yields a coefficient of 0.55. This target value comes from the meta-analysis of estimates of the intensive margin Frisch elasticity from \citet{Chetty:2012}.

To use estimates for the Frisch elasticity for total hours with respect to hourly productivity to calibrate the curvature of the utility function with respect to effort, we derive this expression within our model. Given the specification for the disutility from work (\ref{e:disutility}), the linear utility from consumption, and the worker technology (\ref{e:worker_tech}), the worker's problem (\ref{e:worker_problem}) is: 
\begin{equation}
\max\limits_{x_c,x_m} \; \frac{1}{2}( 1 - \tau ) ( x_{c}^2 + x_{m}^2 )^\eta - \kappa \Big( \frac{x_c}{\alpha_c} \Big)^\rho - \kappa \Big( \frac{x_m}{\alpha_m} \Big)^\rho .
\end{equation}
The optimality condition to the worker's problem for each task $s \in \mathcal{S}$ is:
\begin{equation}
(1-\tau) \eta \big( 2 w(x) \big)^{\frac{\eta-1}{\eta}} =  \kappa \rho \frac{x_s^{\rho-2}}{\alpha^{\rho}_s} , \label{e:foc_s}
\end{equation}
where $w(x) = \frac{1}{2} ( x^2_{c} + x^2_{m} )^\eta + \zeta $ by wage equation (\ref{e:eq_wages}) with $\zeta$ representing minimum earnings in our data. That is, the marginal consumption utility from supplying extra tasks equals the marginal cost of effort. Taking the ratio of these optimality conditions, this implies that the skill, effort and task intensity ratio are related by:
\begin{equation}
\frac{\alpha_m}{\alpha_c} = \Big( \frac{x_m}{x_c} \Big)^\frac{\rho-2}{\rho} = \Big( \frac{\ell_m}{\ell_c} \Big)^\frac{\rho-2}{2} , \label{e:skill_ratio} 
\end{equation}
where the second equality follows from the worker task technology (\ref{e:worker_tech}).  The marginal rate of substitution between activities, $\big( \frac{\ell_c}{\ell_m} \big)^{\rho-1}$, is equal to the ratio of marginal benefits between activities, $\big( \frac{\alpha_{c}}{\alpha_{m}} \big)^{2} \frac{\ell_{c}}{\ell_{m}}$. Relative efforts are determined by relative skills $\frac{\alpha_{c}}{\alpha_{m}}$. Workers spend more effort on tasks in which they are more talented. 

Using the first-order conditions for effort, and observing that the share of total efforts on each task is constant by (\ref{e:skill_ratio}), we can express the Frisch elasticity of total hours $\ell_c + \ell_m$ as:\footnote{See \Cref{a:worker_problem}.} 
\begin{equation}
\varepsilon = \frac{\partial \log (\ell_c + \ell_m)}{\partial \log z(x)} \bigg\vert_{\lambda} = \frac{\partial \log (\ell_c + \ell_m)}{\partial \log (1 - \tau)} \bigg\vert_{\lambda}  = \frac{1}{\rho - 1} \label{e:frisch_calibration},
\end{equation}
where $\lambda$ is the marginal value of wealth, and $z(x) := w(x) / (\ell_c + \ell_m)$ is productivity per hour. We set $\rho = 2.8$ so that the Frisch elasticity $\varepsilon$ is indeed $0.55$. Finally, we normalize $\kappa = \frac{1}{2\rho}$.
\vspace{0.4 cm}
\noindent \textbf{Skill Distribution}. We now identify the skill distribution pointwise. Using the solution to the worker's problem, together with data on both total earnings and occupational relative task intensity for each worker, we separately identify two sources of worker productivity $(\alpha_c, \alpha_m)$ that rationalize the data as a model outcome. This identification argument is similar to \citet{BK2:2020,BK1:2021} who use explicit solutions for home production models to identify productivity at home and to identify permanent and transitory market productivity using data on consumption, home and market hours.

Using the O*NET task measures, we have information on the relative task intensity for each occupation $\frac{x_m}{x_c}$ and, hence, we identify the relative skills $\frac{\alpha_m}{\alpha_c}$ by equation (\ref{e:skill_ratio}). In order to determine the level of tasks, we use the wage equation (\ref{e:eq_wages}):
\begin{equation}
w(x) = \frac{1}{2} \big( x^2_c +  x^2_m \big)^\eta = \frac{ x^{2 \eta}_c}{2} \bigg( 1 + \Big( \frac{x_m}{x_c} \Big)^2 \bigg)^\eta . \label{e:eq_wages_data}
\end{equation}
Given the skill ratio for an individual's occupation, $\frac{x_m}{x_c}$, and an individual's earnings $w(x)$, this equation uniquely determines the level of cognitive tasks $x_c$, and hence the level of manual tasks $x_m$. By the optimality condition (\ref{e:foc_s}), we identify both cognitive skills $\alpha_c$ and manual skills $\alpha_m$ for each worker. 


\begin{table}[t!]
\def\arraystretch{1.4}%
\begin{center}
\caption{Example of Identification}\label{t:simple_example}
\begin{tabular}{clcccccc}
\hline  \hline
 & & \multicolumn{1}{c}{Relative Task} &  \multicolumn{1}{c}{Wages} & \multicolumn{2}{c}{Task Intensity}   & \multicolumn{2}{c}{Task Skills}   \\
 & & \hspace{0.45 cm}  $x_m/x_c$ \hspace{0.45 cm}  & \hspace{0.45 cm} $ w(x)$ \hspace{0.45 cm} & \hspace{0.45 cm}  $x_m$ \hspace{0.45 cm}  & \hspace{0.45 cm}  $x_c$ \hspace{0.45 cm} & \hspace{0.45 cm}  $\alpha_m^\rho$ \hspace{0.45 cm}  & \hspace{0.45 cm}  $\alpha_c^\rho$ \hspace{0.45 cm}   \\
\hline
1 & \; Baseline		   			  & 		1 & 		1 & 		  1.00     & 		  1.00  & 0.50 &  0.50  \\
2 & \; Task intensity \hspace{0.10 cm} 	   &		3 & 		1 &   		  1.35     &		          0.45  &  0.63 & 0.26  \\ 
3 & \; Wages			  			  &		1 & 		4 & 	  	  2.00     & 		  2.00  &  0.87 & 0.87  \\ 
4 & \; Taxes $\tau = 0.3$	   			  &		1 & 		1 & 	  	  1.00     &  		  1.00  & 0.71 & 0.71 \\ 
5 & \; Firms $\eta = 1.1$	   			   &		1 & 		1 &  		  0.97     & 	           0.97  & 0.42 & 0.42 \\
\hline \hline
\end{tabular}\end{center}
{\footnotesize \Cref{t:simple_example} illustrates the identification of workers' manual and cognitive skills through five examples. We infer higher levels of manual skills with higher manual task intensity (in Row 2), higher earnings (Row 3), higher taxes (Row 4), and with less dispersion in firms' project values (Row 5).}
\end{table}

\vspace{0.4 cm}
\noindent \textbf{Examples}. In order to provide insight into the identification of worker skill heterogeneity, we consider a numerical example. We first consider an economy without taxes $\tau = 0$ and without heterogeneity in firm projects, $\eta = 1$. 

Suppose a worker's occupational relative task intensity is equal to one, $\frac{q_m}{q_c} = \frac{x_m}{x_c} = 1$, and their earnings equal mean earnings, which we normalize to one. By equation (\ref{e:eq_wages_data}), the worker's cognitive task intensity and the worker's manual task intensity are equal to $1$. Using the optimality condition for task inputs (\ref{e:foc_s}), $\alpha_s^\rho = \frac{1}{2}$, implying the worker is equally skilled in both tasks. This worker is presented in the first row of \Cref{t:simple_example}.

Inferred manual skill increases with manual task intensity. Consider some worker with relative manual task intensity equal to three, $\frac{x_m}{x_c} = 3$, and average earnings. By equation (\ref{e:eq_wages_data}), the cognitive task intensity is $x_c=\frac{1}{\sqrt{5}}<1$ and hence the worker's manual task intensity is greater with $x_m = \frac{3}{\sqrt{5}} > 1$. Since $\alpha^{\rho}_s = \frac{1}{2} x_s^{\rho-2}$, it follows that the worker's inferred manual skill increases with relative manual task intensity, while the worker's cognitive skills decreases, as shown in the second row of \Cref{t:simple_example}.

Inferred skill levels increase with earnings. For a worker with a relative task intensity of one, but a high level of earnings, the relative skill intensity is one but the level of each task is greater. Consider a worker earning four times average earnings. By equation (\ref{e:eq_wages_data}), we identify the worker's cognitive task intensity, and therefore the worker's manual task intensity, to be equal to $2$. Using the worker's optimality condition for task inputs (\ref{e:foc_s}), $\alpha_s^\rho = \frac{1}{2} 2^{\rho-2}$, implying that the worker is equally skilled in both tasks, and almost 1.75 times as skilled as a worker in the same occupation earning average earnings. This worker is presented in the third row of \Cref{t:simple_example}.

The presence of taxes does not affect inferred task intensities $x$ but does increase the inferred skill levels $\alpha$. Since the identification of the task intensity is based on pretax earnings (\ref{e:eq_wages_data}), inferred task intensities do not vary with taxes. For $\eta =1$, since the task intensity does not change with taxes, we obtain $\alpha^{\rho}_s = \frac{1}{2(1-\tau)} x_s^{\rho-2}$. When workers are taxed, the marginal benefit from completing tasks is reduced. In order to rationalize the same levels of cognitive and manual task intensity supplied by a worker, it must be less costly for the worker to complete tasks due to increased levels of skills, as shown in the fourth row of \Cref{t:simple_example}.

Finally, increased dispersion in firm project values decreases wage dispersion that is attributed to dispersion in task intensity. Consider the dispersion in firm projects with $\eta > 1$. Reorganizing the wage equation (\ref{e:eq_wages_data}), $x_c = \left(2 w(x) \right)^{\frac{1}{2\eta}} \Big/ \sqrt{1 + \left( \frac{x_m}{x_c} \right)^2}$, shows that higher values of $\eta$ compress the dispersion in task intensity. Further, by combining the first-order condition (\ref{e:foc_s}) with wage equation (\ref{e:eq_wages_data}), we obtain $\alpha^\rho_s \propto w(x)^{\frac{\rho}{2 \eta} - 1}$. An increase in $\eta$ decreases the effective dispersion in skills. Dispersion in firm projects magnifies underlying differences in task intensity due to the positive sorting between workers and projects. Equivalently, small differences in effective worker skills generate large earnings differences.

\begin{table}[t!]
\def\arraystretch{1.4}%
\begin{center}
\caption{Illustration of Identification}\label{t:example}
\begin{tabular}{lcccccc}
\hline  \hline
\multicolumn{1}{l}{Occupation} &  \multicolumn{1}{c}{Relative} &  \multicolumn{1}{c}{Wages} & \multicolumn{1}{c}{Manual}  &  \multicolumn{1}{c}{Cognitive}  &  \multicolumn{1}{c}{Firm} &  \multicolumn{1}{c}{SOC Code}  \\
  & \hspace{0.15 cm}  $\log \frac{q_m}{q_c}$ \hspace{0.15 cm}  & \hspace{0.15 cm} $\mathbb{E} w(x)$ \hspace{0.15 cm} & \hspace{0.15 cm}  $\frac{\alpha_m - \mathbb{E}\alpha_m}{\sigma_{m}}$ \hspace{0.15 cm}  &\hspace{0.15 cm}   $\frac{\alpha_c - \mathbb{E}\alpha_c}{\sigma_{c}}$ \hspace{0.15 cm} &\hspace{0.15 cm}   $\frac{\alpha_z - \mathbb{E}\alpha_z}{\sigma_{z}}$ \hspace{0.15 cm}  \\
\hline
Gardeners 				\hspace{1.7 cm} & \phantom{-}1.7 & \phantom{1}23 & \phantom{-}0.93     & 		  -2.35  & -1.28 &37$-$3010  \\ 
Cashiers				 	 			 &\phantom{-}0.7 & \phantom{1}20 & \phantom{-}0.47     & 		  -1.16  & -1.62 &41$-$2010  \\ 
Police officers 			 				 & 		     -0.1 & \phantom{1}64 & \phantom{-}0.82     & \phantom{-}0.48  & \phantom{-}0.82 &33$-$3050 \\ 
Physicians 	 			 			 &		     -0.2 & 			184 & \phantom{-}1.77     & \phantom{-}1.32 & \phantom{-}3.11 &29$-$1060  \\ 
Chief executives  						 & 		     -2.1 & 			149 & 		  -2.39     & \phantom{-}1.71  & \phantom{-}2.63 &11$-$1010  \\ 
Actuaries 								 &                 -2.7 & 			136 &  		  -3.46     & \phantom{-}1.65  & \phantom{-}2.43 &15$-$2010  \\ 
\hline \hline
\end{tabular}\end{center}
{\footnotesize \Cref{t:example} shows the identification of worker skills for a number of occupations. Holding constant the relative manual skill intensity, high earnings identify high skill levels as seen by comparing the manual and cognitive skills of police officers and physicians. Holding constant earnings, high manual task intensity identifies high manual skills as seen by comparing the skills of gardeners and cashiers.}
\end{table}

\vspace{0.4 cm}
\noindent Having illustrated the identification with examples, we turn to identification using earnings data. \Cref{t:example} illustrates the identification of underlying skills for representative workers in occupations listed in the first column. The second column shows the relative manual task intensity for these occupations from O*NET task measures. The third column shows average earnings of the workers by occupation in the ACS. \Cref{t:example} shows a negative relation between manual task intensity and average earnings by occupation, in line with Figure \ref{f:relative_wages}. 

In order to identify manual and cognitive skills, we use equations (\ref{e:foc_s}), (\ref{e:skill_ratio}) and (\ref{e:eq_wages_data}). First, we establish that higher earnings identify higher levels of skills, everything else equal. Consider an example of police officers and physicians. Since the relative task intensity for police officers and physicians is comparable, their relative skills are comparable by (\ref{e:skill_ratio}). Average earnings of physicians exceed the average earnings of police officers implying a higher level of both cognitive and manual skills for physicians. Indeed, the fourth and fifth column in \Cref{t:example} show that while both physicians and police officers' cognitive and manual talents exceed the population average, $\alpha_s > \mathbb{E}\alpha_s$, the skills of physicians exceed the skills of police officers in both dimensions. 


Second, we consider two occupations with similar wages to show that high manual task intensity identifies high manual skill all else being equal. While the earnings of gardeners and cashiers are similar, gardening is more demanding in manual skills. By equation (\ref{e:eq_wages_data}), the cognitive task requirements of gardeners are lower than the cognitive task requirements for cashiers. By equation (\ref{e:skill_ratio}) it follows that a gardener has more manual skills than a cashier, but less cognitive skills. The fourth and fifth column in \Cref{t:example} displays this pattern.

   \begin{figure}[!t]
    \begin{subfigure}{\linewidth}
    \centering
        \includegraphics[trim=0.0cm 0.0cm 0.0cm 0.0cm, width=0.56\textwidth,height=0.29\textheight]{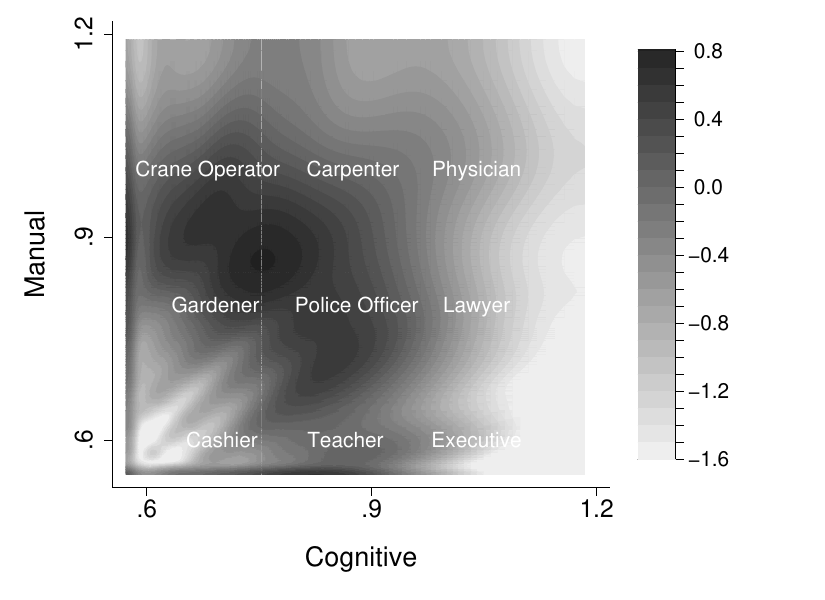}
    \end{subfigure}%
    \vspace{-0.2 cm}
    \caption{Inferred Skill Distribution} \label{f:skill_distribution}
{\scriptsize \vspace{.2 cm} \Cref{f:skill_distribution} shows the inferred worker skill distribution, with bright colors indicating more mass. The panel shows the smoothed distribution of cognitive and manual skills that exactly rationalizes the data which is obtained using data on relative task intensity by occupation and worker earnings, through equations (\ref{e:foc_s}) to (\ref{e:eq_wages_data}). The values are normalized such that one reflects a uniform distribution.}
    \end{figure}

We apply the identification argument to all workers in the ACS to identify their skills. By identifying skills at the worker level, we allow for skill heterogeneity within occupations driven by earnings differences within occupation. As in the example, workers with high earnings have higher cognitive and manual skills than a worker with low earnings in the same occupation. \Cref{f:skill_distribution} shows the resulting distribution of cognitive and manual skills, after 98 percent winsorization and after smoothing the pointwise identified distribution using a kernel density estimation.\footnote{We correct our kernel density estimator at the boundaries of our rectangular type space by reflecting along all boundaries, see, e.g. \citet{Karunamuni:2005}.}

   \begin{figure}[!t]
    \begin{subfigure}{0.49\linewidth}
    \centering
        \includegraphics[trim=0.0cm 0.0cm 0.0cm 0.0cm, width=\textwidth,height=0.25\textheight]{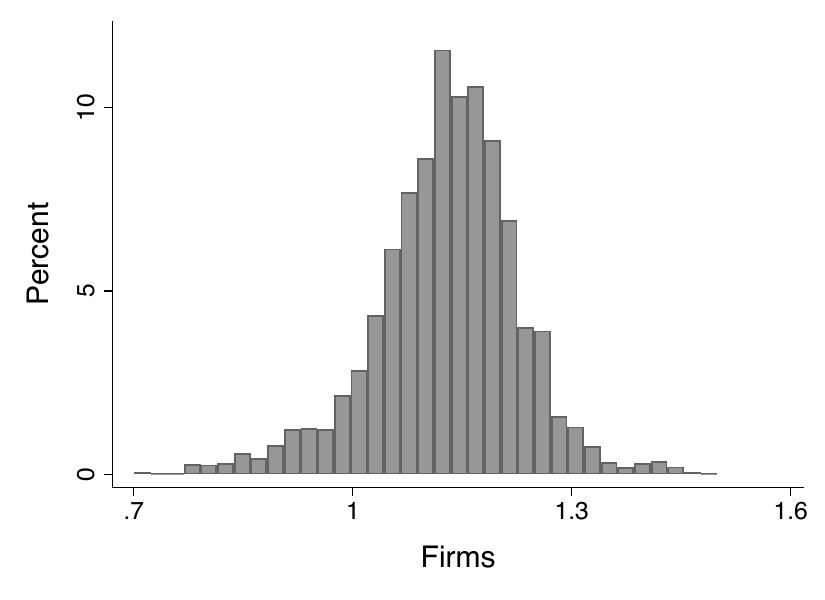}
    \end{subfigure}%
      \begin{subfigure}{0.49\linewidth}
    \centering
        \includegraphics[trim=0.0cm 0.0cm 0.0cm 0.0cm, width=\textwidth,height=0.25\textheight]{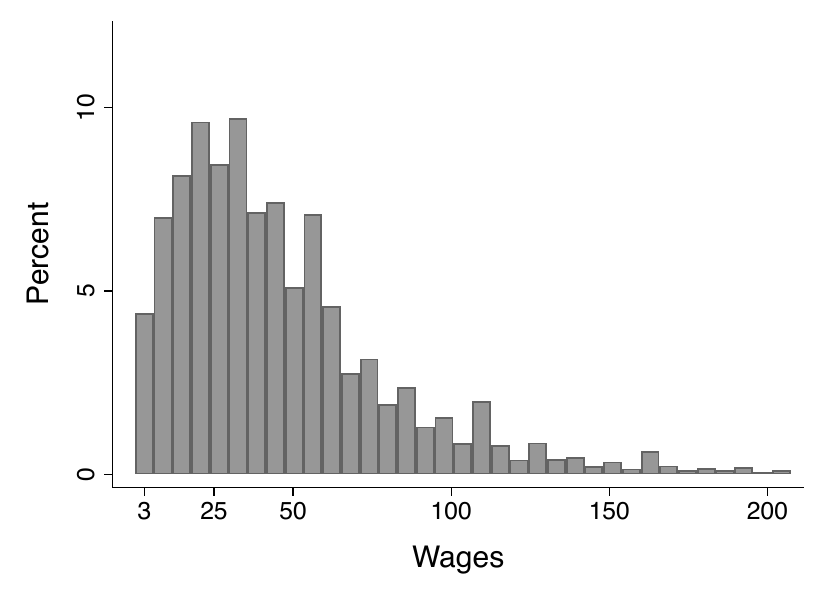}
    \end{subfigure}%
    \vspace{-0.2 cm}
    \caption{Firm and Wage Distribution} \label{f:firm_wage_distribution}
{\scriptsize \vspace{.2 cm} \Cref{f:firm_wage_distribution} shows the histogram for the inferred firm distribution (left panel) and the model implied distribution of wages (right panel).}
    \end{figure}

For illustrative purposes, we introduce representative occupations in \Cref{f:skill_distribution}. Specifically, we provide nine representative occupations within the type space. For example, cashiers are workers with both low cognitive and low manual skills, chief executives have low manual skills but high cognitive skills, while physicians have both high cognitive and high manual skills.

Finally, \Cref{f:firm_wage_distribution} shows the inferred firm productivity distribution in the left panel and the implied wage distribution in the right panel. The left hand distribution shows that the distribution of firm projects is relatively concentrated with project values ranging from 30 percent below the mean to 40 percent above the mean (1.1). By construction, the right panel replicates the empirical wage distribution.

\section{Quantitative Results} \label{s:quantres}

In this section, we present the quantitative results to the planning problem using the empirically relevant model of Section \ref{s:quant}. 

\subsection{Unconstrained Benchmark} \label{s:benchmark}

In order to build intuition for the solution, we first present a benchmark without incentive constraints and firm heterogeneity. The planning problem then simplifies to minimizing resource costs (\ref{e:resources_original2}) subject to the promise keeping condition (\ref{e:promise_keeping}). By using the functional form for preferences, the promise keeping condition simplifies to:
\begin{equation}
\int \Big( \hspace{-0.05 cm} c(\alpha) - \kappa \big( x_c(\alpha) \big/ \alpha_c \big)^\rho - \kappa \big( x_m(\alpha) \big/ \alpha_m \big)^\rho  \Big) \text{d} \Phi \geq \mathcal{U} \label{e:promise_keeping2}.
\end{equation}

At the optimum, cognitive tasks are independent of workers' routine skills, and the elasticity of cognitive tasks with respect to cognitive skills is $\frac{\rho}{\rho - 2}$. Furthermore, the solution does not feature bunching. In order to see this, note that the following condition has to be satisfied:
\begin{equation}
x_s \propto \alpha_s^\frac{\rho}{\rho-2}, \label{e:x_s}
\end{equation}
for each skill $s \in \{c,m\}$. Due to additive separability of tasks in preferences and technology, the efforts on task $s$ depend only on the worker's skills in this task. Equivalently, there is no cross-dependence between tasks. Since (\ref{e:x_s}) describes a one-to-one relation between the worker's skills and efforts in each task, there is no bunching at optimum. That is, in a neighborhood of worker $\alpha$, every pair of distinct workers $(\alpha',\alpha'')$ is assigned distinct allocations as $x(\alpha') \neq x(\alpha'')$. 


   \begin{figure}[t!]
    \begin{subfigure}{0.50\linewidth}
    \centering
        \includegraphics[trim=0.0cm 0.0cm 0.0cm 0.0cm, width=1\textwidth,height=0.29\textheight]{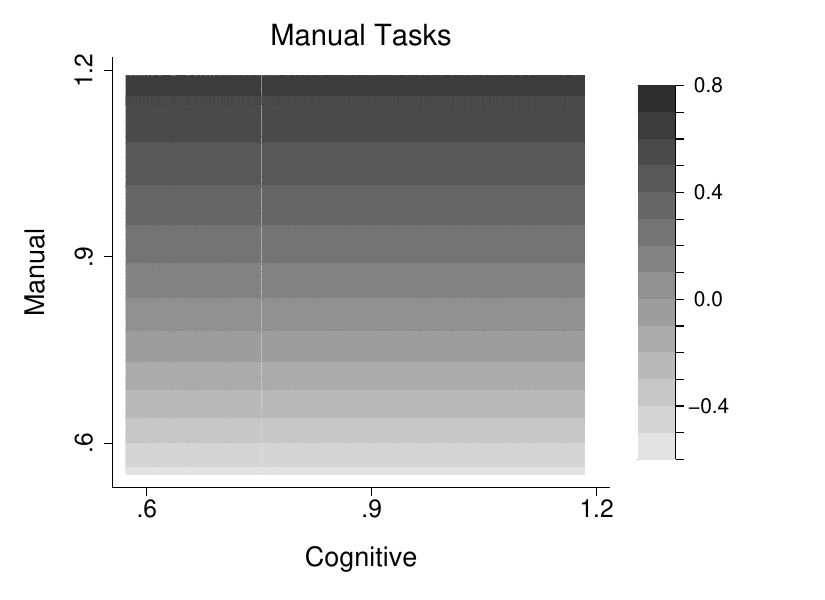}
    \end{subfigure}%
    \begin{subfigure}{0.50\linewidth}
    \centering
        \includegraphics[trim=0.0cm 0.0cm 0.0cm 0.0cm, width=1\textwidth,height=0.29\textheight]{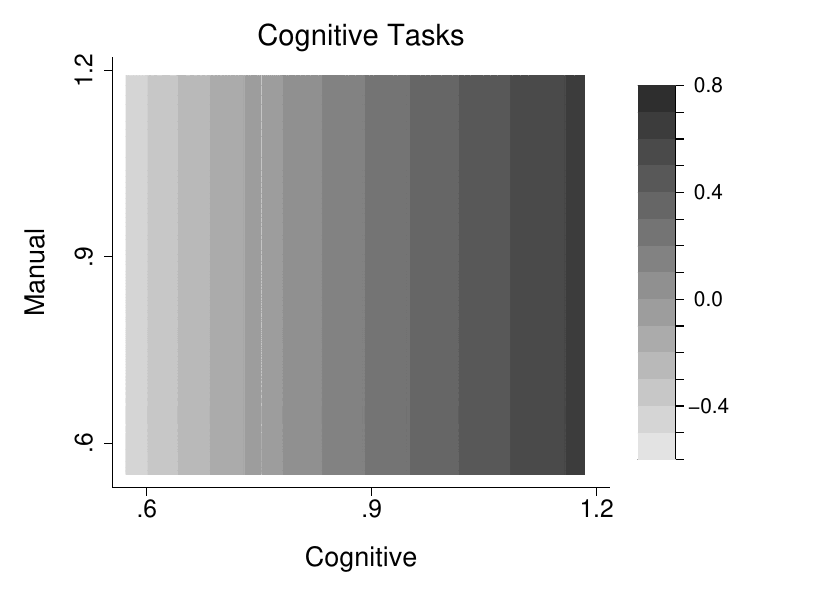}
    \end{subfigure}%
    \vspace{-0.2 cm}
    \caption{Benchmark Allocation} \label{f:data_autarky}
{\scriptsize \vspace{.2 cm} Figure \ref{f:data_autarky} shows the benchmark allocation for task intensity by worker's cognitive and manual skills. The left panel shows the allocation of manual tasks, the right panel illustrates the allocation of cognitive tasks. The optimal allocation does not feature any cross-dependence between tasks: manual task intensity only varies with manual skill, while cognitive task intensity only varies with cognitive skill. }
    \end{figure}



Given the empirical description of the distribution for cognitive and manual skills in \Cref{f:skill_distribution}, equation (\ref{e:x_s}) gives the optimal allocation of both cognitive and manual tasks. Figure \ref{f:data_autarky} visualizes the benchmark allocation of task intensity by worker's cognitive and manual skills. The left panel shows the allocation of manual tasks, the right panel shows the allocation of cognitive tasks. Since (\ref{e:x_s}) rules out any cross-dependence between tasks, the optimal allocation is captured by parallel horizontal and vertical lines, respectively. Manual task intensity only varies with manual skill, while cognitive task intensity only varies with cognitive skill. 

\subsection{Optimal Solution}

   \begin{figure}[t!]
    \begin{subfigure}{0.50\linewidth}
    \centering
        \includegraphics[trim=0.0cm 0.0cm 0.0cm 0.0cm, width=1\textwidth,height=0.26\textheight]{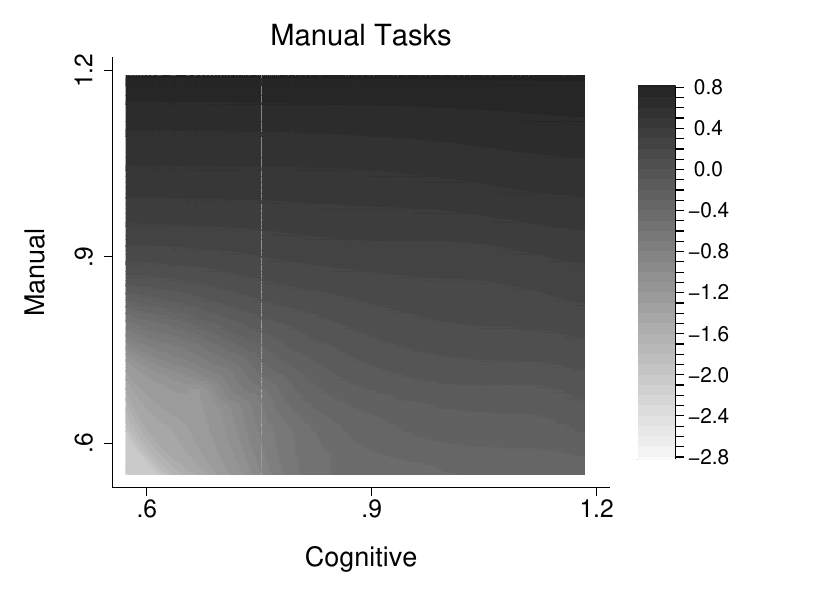}
    \end{subfigure}%
    \begin{subfigure}{0.50\linewidth}
    \centering
        \includegraphics[trim=0.0cm 0.0cm 0.0cm 0.0cm, width=1\textwidth,height=0.26\textheight]{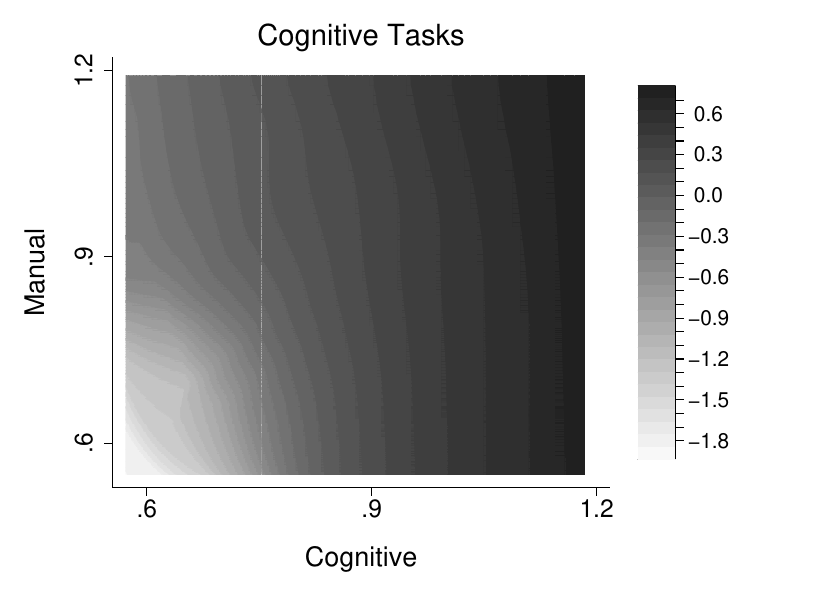}
    \end{subfigure}%
    
    \begin{subfigure}{0.50\linewidth}
    \centering
        \includegraphics[trim=0.0cm 0.0cm 0.0cm 0.0cm, width=1\textwidth,height=0.26\textheight]{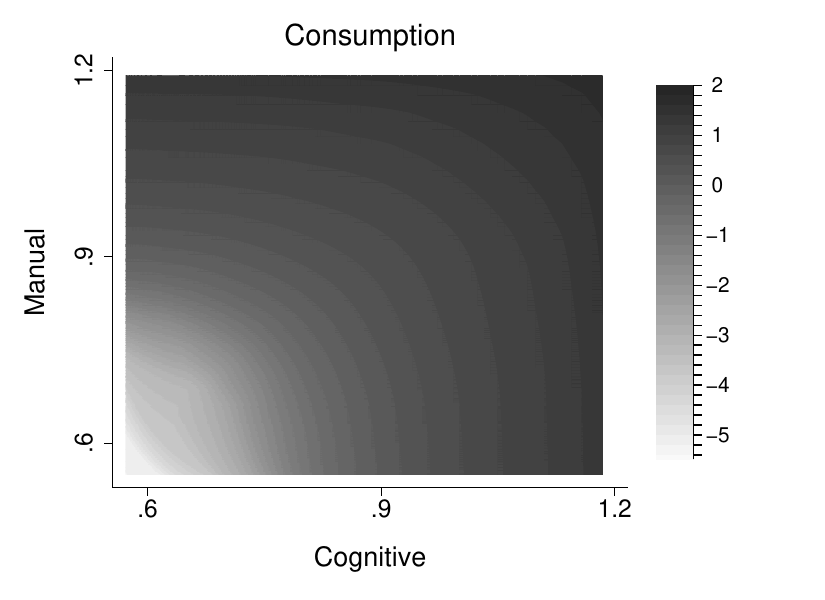}
    \end{subfigure}%
    \begin{subfigure}{0.50\linewidth}
    \centering
        \includegraphics[trim=0.0cm 0.0cm 0.0cm 0.0cm, width=1\textwidth,height=0.26\textheight]{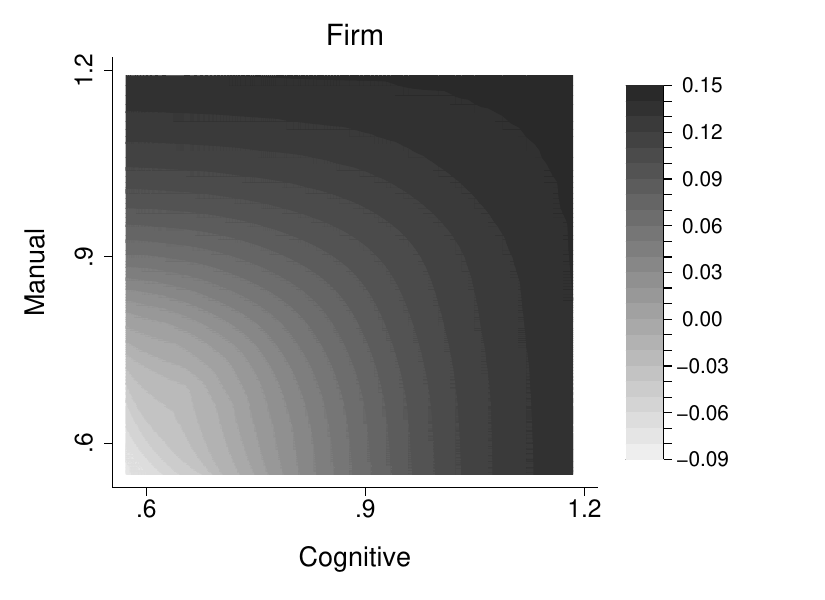}
    \end{subfigure}%
    \vspace{-0.2 cm}
    \caption{Planner Allocation} \label{f:data_results}
{\scriptsize \vspace{.2 cm} \Cref{f:data_results} visualizes the solution by worker's cognitive and manual skills. The top row shows the manual and cognitive task allocation, the bottom row shows the consumption allocation and the assignment of workers to firms. The solution features positive dependence between tasks. For example, optimal cognitive task intensity increases with manual skills.}
    \end{figure}

\Cref{f:data_results} shows the solution to the planner problem. The top row shows the allocation of manual and cognitive tasks, the bottom row shows the allocation of consumption and the assignment of workers to firms. In contrast to the benchmark, optimal task intensity in one skill depends positively on a worker's other skills. Consider the manual task allocation in the top left panel. Similar to the benchmark, the manual task intensity increases with a worker's manual skills holding constant their cognitive skills. In contrast to the benchmark solution, the manual task intensity also increases with workers' cognitive skills. That is, workers with the same manual ability but with a higher cognitive ability conduct a higher level of manual tasks. Moreover, this codependence between cognitive skills and manual tasks intensifies at low levels of cognitive skill. This can be seen by the contour lines being almost negative 45 degree lines at low levels of manual ability, while being almost flat at high levels of manual ability. The same pattern holds for cognitive tasks.

In this economy, the binding incentive constraints are for high types to mimic to be low types, which is also generally the case with unidimensional skill. In order to prevent the high type from pretending to be the low type, the allocation for the low types is distorted. With multidimensional skills, the allocation for the low types is distorted both by reducing the level of task output similar to the unidimensional case, and by increasing the codependence between tasks. The latter is a new type of distortion that emerges in taxation problems with multidimensional skill.


The bottom left panel shows the solution for consumption. Consumption increases with skills. Consumption of workers with top cognitive skills exceeds consumption of workers with top manual skills due to higher absolute levels of skill. The bottom right panel shows the assignment of workers to firms. Given the cognitive and manual tasks, the planner assigns workers with greatest effective skills $X = x^2_c + x^2_m$ to projects of greater value following Proposition \ref{prop:p_assignment}. A physician thus works on a more valuable project than a cashier as in the positive economy. Since the range of the cognitive skills is higher than the range of the manual skills, the high value projects are assigned towards workers with greater cognitive skills.

\vspace{0.4 cm}
\noindent \textbf{Bunching}. We now describe the nature of bunching in the optimal solution. Bunching means that different workers are assigned identical labor supply allocations and, therefore, are also assigned identical consumption allocations (see Section \ref{s:bunching_theory}). We use three distinct methods based on the theoretical analysis in Section \ref{s:planner} and Section \ref{s:characterization} to comprehensively characterize the bunching patterns that emerge in the quantitative model.

\begin{figure}[!t]

\begin{centering}
\includegraphics[width=0.9\textwidth,height=0.28\textheight ]{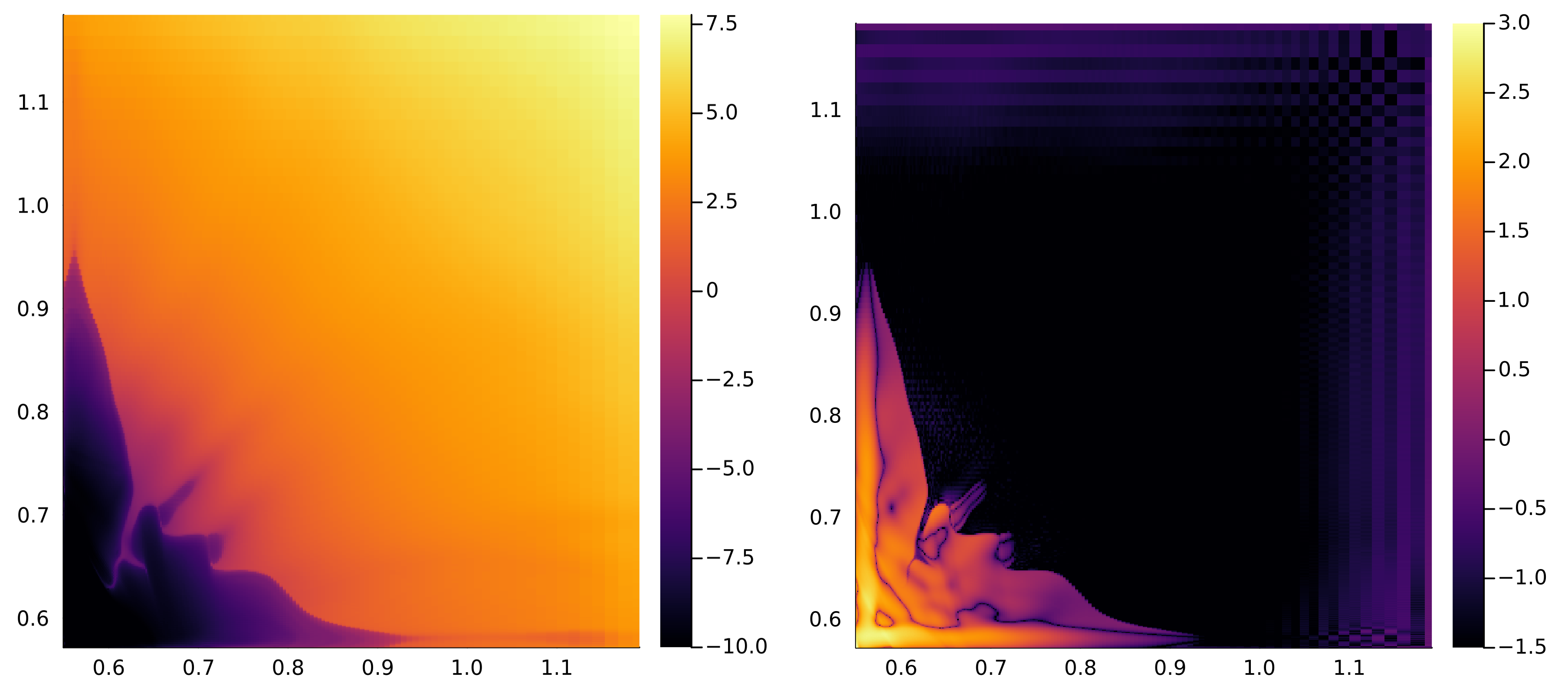}
\par\end{centering}

\caption{Bunching}
{\scriptsize{}{}\vspace{0.2cm}
 \Cref{f:baselinepR1emp} illustrates bunching in the optimal allocation. The left panel identifies bunching by analyzing the determinant of the Hessian matrix of the indirect utility function, and shows the value of the determinant on a log base 10 scale. The right panel identifies bunching using Corollary \ref{p:el} by analyzing deviations from the multidimensional optimal tax formula on a log base 10 scale. The variable on the horizontal axis is cognitive skill $\alpha_c$; the variable on the vertical axis is manual skill $\alpha_m$. Both panels identify that workers in the bottom left region of the type space are bunched under the optimal allocation. \label{f:baselinepR1emp}}
\end{figure}

First, we use our theoretical results in Lemma \ref{l:strconvex} to identify bunching by analyzing the determinant of the Hessian matrix of the indirect utility function. By Lemma \ref{l:strconvex}, if the indirect utility function is not strongly convex at all points in the neighborhood of type $p$, then worker $p$ is bunched. If the Hessian matrix is not invertible, then the indirect utility function is not strongly convex. A matrix is invertible if and only if the determinant is not equal to zero. Therefore, if the determinant of the Hessian matrix equals zero, the matrix is not invertible, so the indirect utility function is not strongly convex and the worker is bunched. Thus bunching is present in regions where the determinant of the Hessian matrix of the indirect utility function is equal to zero. We apply this method in the left panel of \Cref{f:baselinepR1emp}, which shows the value of the determinant on a log base 10 scale. The left panel shows that workers in the bottom left (dark) region are bunched.

Second, we now use our theoretical results in Corollary \ref{p:el} to identify bunching. If the multidimensional optimal tax formula without bunching does not hold, then the indirect utility function is not strongly convex for worker $p$. By Lemma \ref{l:strconvex}, this implies that worker type $p$ is bunched. Numerically, we analyze deviations from the multidimensional optimal tax formula without bunching to establish bunching. We apply this method in the right panel of \Cref{f:baselinepR1emp}. The figure shows the deviations from the multidimensional optimal tax formula without bunching on a log base 10 scale. The right panel delivers the same bunching region as the left panel in the bottom left (light) region.

   \begin{figure}[!t]
    \begin{subfigure}{1\linewidth}
    \centering
        \includegraphics[trim=0.0cm 0.0cm 0.0cm 0.0cm, width=0.58\textwidth,height=0.28\textheight]{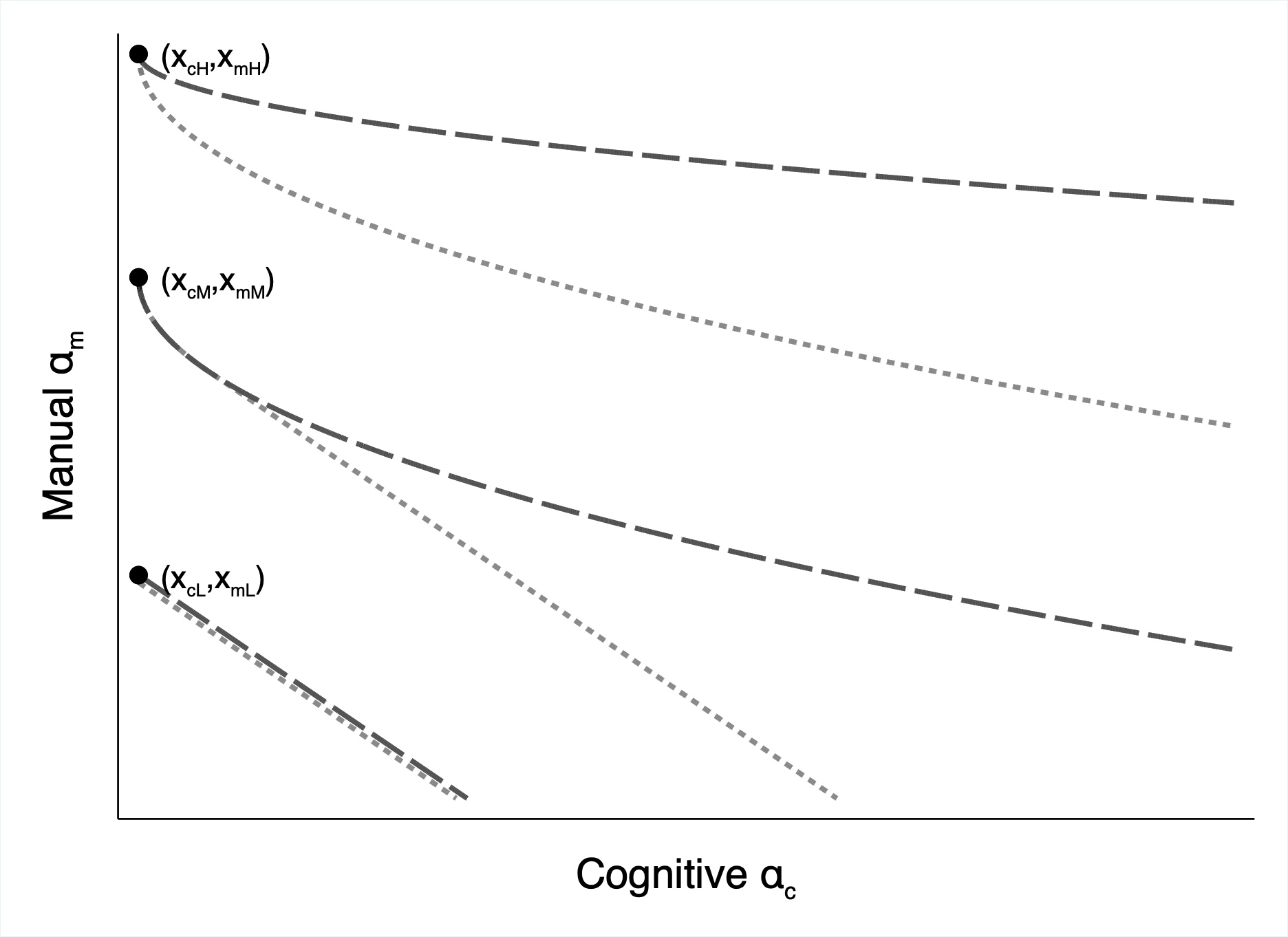}
    \end{subfigure}%
    \vspace{-0.2 cm}
    \caption{Illustration of Bunching} \label{f:bunching_ill}
{\scriptsize \vspace{.2 cm} \Cref{f:bunching_ill} illustrates the procedure to classify bunching using isocurves. The long-dashed lines represent isocurves for different cognitive task levels, while the short-dashed lines represent isocurves for different manual task levels.  Worker $\alpha$ is bunched with worker $\alpha'$ if the isocurves for $x(\alpha)$ intersect $\alpha'$.} 
    \end{figure}


The third approach to identify bunching is based directly on the definition of bunching (see Section \ref{s:bunching_theory}). When different worker types bunch, they are assigned identical task levels. Worker $\alpha$ is bunched if there are other workers who are assigned the same task levels $x(\alpha)$. Visually, we draw the isocurves corresponding to both $x_c(\alpha)$ and $x_m(\alpha)$ displayed on Figure \ref{f:bunching_ill} in the worker space $\alpha$ and assess whether the isocurves intersect for any other worker $\alpha'$. If there exists a worker $\alpha' \neq \alpha$ such that the isocurves intersect, then workers $\alpha$ and $\alpha'$ are bunched. 

\Cref{f:bunching_ill} gives an example of the procedure to classify optimal bunching using isocurves. The long-dash lines represent isocurves for different cognitive task levels, while the short-dash lines represent isocurves for different manual task levels. The dots indicate three hypothetical allocations. First, consider the isocurves corresponding to the allocation of crane operators in the top left corner. The long-dash isocurve represents workers with other combinations of skills $(\alpha_c, \alpha_m)$ who produce the same cognitive tasks $x_{cH}$ as the crane operator. The short-dash isocurve represents workers with other combinations of skills $(\alpha_c, \alpha_m)$ who produce the same manual tasks $x_{mH}$ as the crane operator. The lines intersect only at one point $-$ the skills of the crane operator at the top left corner. That is, no other worker $(\alpha_c, \alpha_m)$ receives the same task allocation $(x_{cH}, x_{mH})$ that is assigned to the crane operator. Next, consider the isocurves corresponding to the allocation of a gardener in the middle of \Cref{f:bunching_ill}. The long-dash isocurve (workers producing the same cognitive tasks $x_{cM}$ as the gardener) overlaps with the short-dash isocurve (workers producing the same manual tasks $x_{mM}$ as the gardener) for high levels of manual skill $\alpha_m$ and for low levels of cognitive skill $\alpha_c$. These workers $(\alpha_c, \alpha_m)$ produce the same cognitive and manual tasks $(x_{cM} , x_{mM})$ as the gardener, indicating that gardeners bunch with workers whose comparative advantage also lies in manual work. We call this type of bunching \textit{targeted bunching}. Finally, consider the bottom-left allocation corresponding to cashiers. In this case, the long-dash and short-dash isolines for cognitive and manual tasks overlap throughout the type space. All workers with skills $(\alpha_c, \alpha_m)$ on that line produce the same cognitive and manual tasks $(x_{cL} , x_{mL})$ as the cashier, despite their skill differences. We call this type of bunching \textit{blunt bunching}, where the planner does not distinguish different worker types when allocating tasks.

   \begin{figure}[t!]
        \begin{subfigure}{0.50\linewidth}
    \centering
        \includegraphics[trim=0.0cm 0.0cm 0.0cm 0.0cm, width=1\textwidth,height=0.26\textheight]{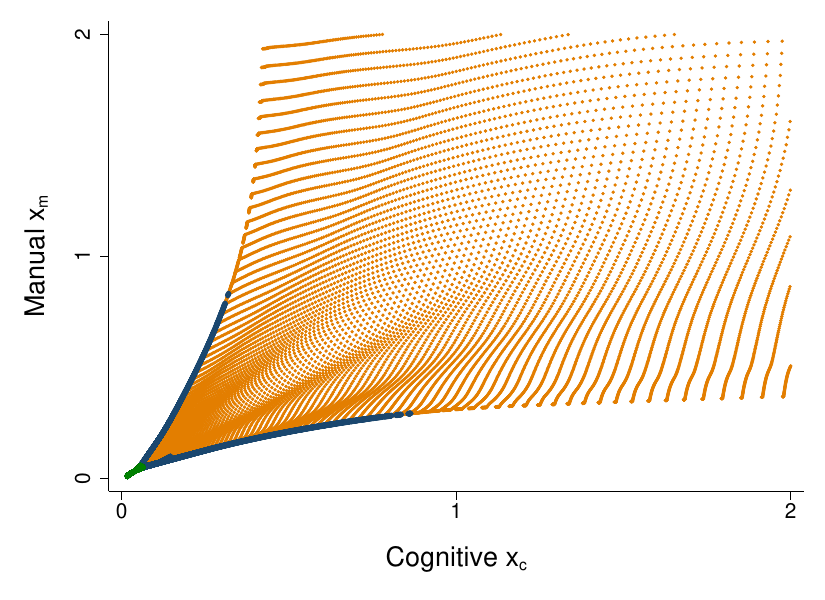}
    \end{subfigure}%
    \begin{subfigure}{0.50\linewidth}
    \centering
        \includegraphics[trim=0.0cm 0.0cm 0.0cm 0.0cm, width=1\textwidth,height=0.26\textheight]{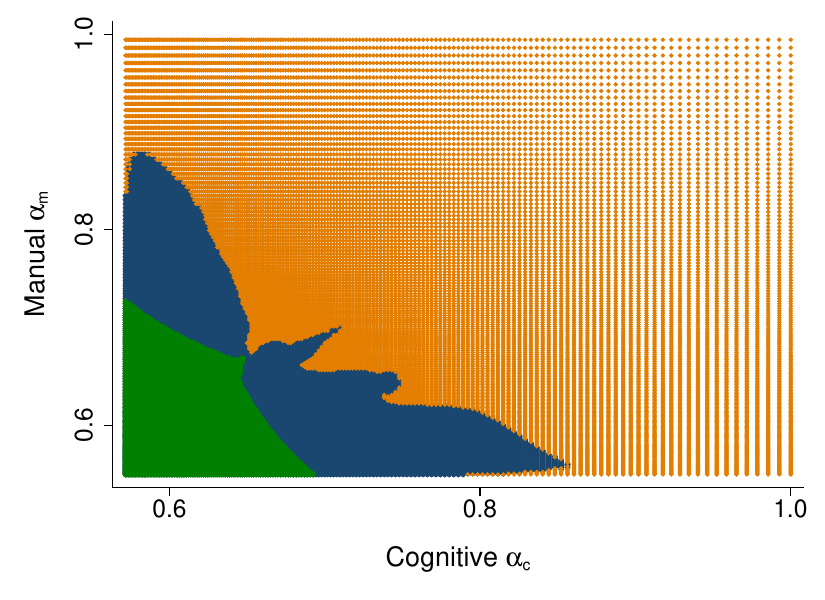}
    \end{subfigure}%
    \vspace{-0.2 cm}
    \caption{Optimal Targeted and Blunt Bunching} \label{f:data_bunching}
{\scriptsize \vspace{.2 cm} \Cref{f:data_bunching} shows bunching at the solution. The left panel demonstrates bunching in the allocation space by displaying combinations of optimal cognitive and manual tasks $(x_c , x_m)$. An allocation is marked in green or in blue if the allocation is assigned to more than one worker, while the allocation is marked orange if the allocation is assigned to one worker. The right panel displays bunching in the worker type space $(\alpha_c ,\alpha_m)$. In this figure, a worker type is marked in green or in blue if the worker's task allocation is also assigned to another worker. The green area indicates the blunt bunching region while the blue area indicates the targeted bunching region.}
    \end{figure}

\Cref{f:data_bunching} shows bunching at the optimal solution and presents two complementary views of the issue. The left panel displays bunching through the allocation of tasks. We display combinations of cognitive and manual tasks $(x_c , x_m)$ that are optimal. An allocation is marked green or blue if the allocation is assigned to more than one worker, while the allocation is marked orange if the allocation is assigned to one worker. There are two main regions of bunching in this picture. The first region is that of the blunt bunching and is given by the green line segment at the bottom of both cognitive and manual tasks. The second region is that of the targeted bunching and is represented by two blue line segments at the borders of the task trapezoid. That is, targeted bunching happens when the task input is low for only one task. The lower (flat) blue line segment represents low manual tasks, and the upper (vertical) line segment represents low cognitive tasks. Note that on these borders, when the task intensity becomes sufficiently high, the allocations are no longer bunched – that is, the blue line turns orange on the edges of the trapezoid. 

\vspace{0.4 cm}
\noindent \textbf{Blunt and Targeted Bunching}. In order to see which workers are bunched, the right panel shows bunched workers in the type space $(\alpha_c ,\alpha_m)$. In this figure, a worker type is marked in blue or green if the worker's task allocation is also assigned to another worker.\footnote{We construct the distance between two allocations $x(p)$ and $x(\hat{p})$ by considering the Euclidean distance between the allocations relative to the Euclidean distance between the respective types $p$ and $\hat{p}$. We classify two allocations to be identical if this ratio is below $10^{-4}$.} This panel shows that bunched workers have low cognitive or low manual skills (or both). Worker types are marked blue in the region of targeted bunching, and marked green in the region of blunt bunching. Workers with both low cognitive and manual skills are in the blunt bunching region. Workers with medium manual or medium cognitive skills are in the targeted bunching region when their skill set is asymmetric. Workers with medium cognitive skills bunch when their manual skills are low, and vice versa. In order to quantify the measure of bunched workers in the economy, we now combine the right panel with the worker type distribution of Figure \ref{f:skill_distribution}. At the optimum, 10.4 percent of the workers is bunched. The blunt bunching region comprises 30 percent of bunched workers. The targeted bunching region comprises 70 percent of bunched workers.
    
       \begin{figure}[!t]
    \begin{subfigure}{0.5\linewidth}
    \centering
        \includegraphics[trim=0.0cm 0.0cm 0.0cm 0.0cm, width=0.95\textwidth,height=0.28\textheight]{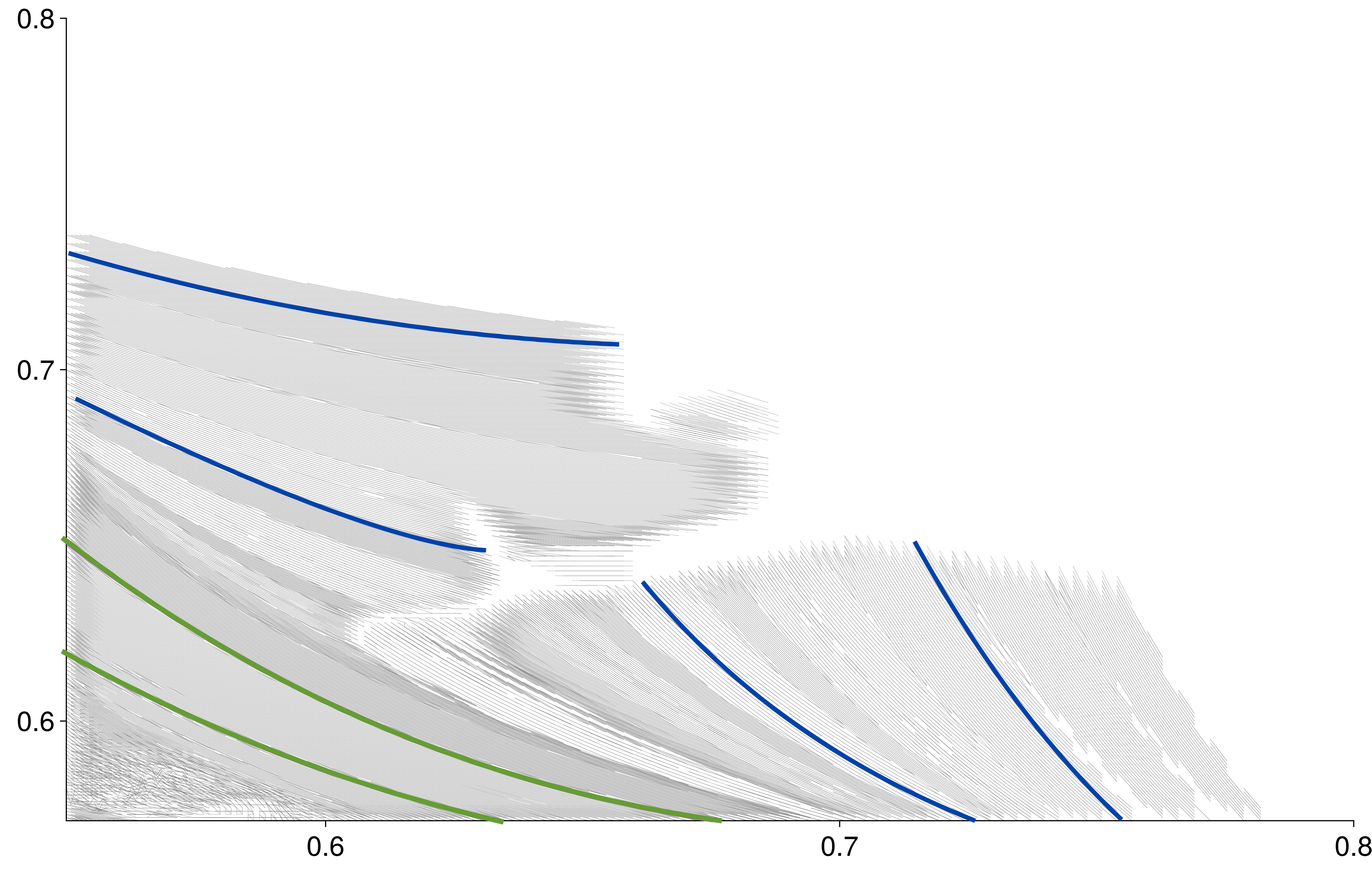}
    \end{subfigure}%
        \begin{subfigure}{0.5\linewidth}
    \centering
        \includegraphics[trim=0.0cm 0.0cm 0.0cm 0.0cm, width=0.95\textwidth,height=0.28\textheight]{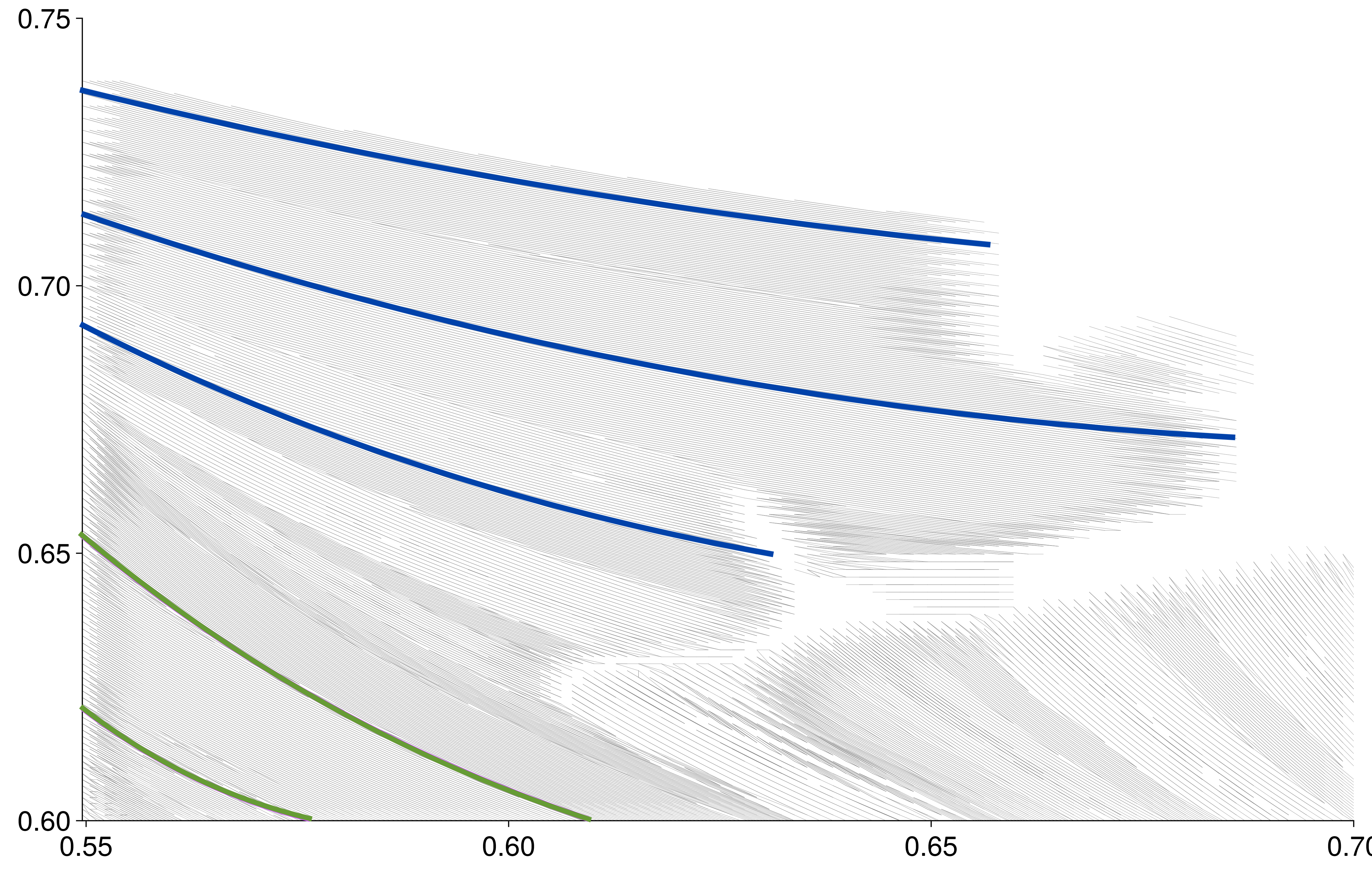}
    \end{subfigure}%
    \vspace{-0.1 cm}
    \caption{Bunching Patterns} \label{f:bunching_relation}
{\scriptsize \vspace{.2 cm} \Cref{f:bunching_relation} shows with whom workers bunch by connecting workers with a line in the worker type space if their allocations are bunched. Workers exclusively bunch with workers that are better in one skill dimension, but worse in another as represented by downward-sloping lines. Green lines indicate the blunt bunching region, blue lines indicate the targeted bunching region. Under blunt bunching, workers on the vertical boundary bunch with workers on the horizontal boundary, unlike under targeted bunching. The right-hand panel zooms in to distinguish the blunt bunching region from the targeted bunching region.}
    \end{figure}

Workers bunch with other workers both near and far. While Figure \ref{f:data_bunching} shows at what allocations workers are bunched and which workers are bunched, it does not show with whom workers bunch. These bunching relations are shown in \Cref{f:bunching_relation}, which connects two workers with a line in the type space if they are bunched.\footnote{To facilitate the presentation, we display the bunching relations for one quarter of all workers. In addition, we display at most two relations for each worker.}


Workers do not bunch with workers over whom they have an absolute advantage or who have an absolute advantage over them. Workers exclusively bunch with workers that are better in one skill dimension but worse in another. In \Cref{f:data_bunching} this is evident since all connections are represented by downward-sloping lines. Gardeners with somewhat better cognitive skills, but somewhat lower manual skills are bunched with gardeners with less cognitive skills and higher manual skills. Despite the slight difference in skills, the planner assigns both identical tasks. 

Within our bunching regions we distinguish two distinct patterns. In the lower triangle, which we indicate by green lines, the planner bunches together cognitive and manual tasks for all workers, ranging from those with relatively high manual skills to those with relatively high cognitive skills. Workers with the same effective skill index produce the same cognitive and manual tasks. Work is not tailored towards workers' specific skills but to an overall level of skill. In this blunt region, the planner optimally satisfies the incentive constraints by assigning identical allocations to different workers. Deviations are deterred bluntly at the cost of efficiency.

The planner also bunches workers in the targeted bunching regions, but in a less rudimentary way because the efficiency cost of bunching increases with worker skills. The planner separately bunches together cognitive and manual tasks for workers that are relatively skilled in manual tasks and similarly bunches together the cognitive and manual tasks for workers that are relatively skilled in cognitive tasks. Unlike in the lower triangle, however, workers that are relatively skilled in cognitive labor do not bunch with workers that are skilled in manual labor. In sum, the planner separates on worker's comparative advantage and bunches on worker's comparative disadvantage in the targeted bunching region.

In the regions without bunching, the planner distorts allocations to deter deviations by workers to allocations they find desirable similar to the unidimensional case. This is incentive provision through distinct distorted allocations. 

\vspace{0.4 cm}
\noindent \textbf{Taxation: Wedges and Tax Implementation}. We next study the implications for optimal taxes by studying the labor wedge (\ref{e:optimal_wedges}) for cognitive and manual tasks. Using linear preferences for consumption, the labor wedge is: 
\begin{equation}
1 - \tau_s = \frac{1}{2} \frac{x_s(\alpha)^{\rho-2}}{z(\alpha) \alpha_s^\rho} \qquad \implies \qquad \log ( 1 - \tau_s ) \propto - \log z + (\rho-2 ) \log  x_s - \rho \log \alpha_s  \label{e:optimal_wedges_quant} .
\end{equation}
The labor wedge is determined by the optimal assignment through $\log z$, by the allocation of tasks through $(\rho-2 ) \log  x_s$, and by worker skills through $\rho \log \alpha_s$. Workers at better firms face a higher labor wedge $\tau_s$. If the planner reduces task inputs, the labor wedge increases. Keeping allocations constant, the labor wedge increases with worker skills. 

  \begin{figure}[t!]
        \begin{subfigure}{0.50\linewidth}
    \centering
        \includegraphics[trim=0.0cm 0.0cm 0.0cm 0.0cm, width=1\textwidth,height=0.26\textheight]{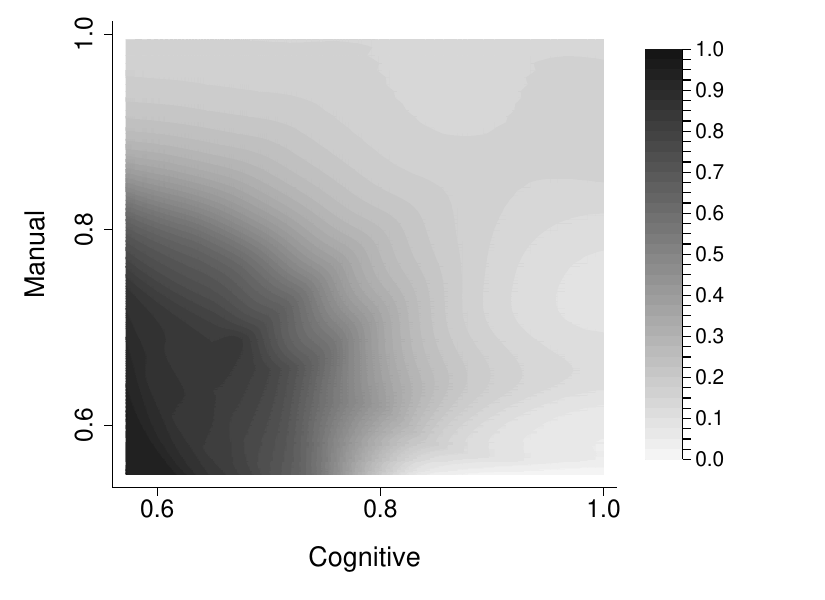}
    \end{subfigure}%
    \begin{subfigure}{0.50\linewidth}
    \centering
        \includegraphics[trim=0.0cm 0.0cm 0.0cm 0.0cm, width=1\textwidth,height=0.26\textheight]{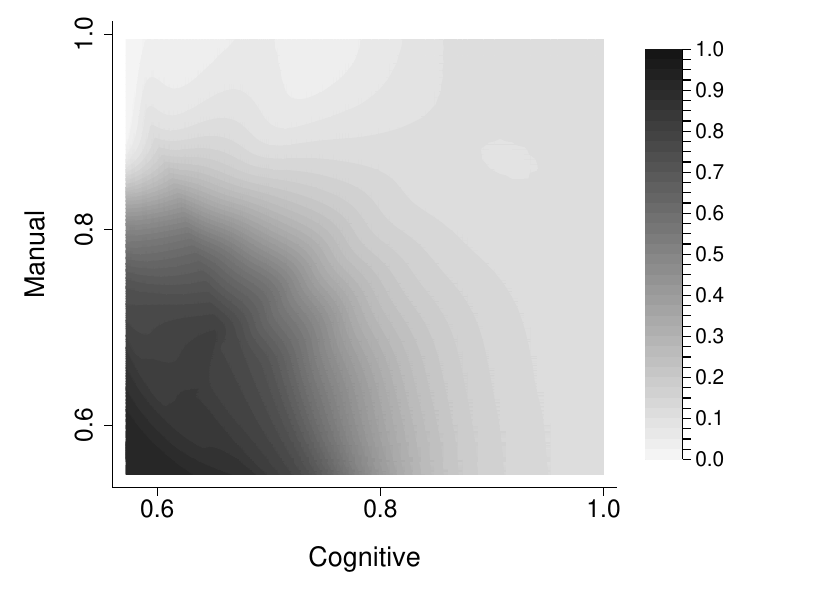}
    \end{subfigure}%
    \vspace{-0.2 cm}
    \caption{Tax Wedges} \label{f:data_taxes}
{\scriptsize \vspace{.2 cm} \Cref{f:data_taxes} visualizes the tax wedges for the planner solution. The left panel displays the manual tax wedge, the right panel displays the cognitive tax wedge.}
    \end{figure}

The optimal tax wedges are presented in Figure \ref{f:data_taxes}. First, the respective wedge is zero for the worker with the highest respective skill in either the cognitive or manual dimension. Consider the graph for the manual tax in the left panel. Workers with the highest manual skill are those represented by the top horizontal boundary of the graph. The manual tax on them is zero. Note that for the workers with the top cognitive skill, who are represented as the right vertical boundary, the manual tax is not zero. The manual tax on the best crane operator (highest manual and low cognitive) and physician (highest manual and high cognitive) is zero while it is positive on the top executive (low manual and highest cognitive).

The cognitive tax is displayed in the right panel of Figure \ref{f:data_taxes}. The cognitive tax for workers with the highest cognitive skill is zero. For workers with the top manual skill the cognitive tax is not zero. The cognitive tax on top executives (low manual and highest cognitive) and the best physician (highest manual and high cognitive) are zero while it is positive on the best crane operator (highest manual and low cognitive).

   \begin{figure}[!t]
    \begin{subfigure}{0.5\linewidth}
    \centering
        \includegraphics[trim=0.0cm 0.0cm 0.0cm 0.0cm, width=1\textwidth,height=0.29\textheight]{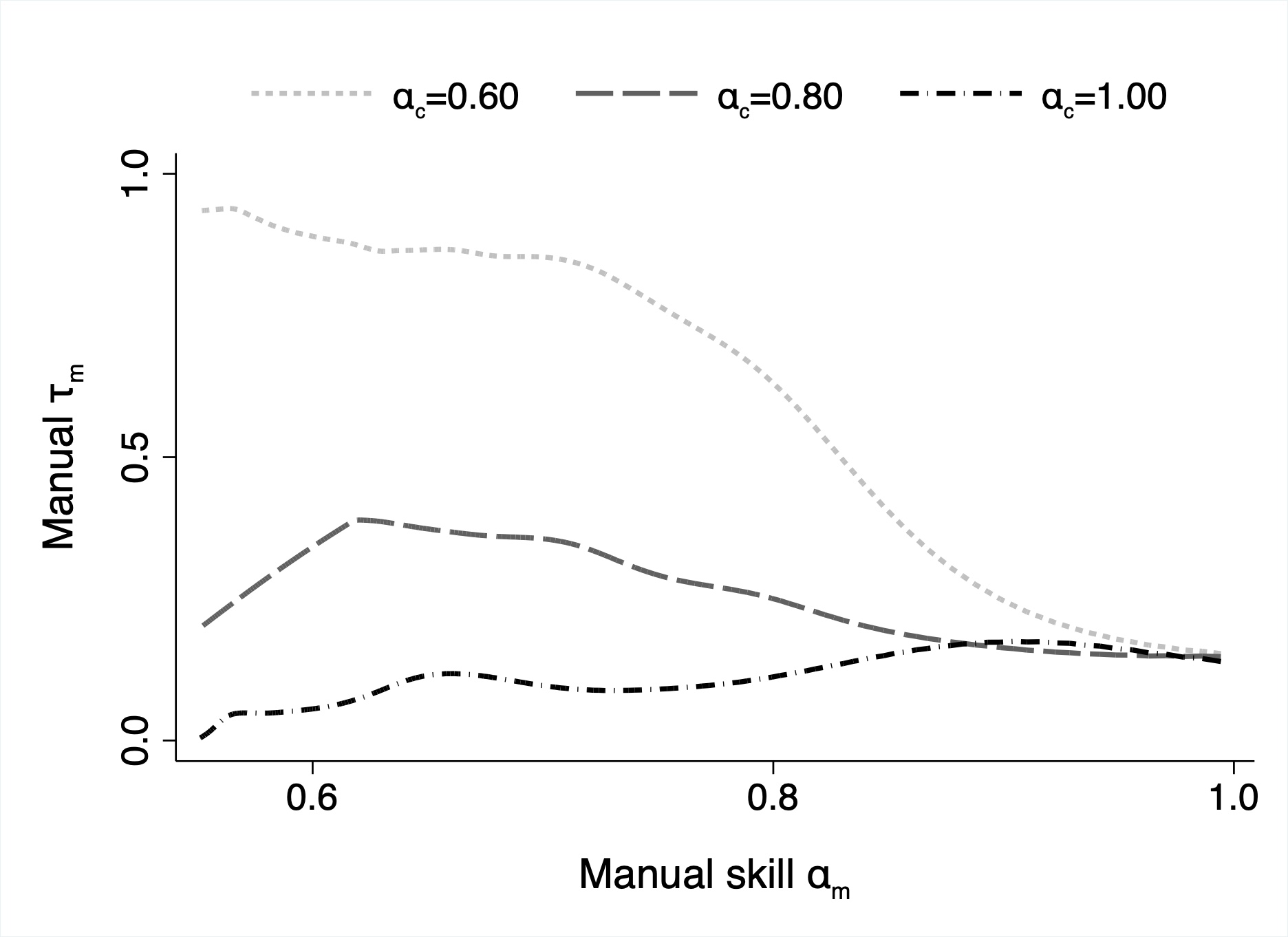}
    \end{subfigure}%
        \begin{subfigure}{0.5\linewidth}
    \centering
        \includegraphics[trim=0.0cm 0.0cm 0.0cm 0.0cm, width=1\textwidth,height=0.29\textheight]{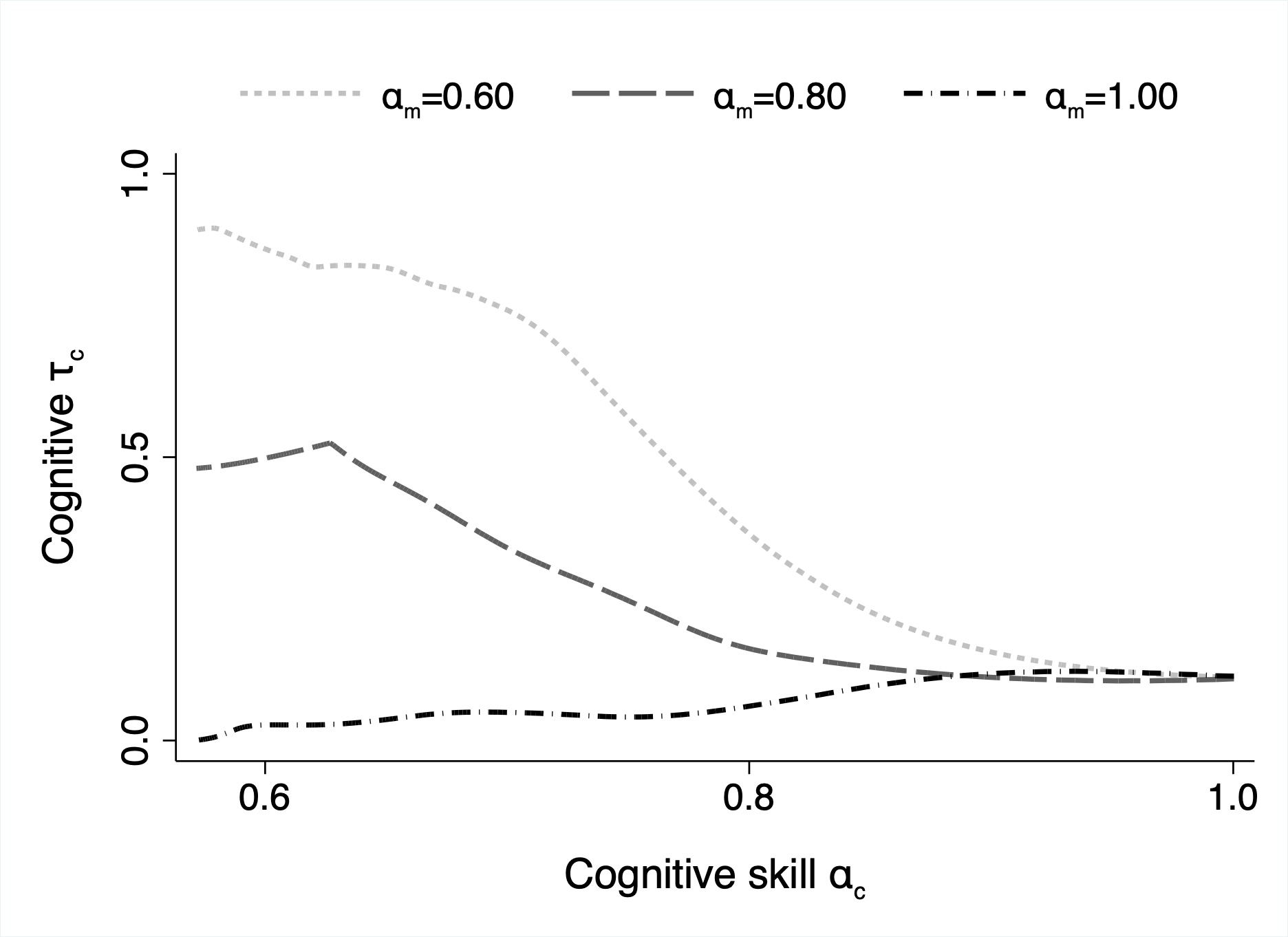}
    \end{subfigure}%
    \vspace{-0.2 cm}
    \caption{Manual and Cognitive Tax Wedge} \label{f:cognitive_wedge}
{\scriptsize \vspace{.2 cm} \Cref{f:cognitive_wedge} describes how the tax wedges vary with the level of worker skill. \Cref{f:cognitive_wedge} displays the manual wedge in the left panel, and the cognitive tax wedge in the right panel, for three groups of workers separated by their cognitive and manual skills, respectively.}
    \end{figure}

Second, we describe how taxes change with the level of the skill. \Cref{f:cognitive_wedge} plots the manual (on the left) and cognitive taxes (on the right) for three groups of workers separated by their skills corresponding to the heatmap in \Cref{f:data_taxes}. Consider top manual workers. Those are physicians (high cognitive), carpenters (medium cognitive), and crane operators (low cognitive). These workers are located in the top horizontal strip of \Cref{f:data_taxes} and represented by the dash-dotted line in the right panel of \Cref{f:cognitive_wedge}. These workers are not in bunched regions, and the cognitive tax on them is low. 

Consider medium manual workers. Those are lawyers (high cognitive), police officers (medium cognitive), and gardeners (low cognitive). These workers are located in the middle horizontal strip of \Cref{f:data_taxes} and are represented by the long-dash line in the right panel of \Cref{f:cognitive_wedge}. The cognitive tax on them is higher. Three forces give this shape from equation (\ref{e:optimal_wedges_quant}) $-$ the assignment to heterogeneous firms, the task allocation, and the worker skill. The assignment is monotonic: lawyers are assigned to better firms than police officers and gardeners. This force gives a higher labor wedge on lawyers compared to the gardeners. Conversely, the task allocation increases with skills leading to a reduction in labor wedges. Finally, the labor wedge increases in worker skills. The low cognitive skill workers (gardeners) face a higher cognitive labor wedge than the highest cognitive skill workers (lawyers) which is hence driven by the low amounts of cognitive tasks they conduct. The high level of tax distortion on gardeners is also driven by the fact that they are in the targeted bunching region while lawyers and police officers are not bunched. 


Consider low manual workers. Those are executives (high cognitive), teachers (medium cognitive), and cashiers (low cognitive). These workers are located in the bottom horizontal strip of \Cref{f:data_taxes} and are represented by the dotted line in the right panel of \Cref{f:cognitive_wedge}. The cognitive tax on them is higher than on medium manual workers. The cognitive tax is generally decreasing in skill with executives facing a lower marginal tax rate than cashiers. The high level of distortions on cashiers and teachers is also driven by the fact that they are bunched. The teachers are in the targeted bunching region – they have comparative advantage in their cognitive skills and are separated along the cognitive dimension but bunched in the manual dimension. The cashiers are in the blunt bunching region. They are bunched in both the manual and the cognitive dimensions and this leads to their allocation being heavily distorted. In the left panel of of \Cref{f:cognitive_wedge} we repeat the tax analysis for the manual tax skill wedge, using the left panel of \Cref{f:data_taxes}. 



In Appendix \ref{s:implementationtax}, we show that the planner allocation is implementable through a tax system over task inputs $x$. Since the effective skill index $X = x_c^2 + x_m^2$ determines worker wages $w(X)$ where $w$ is strictly increasing (Section \ref{s:positive}), there is a one-to-one mapping between the skill index $x_c^2 + x_m^2$ and earnings $w$. Therefore, there is a one-to-one map between $(w,x_m/x_c)$ and $(x_c,x_m)$. Since the optimum is implementable by a tax function $T(x_m,x_c)$ it is also implementable by a tax function that is only a function of income and line of work $\hat{T}(w,x_m/x_c)$.

\section{Conclusion}

We advance the understanding of optimal tax policy in multidimensional environments with bunching theoretically and quantitatively. Our optimal tax conditions generalize classic unidimensional optimal tax conditions of \citet{Mirrlees:1971}, \citet{Diamond:1998} and \citet{Saez:2001} to multidimensional taxation problems and account for global incentive constraints and bunching. For an empirically relevant model, we show that bunching is both substantial and nuanced and, hence, importantly impacts optimal policy design.

\clearpage


\vspace{-0.0cm}

{ {
\bibliographystyle{econometrica}	
\vspace{-0.0cm}
\baselineskip15.0pt
\bibliography{BTZ2_bib}
} }

\clearpage

\pagebreak

\renewcommand{\theequation}{A.\arabic{equation}} \setcounter{equation}{0}
\renewcommand{\thefigure}{A.\arabic{figure}}\setcounter{figure}{0}
\renewcommand{\thetable}{A.\arabic{table}}\setcounter{table}{0}
\setcounter{page}{1}
\newpage \appendix
\newpage

\begin{center}
{\Large Bunching and Taxing Multidimensional Skills \\}
\bigskip
{\Large Online Appendix \\}
\bigskip
{\large  Job Boerma, Aleh Tsyvinski and Alexander Zimin \\}
\bigskip
{\large March 2025}
\end{center}
\vspace{0.5cm}


\section{Proofs}

In this appendix, we formally prove the results in the main text.

\subsection{Proposition \ref{prop:p_assignment}} \label{pf:p_assignment}

To understand the optimal assignment, we consider a discrete version of the problem with 
identical discrete worker distributions $\{ x_{1s} \} = \{ x_{2s} \}$, which we denote by $F_x$, and a discrete firm distribution $\{z_{s}\} \sim F_z$, which we denote by $F_z$, for $s = \{1, \dots, n \}$. The discrete problem is to find an assignment $\gamma$ to maximize output:
\begin{equation}
\max\limits_{\gamma \in \underline{\Gamma}} \; \sum \gamma_{ijk} \hspace{0.05 cm}  y ( x_{1i},x_{2j},z_{k} ). \label{e:assignment_discrete}
\end{equation}
where $\gamma \in \underline{\Gamma} := \{ \gamma_{ijk} \geq 0 \hspace{0.08 cm} \big\vert  \hspace{0.04 cm} \sum_{jk} \gamma_{ijk} = 1, \sum_{ik} \gamma_{ijk} = 1, \sum_{ij} \gamma_{ijk} = 1 \}$. We next solve this problem in steps. 

First, we prove that without loss we can focus on assignments $\gamma$ that are symmetric in worker inputs, so that $\gamma_{ijk} = \gamma_{jik}$. Suppose a solution $\gamma$ is not symmetric, and use that the worker input samples are identical to define another feasible transport plan $\hat{\gamma}$ so that $\hat{\gamma}_{ijk} = \gamma_{jik}$ for all workers and projects. When $\gamma$ solves the assignment problem, so does $\hat{\gamma}$ because $\sum \hat{\gamma}_{ijk} y ( x_{1i},x_{2j},z_{k} ) = \sum \gamma_{jik} y ( x_{1i},x_{2j},z_{k} )  = \sum \gamma_{jik} y ( x_{2j},x_{1i},z_{k} ) = \sum \gamma_{ijk} y ( x_{1i},x_{2j},z_{k} )$, where the second equality follows as the production technology is symmetric in worker inputs, and the third equality follows by relabeling. This implies that the assignment $\frac{1}{2} \left( \gamma + \hat{\gamma} \right)$ is also a solution, which is indeed a symmetric solution. In summary, for every optimal assignment $\gamma$ there is a symmetric assignment $\frac{1}{2} \left( \gamma + \hat{\gamma} \right)$ that also solves the assignment problem. Without loss of generality we can therefore focus on assignments $\gamma$ that are symmetric in worker inputs.

Second, we prove it is optimal to self-sort workers so that $\gamma_{ijk} \neq 0$ implies that the workers are identical, or $i = j$. Consider an optimal symmetric assignment $\gamma$ and some project $z_k$, and denote the joint distribution of workers assigned to this project by $\gamma^k_{ij} := \gamma_{ijk}$. We construct the marginal distributions of workers and coworkers assigned to this project as $\mu^k_{1i} := \sum_j \gamma^k_{ij}$ and $\mu^k_{2j}  := \sum_i \gamma^k_{ij}$. Due to the symmetry of the assignment function $\gamma$, the worker and coworker distribution within the firm are identical, $\mu^k = \mu^k_1 = \mu^k_2$. Further, we let $\hat{\gamma}^k$ denote the optimal reassignment of workers and coworkers within the project:
\begin{equation}
\max\limits_{\gamma^k \in \tilde{\Gamma}(\mu,\mu)} \; \sum \gamma_{ij} \hspace{0.05 cm}  y \big( x_{1i},x_{2j} , z_{k} \big). \label{e:assignment_discrete_within}
\end{equation}
that is, $\hat{\gamma}^k$ solves the assignment problem within a firm given worker and coworker distribution $\mu$.

Within firms it is optimal to self-sort workers. Suppose some firm $z$ is assigned some worker and coworker distribution $\mu$, and consider two identical samples from this distribution. The within-firm assignment problem given these identical worker samples $\{x_{ij}\}$ for $i \in \{1,2\}$ and $j \in \mathcal{J} := \{1,\dots,J\}$ is to choose an assignment, or equivalently a permutation $\sigma$, to maximize output $z \sum\limits_{j\in\mathcal{J}} \hspace{0.05 cm} \big( x^j_{1c} x^{\sigma(j)}_{2c} + x^j_{1m} x^{\sigma(j)}_{2m} \big)$. Using the rearrangement inequality, aggregate output is bounded by:
\vspace{-0.3 cm}
\begin{equation}
\max\limits_{\sigma} \; z \sum\limits_{j\in\mathcal{J}} \hspace{0.05 cm} \big( x^j_{1c} x_{2c}^{\sigma(j)} + x_{1m}^{j} x_{2m}^{\sigma(j)} \big) \leq z \sum\limits_{j\in\mathcal{J}} \hspace{0.05 cm} \big( x^j_{1c} x^j_{2c} + x^j_{1m} x^j_{2m} \big) = z \sum\limits_{j\in\mathcal{J}} \hspace{0.05 cm} \big( (x^{j}_c)^{2} + (x^{j}_m)^2 \big) .\label{c_monotone}
\end{equation}
The final equality follows as the worker and coworker distributions are identical. We conclude that self-sorting within every project attains maximum production. The rearrangement inequality implies optimality of positively sorting the skills of workers within each firm as the production technology for each unidimensional task is supermodular as in \citet{Becker:1973}. In our environment with multidimensional skills positive sorting within each task is indeed attained by self-sorting, implying that $\hat{\gamma}^k$ is a diagonal matrix. 

To formally establish that the optimal assignment function features self-sorting for each firm, also with continuous marginal distributions $\mu$, we observe that (\ref{c_monotone}) implies that the self-sorting set $M \subset X \times X$ is $c$-monotone (see, e.g., \citet{Bogachev:2012} or \citet{Ambrosio:2013}). 

\vspace{0.1 cm}
\begin{definition*}
The set $M$ is \underline{$c$-monotone\vphantom{g}} if for all pairings $(x_{11},x_{21}),(x_{12},x_{22}),\dots,(x_{1n},x_{2n}) \in M$:
\begin{equation}
\sum\limits_{j\in\mathcal{J}} y(x_{1j},x_{2j}) \geq \sum\limits_{j\in\mathcal{J}}  y(x_{1j},x_{2\sigma(j)})
\end{equation}
for any permutation $\sigma$.
\end{definition*}
\vspace{0.1 cm}

\noindent The $c$-monotonicity condition directly implies the weaker condition that the matching set $M$ is $c$-cyclically monotone, or $\sum\limits_{j\in\mathcal{J}} y(x_{1j},x_{2j}) \geq \sum\limits_{j\in\mathcal{J}} y(x_{1j},x_{2j+1})$, where $x_{2n+1} = x_{21}$. The self-sorting assignment $\gamma_z$ with support on matching set $M$ is optimal as this statement is equivalent to the support of $\gamma_z$ being $c$-cyclically monotone following Theorem 1.2.7 in \citet{Bogachev:2012} or Theorem 1.13 in \citet{Ambrosio:2013}.

Given optimal self-sorting of workers within each firm, we construct a diagonal assignment $\hat{\gamma}$ by replacing $\gamma^k$ with the optimal self-sorted $\hat{\gamma}^k$ for every project $k$. Since $\hat{\gamma}^k$ solves the assignment problem within each firm, $\sum \hat{\gamma}_{ijk}   y ( x_{1i},x_{2j} , z_{k}) \geq \sum \gamma_{ijk} y ( x_{1i},x_{2j} , z_{k})$. By the construction of $\hat{\gamma}$, it holds that $\hat{\gamma}_{ijk} \neq 0$ implies that workers are identical $i = j$ for any $k$. Without loss of generality, an optimal assignment indeed features self-sorting of workers with coworkers within projects $z$. We define effective worker skills, or a team's quality, by $X := x^2_{c} + x^2_{m}$. 

Finally, the optimal assignment sorts the best teams with the most valuable firm projects. Since the optimal assignment of workers within firms is self-sorting, the Kantorovich problem (\ref{e:assignment_discrete}) simplifies to finding assignment $\gamma_{ik} \in \underline{\Gamma} := \{ \gamma_{ik} \geq 0 \hspace{0.08 cm} \big\vert \hspace{0.04 cm} \sum_{k} \gamma_{ik} = 1, \sum_{i} \gamma_{ik} = 1 \}$ to solve:
\begin{equation}
\max\limits_{\gamma \in \underline{\Gamma}} \; \sum \gamma_{ik} \hspace{0.05 cm}  y \big( X_{i},z_{k} \big), \label{e:assignment_discrete_reduced}
\end{equation}
that is, to assign teams to firms. Given that the reduced-form production technology is supermodular, the solution to this problem is a positive sorting between the team quality $X_i$ and the project value $z_k$. The solution to the original multimarginal Kantorovich problem (\ref{e:assignment_discrete}) is then constructed using $\gamma_{ijk} = \gamma_{ik} \delta_{ij}$, where $\delta$ is the Kronecker delta function.

While we constructed the solution to the Kantorovich formulation of the assignment problem in the main text, we observe that the optimal assignment to the discrete planning problem $\gamma_{ijk}$ is a Monge solution, meaning $\gamma_{ijk} \in \{ 0 , 1 \}$. This means the optimal assignment is a solution to the discrete planning problem of choosing permutations $\sigma_i$, to maximize output 
\begin{equation}
\max\limits_{\sigma_1,\sigma_2} \; \sum_s \hspace{0.05 cm}  y \big( x_{1\sigma_1(s)},x_{2\sigma_2(s)},z_{s} \big). 
\end{equation}
given the identical worker samples $\{ x_{is} \}$ drawn from the distribution $F_x$ and a firm project sample $\{z_{s}\}$ drawn from the distribution $F_z$ for $s = \{1, \dots, n \}$.

Before proceeding, we observe that a transport problem with two identical worker distributions $F_x$ with measure one for each role is equivalent to a transport problem with a single distribution of workers $F_x$ with measure two. Owing to the symmetry of workers' skills in production (\ref{e:firm_tech}), these are equivalent. Intuitively, any assignment for a problem with distinct worker distributions can be made symmetric. Consider assignment $\gamma$ that solves the discrete assignment problem for distinct worker distributions $F_{x_1}$ and $F_{x_2}$. The transpose of the assignment $\gamma$ along the worker input dimensions, which we denote by $\gamma'$, solves the discrete assignment problem with worker input distributions $F_{x_2}$ and $F_{x_1}$. This implies that symmetric assignment $\hat{\gamma} := \frac{\gamma + \gamma'}{2}$ solves the discrete assignment problem with worker distributions $F_x = \frac{1}{2} ( F_{x_1} + F_{x_2} )$.

\vspace{0.4 cm}
\noindent \textbf{Continuous Distributions}. To obtain the solution for continuous distributions of workers and coworkers, we extend our argument for the discrete distributions. We construct an assignment $\gamma$ that self-sorts workers and coworkers to obtain a unidimensional distribution for team quality $X$. The assignment $\gamma$ combines self-sorting of workers with positive sorting between worker skill index $X$ and projects $z$. 

This assignment $\gamma$ solves the Kantorovich problem (\ref{e:assignment}). To prove this claim, denote the support of the assignment by $M$, the matching set. Consider a collection of points within the matching set, $\{ (x_{1s},x_{2s},z_s) \} \in M$, then for each of those points it holds $x_{1s} = x_{2s}$, and that $z_s \leq z_{s'}$ implies $X_s \leq X_{s'}$. Since the support is constructed by using a Monge solution for the discrete assignment problem, $\sum_s \hspace{0.05 cm}  y ( x_{1\sigma_1(s)},x_{2\sigma_2(s)},z_{s} ) \leq \sum_s \hspace{0.05 cm}  y ( x_{1s},x_{2s},z_{s} )$ for all permutations $\sigma_1,\sigma_2$. Equivalently, the matching set $M$ is $c$-monotone. By Theorem 1.2 in \citet{Griessler:2018}, the assignment $\gamma$ solves the Kantorovich problem. This concludes the proof to Proposition \ref{prop:p_assignment}.

\subsection{Incentive Constraints}\label{a:ic}

\noindent We first discuss the differentiability of the indirect utility function $u$. Feasibility of an allocation is exactly equivalent to the indirect utility function $u$ being feasible, which means that $u$ is convex, nonnegative, and decreasing, together with the additional constraint $-x(p) \in \partial u(p)$, where $\partial u$ denotes the subdifferential of $u$ at a given point. 

Consider the incentive constraint (\ref{e:linear_ic}):
\begin{equation*}
c(p) - p_c x_c(p) - p_m x_m(p) \ge c(p') - p_c x_c(p') - p_m x_m(p')
\end{equation*}
Using the notation for the indirect utility function (\ref{e:indirect_utility}), that is, $u(p) = c(p) - p_c x_c(p) - p_m x_m(p)$, this can be equivalently rewritten as
\begin{equation*}
u(p) - u(p') \ge \langle p - p', -x(p') \rangle.
\end{equation*}
This inequality for each pair of types $(p, p')$ is equivalent to the convexity of $u$ together with the constraint that $-x(p) \in \partial u(p)$, where $\partial u$ denotes the subdifferential of $u$ at given point. In this case, we know by Alexandrov theorem that $u$, since it is a convex function, is differentiable almost everywhere, and hence $\partial u(p) = \nabla u(p)$ almost everywhere.

\vspace{0.4 cm}
\noindent We next show which incentive compatibility constraints are redundant to the planner problem in the numerical analysis. We establish that every reducible incentive constraint is redundant in the presence of the irreducible constraints, which shrinks the set of incentive constraints that needs to be taken into account by the planner. To show this result, we let $L \subseteq \mathbb{R}^2$ be a finite subset and we first define irreducible constraints. 

\begin{figure}[!t]
                \begin{subfigure}{\textwidth}
                                \begin{center}
                                                \begin{tikzpicture}
                                                \filldraw(0.0,0.0) circle[radius=2.5pt];
                                                \filldraw(5.0,0.0) circle[radius=2.5pt];
                                                \filldraw(10.0,0.0) circle[radius=2.5pt];                                              
                                                \filldraw(0.0,3.0) circle[radius=2.5pt];
                                                \filldraw(5.0,3.0) circle[radius=2.5pt];
                                                \filldraw(10.0,3.0) circle[radius=2.5pt];
                                                \filldraw(0.0,6.0) circle[radius=2.5pt];
                                                \filldraw(5.0,6.0) circle[radius=2.5pt];
                                                \filldraw(10.0,6.0) circle[radius=2.5pt];                                                
                                                \draw (0.0,0.0) node[below] {A};
                                                \draw (0.0,6.0) node[above] {D};
                                                \draw (10.0,0.0) node[below] {E};
                                                \draw (5.0,3.0) node[above] {B};
                                                \draw (10.0,6.0) node[above] {C};
                                                \draw[dotted,line width=0.8mm,color=orange] (0.0,0.0) parabola [bend at start] (10,6);
\draw[decoration={markings, mark=at position (0.58) with  {\arrow[black,line width=0.7mm]{>}};},postaction=decorate,line width=0.3mm] (0.0,3.0) -- (0.0,0.0);
\draw[decoration={markings, mark=at position (0.42) with  {\arrow[black,line width=0.7mm]{<}};},postaction=decorate,line width=0.3mm] (0.0,3.0) -- (0.0,0.0);
\draw[decoration={markings, mark=at position (0.55) with  {\arrow[black,line width=0.7mm]{>}};},postaction=decorate,line width=0.3mm] (5.0,0.0) -- (0.0,0.0);
\draw[decoration={markings, mark=at position (0.45) with  {\arrow[black,line width=0.7mm]{<}};},postaction=decorate,line width=0.3mm] (5.0,0.0) -- (0.0,0.0);
\draw[decoration={markings, mark=at position (0.53) with  {\arrow[black,line width=0.7mm]{>}};},postaction=decorate,line width=0.3mm] (5.0,6.0) -- (0.0,0.0);
\draw[decoration={markings, mark=at position (0.46) with  {\arrow[black,line width=0.7mm]{<}};},postaction=decorate,line width=0.3mm] (5.0,6.0) -- (0.0,0.0);
\draw[decoration={markings, mark=at position (0.54) with  {\arrow[blue,line width=0.7mm]{>}};},postaction=decorate,line width=0.3mm, blue] (5.0,3.0) -- (10.0,6.0);
\draw[decoration={markings, mark=at position (0.46) with  {\arrow[blue,line width=0.7mm]{<}};},postaction=decorate,line width=0.3mm, blue] (5.0,3.0) -- (10.0,6.0);
\draw[decoration={markings, mark=at position (0.28) with  {\arrow[black,line width=0.7mm]{>}};},postaction=decorate,line width=0.3mm] (10.0,3.0) -- (0.0,0.0);
\draw[decoration={markings, mark=at position (0.23) with  {\arrow[black,line width=0.7mm]{<}};},postaction=decorate,line width=0.3mm] (10.0,3.0) -- (0.0,0.0);
\draw[decoration={markings, mark=at position (0.54) with  {\arrow[black,line width=0.7mm]{>}};},postaction=decorate,line width=0.3mm] (5.0,3.0) -- (0.0,0.0);
\draw[decoration={markings, mark=at position (0.46) with  {\arrow[black,line width=0.7mm]{<}};},postaction=decorate,line width=0.3mm] (5.0,3.0) -- (0.0,0.0);
                                   	      \end{tikzpicture}
                                \end{center}
                \end{subfigure}%
    		\caption{Reducible and Irreducible Incentive Constraints} \label{f:ic}
{\scriptsize \vspace{.2 cm} \Cref{f:ic} shows reducible and irreducible incentive constraints for worker $A$. When the irreducible incentive constraints between workers $A$ and $B$ are satisfied (as indicated by the black solid line between workers $A$ and $B$), and the irreducible incentive constraints between workers $B$ and $C$ are satisfied (as indicated by the blue solid line between workers $B$ and $C$), then reducible incentive constraints between $A$ and $C$ are satisfied (as indicated by the orange dashed line). Every reducible incentive constraint is satisfied when the irreducible constraints are. The black solid lines denote all the irreducible incentive constraints for worker $A$.}

\end{figure}
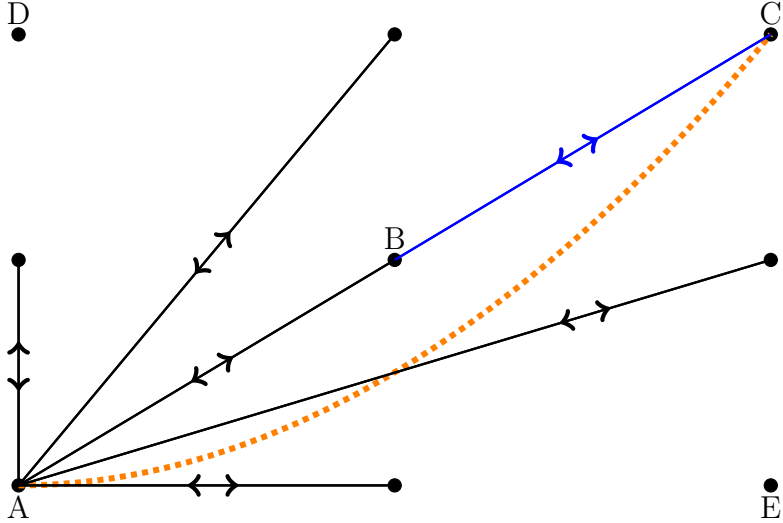

\vspace{0.1 cm}
\begin{definition*}
A couple of points $(p,q) \subseteq L$ is \underline{irreducible\vphantom{g}} if there is no point $m\in L$ on the interval between $p$ and $q$. 
\end{definition*}

\begin{lemma}\label{prop:reducible_ics}
All reducible incentive constraints are implied by irreducible incentive constraints. 
\end{lemma}

\noindent The proof is presented below.

\Cref{f:ic} shows the reducible and the irreducible incentive constraints for worker $A$. When the irreducible incentive constraints between workers $A$ and $B$ are satisfied (as indicated by the black solid line between workers $A$ and $B$), and the irreducible incentive constraints between workers $B$ and $C$ are satisfied (as indicated by the blue solid line between workers $B$ and $C$), then reducible incentive constraints between $A$ and $C$ are satisfied (as indicated by the orange dashed line). Every reducible incentive constraint is satisfied when the irreducible constraints are. The black solid lines denote the irreducible incentive constraints for worker $A$. The incentive constraints between workers $A$ and $D$ as well as between workers $A$ and $E$ are also reducible.


\vspace{0.4 cm}
\noindent We next establish that no other incentive constraints can be eliminated a priori. The set of feasible allocations strictly increases by removing any of the irreducible incentive constraint.

\vspace{0.05 cm}
\begin{lemma} \label{prop:minimal_ics}
Consider any irreducible incentive constraint where worker type $p_0$ does not want to report to be of type $q_0$. If we eliminate such an incentive constraint, then there exists an allocation that satisfies all other incentive constraints while worker type $p_0$ wants to report $q_0$. That is, for any irreducible pair $(p_0,q_0) \subseteq L$ there exist functions $(u,x)$ such that:
\begin{equation*}
u(p) - u(q) \geq \langle p - q, - x (q) \rangle 
\end{equation*}
for all $(p,q) \in L$ where $(p,q) \neq (p_0,q_0)$ and 
\begin{equation*}
u(p_0) - u(q_0) < \langle p_0 - q_0, - x (q_0) \rangle .
\end{equation*}
\end{lemma}

\noindent The proof is presented below. 

We denote the set of utility allocations that satisfy the set of irreducible linear incentive constraints by $\mathcal{I}$. 

\vspace{0.4 cm}
\noindent \textbf{Proof to Lemma \ref{prop:reducible_ics}}. Let $v$ be a ray from a parameter point $p = (p_c,p_m)$, and let scalar parameters $\lambda$ and $\beta$ such that $0 < \lambda < \beta$. We consider points $p + \lambda v$ and $p + \beta v$.

We first show that if incentive constraints between points $p$ and $p + \lambda v$ as well as $p + \lambda v$ and $p + \beta v$ are satisfied, then incentive constraints between $p$ and $p + \beta v$ are implied. By (\ref{e:linear_ic2}), we know the incentive constraint that $p$ does not want to report $q = p + \lambda v$ implies:
\begin{equation}
u(p) - u(p + \lambda v) \geq - \lambda \langle v, \partial u(p + \lambda v) \rangle = \lambda \langle v, x(p + \lambda v) \rangle .
\end{equation}
Similarly, the incentive constraint that $q = p + \lambda v$ does not want to report $p$ implies:
\begin{equation*}
u( p + \lambda v) - u(p) \geq  \lambda \langle  v , \partial u(p) \rangle = - \lambda \langle  v , x (p) \rangle
\end{equation*}
Adding these two constraints, we obtain $\langle  v , x (p) \rangle \geq \langle v, x(p + \lambda v) \rangle$.

Using the incentive constraint that $p + \beta v$ does not want to report $p + \lambda v$:
\begin{equation}
u(p + \beta v) - u(p + \lambda v) \geq - (\beta - \lambda) \langle v , x(p + \lambda v) \rangle , \label{e:IC1}
\end{equation}
we show that given these incentive constraints, the constraint between points $p + \beta v$ and $p$ is implied. We evaluate 
\begin{align*}
u(p + \beta v) - u(p) & = \big( u(p + \beta v) - u(p + \lambda v) \big) + \big( u(p + \lambda v) - u(p) \big) \\
& \geq - (\beta - \lambda) \langle v , x(p + \lambda v) \rangle - \lambda \langle  v , x (p) \rangle \\
& \geq - (\beta - \lambda) \langle v , x(p) \rangle - \lambda \langle  v , x (p) \rangle = - \beta \langle v , x(p) \rangle
\end{align*}
where the first inequality follows from the first and third incentive constraint, while the second inequality follows as $\langle  v , x (p) \rangle \geq \langle v, x(p + \lambda v) \rangle$. The final equality indeed implies that $p + \beta v$ does not want to report $p$.

Similarly, we use the incentive constraint that $p + \lambda v$ does not want to report $p + \beta v$:
\begin{equation}
u(p + \lambda v) - u(p + \beta v) \geq - (\lambda - \beta) \langle v, x(p + \beta v) \rangle \label{e:IC2}
\end{equation}
in order to prove that $p$ does not want to report $p + \beta v$:
\begin{align*}
u(p) - u(p + \beta v) & = \big( u(p) - u(p + \lambda v) \big) + \big( u(p + \lambda v) - u(p + \beta v) \big) \\
& \geq \lambda \langle v, x(p + \lambda v) \rangle - (\lambda - \beta) \langle v, x(p + \beta v) \rangle \\
& \geq \lambda \langle v, x(p + \beta v) \rangle - (\lambda - \beta) \langle v, x(p + \beta v) \rangle = \beta \langle v, x(p + \beta v) \rangle
\end{align*}
where the final inequality follows by adding (\ref{e:IC1}) and (\ref{e:IC2}), which implies $\langle v, x(p + \lambda v) \rangle \geq \langle v, x(p + \beta v) \rangle$. This shows we do not need to incorporate incentive constraints between $p$ and $p + \beta v$ when we incorporate the incentive constraint between $p$ and $p + \lambda v$, and between $p + \lambda v$ and $p + \beta v$.

The final step is that our result so far held for general scalar parameters $\lambda$ and $\beta$ be such that $0 < \lambda < \beta$. We note that for $\underline{\lambda}$ so that $0 < \underline{\lambda} < \lambda < \beta$, we can show that the constraints between $p$ and $p + \lambda v$ are implied by the constraints between $p$ and $p + \underline{\lambda} v$ as well as $p  + \underline{\lambda} v$ and $p + \lambda v$. Hence, for every point $p$ we only need to consider the constraints for the lowest possible values for $\lambda$. These constraints are irreducible.

\vspace{0.4 cm}
\noindent \textbf{Proof to Lemma \ref{prop:minimal_ics}}. By induction. The induction base is a set of points in $L$ with $|L| = 2$. Let the points within the set be $p_0$ and $q_0$. We show there exist functions $(u,x)$ so that: 
\begin{align*}
u(q_0) - u(p_0) & \geq \langle q_0 - p_0, - x (p_0) \rangle \\
u(p_0) - u(q_0) & < \langle p_0 - q_0, - x (q_0) \rangle .
\end{align*}
Construct the function $x(p_0) = x(q_0) = 0$, and $u(p_0) = 0$ and $u(q_0) = 1$.

Induction step for $|L| = n + 1$ points. Let $z$ denote a vertex of the convex hull of set $L$ which is neither $p_0$ nor $q_0$. Such a point indeed exists, else the convex hull is an interval between $p_0$ and $q_0$, implying that any other point of the set $L$ would be a point between $p_0$ and $q_0$ contradicting $(p_0,q_0)$ is irreducible.

Remove the point $z$ from the set $L$. By induction step at $n$, there exist functions $(u,x)$ for the lattice $L\setminus \{ z \}$ such that for the same irreducible pair $(p_0,q_0) \subseteq L \setminus \{ z \}$: 
\begin{equation*}
u(p) - u(q) \geq \langle p - q, - x (q) \rangle 
\end{equation*}
for all $(p,q) \in L \setminus \{ z \}$ and $(p,q) \neq (p_0,q_0)$ and 
\begin{equation*}
u(p_0) - u(q_0) < \langle p_0 - q_0, - x (q_0) \rangle .
\end{equation*}

We need to extend the functions $u$ and $x$ onto the point $z$. Here, we will use that $z$ is a vertex of the convex hull. We construct the value for the functions $u$ and $x$ at point $z$. At the point $z$, we require: 
\begin{align*}
u(z) - u(p) & \geq \langle z - p, - x (p) \rangle \\
u(p) - u(z) & \geq \langle p - z, - x (z) \rangle 
\end{align*}
for all $p \in L \setminus \{ z \}$. Reorganizing, the first inequality becomes 
\begin{equation*}
u(z) \geq \max_{p \neq z} \; \big\{ u(p) + \langle z - p, - x (p) \rangle \big\} =: \mathrm{C}
\end{equation*}
where we observe constant $\mathrm{C}$ is independent of both $u(z)$ and $x(z)$. We set $u(z) = \mathrm{C}$. As a result, the second inequality is written as: 
\begin{equation*}
u(p) - \mathrm{C} \geq \langle p - z, - x (z) \rangle .
\end{equation*}


To show that the inequality is satisfied, we use the following variation of Farkas' Lemma. For any convex polytope $\mathrm{P}$ and for any vertex $v$ of this polytope, there exist a hyperplane such that $v$ belongs to the hyperplane while all other points of the convex polytope $\mathrm{P}$ lie strictly on one side of it. Equivalently, there exists a vector $h$ such that $\langle x - v, h \rangle < 0$ for all $x \in \mathrm{P} \setminus \{ v \}$. 

Since point $z$ is a vertex of the convex hull of $L$ there exists $\tilde{x}(z)$ so that $\langle p - z, - \tilde{x}(z) \rangle < 0$ for every $p \in L \setminus \{ z \}$. Define the constant $\mathrm{C}_p = \langle p - z, - \tilde{x}(z) \rangle < 0$. Then there exists positive value $t_p > 0$ so that: 
\begin{equation*}
\langle p - z, - t_p \tilde{x}(z) \rangle = t_p \mathrm{C}_p < u(p) - \mathrm{C}.
\end{equation*}
Further, let $t = \max\limits_{p \in L \setminus \{ z \}} t_p$, implying that $\langle p - z , - t \tilde{x}(z) \rangle = t \mathrm{C}_p \leq t_p \mathrm{C}_p < u(p) - \mathrm{C}$ for all $p \in L \setminus \{ z \}$. Hence, we set $x(z) = t \tilde{x}(z)$ to conclude our claim. 

\subsection{Lemma \ref{l:strconvex}} \label{pf:strconvex}

We first prove that if the indirect utility function (\ref{e:indirect_utility}) is strongly convex, then there is no bunching. By the inverse function theorem, using that the indirect utility function $u$ is twice continuously differentiable, if the Jacobian matrix of the mapping from $p$ to $x$ is invertible then the labor task allocation is invertible. The Jacobian matrix of the mapping from $p$ to $x$ is the negative to the Hessian matrix of the indirect utility function $ 
\Big( \begin{smallmatrix} \partial x_{c} / \partial p_c & \partial x_{c} / \partial p_m \\ \; \partial x_{m} / \partial p_c \;\; & \;\; \partial x_{m} / \partial p_m \; \end{smallmatrix}\Big)$ using $x(p) = - \nabla u(p)$. Since the utility function $u$ is strongly convex for worker $p$, its Hessian matrix is invertible, and hence the Jacobian matrix is. Summarizing, if the utility function is strongly convex, then there is no bunching. 


Now we prove that if the indirect utility function (\ref{e:indirect_utility}) is not strongly convex, i.e. the Hessian matrix is degenerate for all workers in the neighborhood of $p$, worker $p$ is bunched, or $p \in \mathcal{B}$. To prove this statement, consider a mapping $f$ from worker type $p$ to the labor allocation $x$, and let $\mathcal{P}$ denote the neighborhood of workers around $p$ such that Hessian matrix $H(u)$ is degenerate for all workers $p \in \mathcal{P}$.  Since the Jacobian matrix of the mapping $f$ is the negative to the Hessian matrix of the indirect utility function, the Jacobian matrix is degenerate for all workers $p \in \mathcal{P}$. Equivalently, $\mathcal{P}$ is a critical set. By Sard's theorem it follows that the image $f(\mathcal{P})$ has Lebesgue measure zero. 

To prove that worker $p$ is bunched, suppose by contradiction they are not, $p \notin \mathcal{B}$. Equivalently, the mapping $f$ is injective in a neighborhood $\hat{\mathcal{P}} \subseteq \mathcal{P}$. By the invariance of domain, the image $f(\hat{\mathcal{P}})$ is a non-empty open set. This implies that the Lebesgue measure is strictly positive for the image, contradicting the implication from Sard's theorem. Thus, workers are bunched when the optimality condition does not hold.


\subsection{Proposition \ref{p:implementable}} \label{pf:implementable}

To prove the proposition, we prove Lemma \ref{lemma:sd_basic} and Lemma \ref{lemma:sd2_basic}.

\vspace{0.20 cm}
\begin{lemma}\label{lemma:sd_basic}
Let $( c,x )$ solve the planner problem. The following condition holds with equality at an optimum:
\begin{equation}
\int \big( \mathcal{C}'( c ) c + z \big( \mathcal{X}' ( x_c ) x_c + \mathcal{X}' ( x_m ) x_m \big) \big) \pi  \text{d}p  = \lambda \int \big( c  - p_c x_c -p_m x_m \big) \pi \text{d}p \label{e:stoch_dominance_basic}.
\end{equation}
\end{lemma}

\vspace{0.2 cm}
\begin{proof}
Consider an allocation $(c,x)$ that satisfies the incentive compatibility constraints. Consider multiplying this allocation by a constant factor $\zeta > 0$, to obtain the scaled allocation $\zeta (c,x)$. The Lagrangian of the scaled allocation exceeds the Lagrangian of the optimal allocation $(c,x)$, or $\mathcal{L}( c,x ) \leq \mathcal{L}( \zeta (c,x))$. Therefore, we can consider a variation around the optimal allocation $(c,x)$, where we scale the allocation by a small factor $\varepsilon$, so that alternative allocation $(c + \varepsilon c, x +\varepsilon x)$ is feasible. Given such a variation, the implied change in the resource cost is:
\begin{equation} 
\Delta = \varepsilon \Big( \int \big( \mathcal{C}'( c ) c + z \big( \mathcal{X}' ( x_c ) x_c + \mathcal{X}' ( x_m ) x_m \big) \big) \pi  \text{d}p  - \lambda \int \big( c - p_c x_c - p_m x_m \big) \pi \text{d}p \Big) + o(\varepsilon).
\end{equation}
At an optimum, neither a positive $(\varepsilon > 0)$ nor a negative $(\varepsilon < 0)$ small variation decreases the cost of resources, so $\Delta = o(\varepsilon)$, which establishes (\ref{e:stoch_dominance_basic}).\end{proof}

\begin{lemma}\label{lemma:sd2_basic}
Let $( c,x )$ solve the planning problem, then the implementability condition (\ref{e:stoch_dominance_weak_basic1}) holds for any feasible allocation $(\hat{c},\hat{x}) \in \mathcal{I}$.
\end{lemma}

\begin{proof} If two allocations $(c,x)$ and $(\hat{c},\hat{x})$ satisfy the incentive constraints, so does their convex combination $(\tilde{c},\tilde{x}) = (1 - \varepsilon) (c,x) +\varepsilon (\hat{c},\hat{x})$ for any $\varepsilon \in (0,1)$. If $(c,x)$ is a planner solution, it follows that for any $\varepsilon \in (0,1)$ and for any feasible allocation $(\hat{c},\hat{x})$, a convex combination of the alternative allocation and the solution increases the Langrangian value relative to its optimum, $\mathcal{L}(\tilde{c},\tilde{x}) - \mathcal{L}( c,x ) \geq 0$. By construction of the convex combination, this is equivalent to $\mathcal{L}\big((c,x) + \varepsilon \big((\hat{c},\hat{x}) - (c,x) \big) \big) - \mathcal{L}\big( c,x \big) \geq 0$. 

To further develop this, note that $\mathcal{L}((c,x) + \varepsilon ((\hat{c},\hat{x}) - (c,x) ) ) - \mathcal{L}( c,x ) = \varepsilon ( \int ( \mathcal{C}'( c ) (\hat{c} - c) + z \sum \mathcal{X}' ( x_s ) (\hat{x}_s - x_s) ) \pi  \text{d}p  - \lambda \int ( (\hat{c} - c) - \sum p_s (\hat{x}_s - x_s) ) \pi \text{d}p \big) + o(\varepsilon) = \varepsilon ( \int ( \mathcal{C}'( c )\hat{c} + z \sum \mathcal{X}' ( x_s ) \hat{x}_s ) \pi  \text{d}p  - \lambda \int ( \hat{c} - \sum p_s \hat{x}_s ) \pi \text{d}p ) + o(\varepsilon)$, where the final equality follows by the optimality condition in Lemma \ref{lemma:sd_basic}. Equation (\ref{e:stoch_dominance_weak_basic1}) follows because for any $\varepsilon \in (0,1)$ the previous condition is positive, that is:
\begin{equation}
\int \big( \mathcal{C}'( c ) \hat{c} + z \big( \mathcal{X}' ( x_c ) \hat{x}_c + \mathcal{X}' ( x_m ) \hat{x}_m \big) \big) \pi \text{d}p  \geq \lambda \int \big( \hat{c} - p_c \hat{x}_c - p_m \hat{x}_m \big) \pi \text{d}p \tag{\ref{e:stoch_dominance_weak_basic1}},
\end{equation}
for any allocation $(\hat{c},\hat{x}) \in \mathcal{I}$.\end{proof}

\subsection{Global Optimal Tax Formula in One Dimension}\label{a:ssd}

In this appendix, we develop the connection of our general optimal tax condition with stochastic dominance to the classic ABC formula. First, we consider equation (\ref{e:generalized_el}) under unidimensional skill heterogeneity. With a slight abuse of notation, we denote the unidimensional skill by $p$. In this case, equation (\ref{e:generalized_el}) simplifies to:
\begin{equation}
\int \partial_{p} \big( \pi (p \mathcal{C}'( c )+ z  \mathcal{X}' ( x )) \big) \hat{u} \text{d}p \geq \int \pi (\lambda - C'(c)) \hat{u} \text{d}p \label{e:generalized_el1},
\end{equation}
for any decreasing, nonnegative and convex indirect utility function $\hat{u}$ with $\hat{u}(\bar{p})=0$.\footnote{Asserting there is no bunching at the top of the unidimensional worker skill distribution, both the boundary conditions are zero under the additional condition that $\int \pi C'(c) \text{d}p \geq \lambda$.} Moreover, with one dimension of worker heterogeneity, the measure $f$ second-order stochastically dominates the measure $g$ if and only if $\int_{\underline{p}}^{\hat{p}} F(p) \text{d} p \geq \int_{\underline{p}}^{\hat{p}} G(p) \text{d} p$, where $F$ and $G$ denote cumulative distribution functions (see the next paragraph). When the unidimensional measure $\partial_{p} \pi (p \mathcal{C}'( c )+ z  \mathcal{X}' ( x ))$ second-order stochastically dominates the unidimensional measure $\pi (\lambda - C'(c))$ it thus implies:
\begin{equation}
\int^p_{\underline{p}} \frac{\pi(s)}{u'(\mathcal{C}(c(s)))} s \frac{\tau}{1 - \tau} \text{d}s \leq \int^p_{\underline{p}} \int^t_{\underline{p}} \frac{\pi(s)}{u'(\mathcal{C}(c(s)))} \left( 1 - u'(\mathcal{C}(c(s))) \lambda \right) \text{d}s \text{d}t \label{e:generalized_el2}
\end{equation}
for every worker $p$, where we use the definition of the labor skill wedge (\ref{e:optimal_wedges}), which changes the inequality sign, and also use that $\mathcal{C}'(c)=1/u'(\mathcal{C}(c))$. At an optimum, the utility-weighted average benefit of increasing marginal tax rates for all workers below $p$, on the right, exceeds the corresponding costs. The benefit of an increase in a marginal tax rate is an increase in revenues collected from workers below $p$ (high $\alpha$) net of the cost of tightening the promise-keeping constraint, $\int^t_{\underline{p}} \frac{\pi(s)}{u'(\mathcal{C}(c(s)))} \left( 1 - u'(\mathcal{C}(c(s))) \lambda \right) \text{d}s$. The cost of increasing the marginal tax for all workers below $p$  is captured by the marginal utility-weighted labor wedge. Our optimal tax formula as stochastic dominance (\ref{e:generalized_el}) extends this logic to multidimensional skills.

\vspace{0.4 cm}
\noindent \textbf{Second-Order Stochastic Dominance in One Dimension}. Let $\Upsilon_a(p)$ denote a decreasing, nonnegative and convex function parameterized by $a$ that is strictly positive for all $p < a$ and is equal to zero for all $p \geq a$. Specifically, we let $\Upsilon_a(p) := \max(a-p,0)$. Given that $\Upsilon_a(p)$ is decreasing, nonnegative, and convex, measure $f$ second-order stochastically dominating measure $g$ implies that $\int \Upsilon_a f \text{d} p \geq \int \Upsilon_a g \text{d} p$ for all $a$ following (\ref{e:stoch_dom}). Given the specification for $\Upsilon_a(p)$ this is equivalent to $\int_0^a (a-p) f \text{d} p \geq \int_0^a (a-p) g \text{d} p$ for all $a$, alternatively $\int F \text{d}p \geq \int G \text{d}p$. Since any unidimensional decreasing, nonnegative and convex indirect utility function $\hat{u}$ with $\hat{u}(\bar{p})=0$ can be considered as a positive combination of $\Upsilon_a(p)$, the claim holds.

\vspace{0.4 cm}
\noindent \textbf{Relation to \citet{Mirrlees:1971}, \citet{Diamond:1998}, and \citet{Saez:2001}}. Second, we discuss in more detail how the optimal tax condition (\ref{e:general_abceq}) directly relates to the ABC formulas in \citet{Diamond:1998} and \citet{Saez:2001}, and to the optimal tax condition in \citet{Mirrlees:1971}. To see the relationship, we first take the ABC formula in \citet{Saez:2001}, which is equation (25) in his paper:
\begin{equation}
\frac{\tau_l(\theta)}{1-\tau_l(\theta)} = \left( 1 + \frac{1}{\varepsilon} \right) \frac{1 - F(\theta)}{\theta f(\theta)} \int_\theta^\infty \left[ 1 - \frac{u_c(x)}{p} \right] \frac{u_c(\theta)}{u_c(x)} \frac{f(x)}{1 - F(\theta)} \text{d}x
\end{equation}
Reorganizing this expression, letting $\rho = \left(1 + \frac{1}{\varepsilon} \right)$, we obtain:
\begin{equation}
\frac{1}{\rho} \frac{f(\theta)}{u_c(\theta)} \theta  \frac{\tau_l(\theta)}{1-\tau_l(\theta)} = \int_\theta^\infty \left[ \frac{1}{u_c(x)} - \frac{1}{p} \right] f(x) \text{d} x ,
\end{equation}
which is, in fact, the representation of the optimal tax formula in equation (33) in \citet{Mirrlees:1971}. Differentiating this expression with respect to type $\theta$, we obtain:
\begin{equation}
\frac{1}{\rho} \partial_\theta \left( \frac{f(\theta)}{u_c(\theta)}  \theta \frac{\tau_l(\theta)}{1-\tau_l(\theta)} \right) = \left[ \lambda - \frac{1}{u_c(\theta)} \right] f(\theta) .
\end{equation}
This is the unidimensional analog to our characterization in equation (\ref{e:general_abceq}), where we observe that the multiplier $p$ on the resource constraint in \citet{Saez:2001} is the inverse of the multiplier on the promise keeping constraint $\frac{1}{\lambda}$, which follows directly from the Lagrangian (\ref{e:lagrange_basic}).

\subsection{Global Optimal Tax Formula}\label{a:globaloptimaltax}

In this appendix, we derive the global optimal taxation formula (\ref{e:general_abceq}). For ease of presentation, we assume that the convex indirect utility functions are smooth, meaning the second derivative is well-defined and continuous.

\vspace{0.3 cm}
\noindent First, we reformulate the planner problem using the definition of the indirect utility function (\ref{e:indirect_utility}). Given the indirect utility function $u$, both consumption and effort allocations can be expressed in terms of $u$ and its gradient. The planner chooses an indirect utility function $u$ to minimize the resource cost of providing welfare:
\begin{equation}
\min_{u \in C} \int \left( \mathcal{C} ( u (p) - \nabla u(p) \cdot p ) + z(p) \left( \mathcal{X} \left( - \frac{\partial u(p)}{p_c} \right) + \mathcal{X} \left( - \frac{\partial u(p)}{p_m} \right) \right) \right) \pi (p) \text{d} p ,  \label{e:resources_original4}
\end{equation}
subject to the incentive constraint that requires the indirect utility to be convex and decreasing in worker type $p$ (Lemma \ref{lemma:p_convexity}), or $u \in C$, and the promise keeping condition:
\begin{equation}
\int u(p) \pi (p) \text{d} p \geq \mathcal{U} \label{e:promise_keeping_linearu}.
\end{equation}

We introduce multiplier $\lambda \geq 0$ on the promise keeping condition (\ref{e:promise_keeping_linearu}) in order to formulate the Lagrangian:
\begin{equation*}
\min_{u \in C} \int \hspace{-0.05 cm} \left( \mathcal{C} ( u (p) - \nabla u(p) \cdot p ) + z(p) \left( \mathcal{X} \left(-\frac{\partial u(p)}{\partial p_c}\right) + \mathcal{X} \left(-\frac{\partial u(p)}{\partial p_m}\right) \right) \right) \pi \, \text{d} p - \lambda \left( \int u(p) \pi \, \text{d} p - \mathcal{U} \right) 
\end{equation*}
In the remainder of this appendix, we refer to the problem of minimizing the Lagrangian shorthand as:
\begin{equation}
\min_{u \in C} J(u) = \min_{u \in C} \int_P L(p,u(p), \nabla u(p)) \text{d} p ,
\end{equation}
where $L(p,u(p), \nabla u(p))$ is the contribution of worker type $p$ to the Lagrangian and where $P$ denotes the type space, a bounded convex subset of $\mathbb{R}^2$.

\subsubsection{Convex Indirect Utility}

The restriction that the indirect utility function is convex is equivalent to the Hessian matrix of the indirect utility function being positive semidefinite, $H(u)\succeq 0$ for all $p \in P$, which in turn is equivalent to:
\begin{equation}
v^T H(u) v \geq 0 , 
\end{equation}
for all $v \in \mathbb{R}^2$. We treat these inequalities as an infinite series of constraints parameterized by both $v$ and $p$. For each of these constraints, we introduce a corresponding multiplier $\lambda(v, p) \geq 0$. The objective function is augmented to include these multipliers:
\begin{equation}
\min_{u \in C} \; J(u) - \int_{P} \int_{\mathbb{R}^2} \lambda(v, p) \big( v^T H(u) v \big) \text{d}v \text{d}p. \label{eq:objective_app1}
\end{equation}

Next, fix a worker type $p$ and consider the multipliers associated with this worker type. Using $\langle \cdot, \cdot \rangle$ to denote the inner products of matrices, we write:
\begin{equation}
\int_{\mathbb{R}^2} \lambda(v, p) \big( v^T H(u) v \big)  \text{d} v = \bigg\langle H(u), \int_{\mathbb{R}^2}  \lambda(v, p) v v^T \text{d} v \bigg\rangle.
\end{equation}
This equation expresses the quadratic form as an inner product between the Hessian matrix $H(u)$ and a matrix defined by the integral $\int_{\mathbb{R}^2} \lambda(v, p) v v^T \text{d} v$, which we denote by $A$.

We next characterize all matrices that can be represented as a convex combination of outer products $v v^T$ over the directions $v \in \mathbb{R}^2$:
\begin{equation}
A = \int_{\mathbb{R}^2} \lambda(v,p) v v^T \text{d} v \label{eq:Amatrix}
\end{equation}
where $\lambda(v,p) \geq 0$ is a non-negative function.

\vspace{0.2 cm}
\begin{lemma}\label{l:convexity} Matrix $A$ is a symmetric and positive semidefinite if and only if $A = \int_{\mathbb{R}^2} \lambda(v,p) v v^T \text{d} v$  with respect to some function $\lambda(v,p) \geq 0$.
\end{lemma}

\vspace{0.1 cm}
\begin{proof}
We first establish that $A = \int_{\mathbb{R}^2} \lambda(v,p) v v^T \text{d} v$  with respect to some function $\lambda(v,p) \geq 0$ implies $A$ is symmetric and positive semidefinite. Matrix $A$ is symmetric as each outer product $v v^T$ is a symmetric matrix, and any linear combination or integral of symmetric matrices is symmetric. Second, $A$ is positive semidefinite because for any vector $w \in \mathbb{R}^2$, 
$w^T A w = \int_{\mathbb{R}^2} \lambda(v,p) (w^T v)^2 \text{d} v \geq 0$, which shows $A$ is positive semidefinite.

Next, we show the converse also holds. We start with a matrix that is symmetric and positive semidefinite and show that it can be written in the form of equation (\ref{eq:Amatrix}). Any symmetric matrix $A \in \mathbb{R}^{2 \times 2}$ can be factorized using its eigenvalue decomposition. If $A$ is a positive semidefinite matrix, it can be written as:
\begin{equation*}
A = Q \Lambda Q^T \;,
\end{equation*}
where $Q \in \mathbb{R}^{2 \times 2}$ is an orthogonal matrix whose columns are the eigenvectors of $A$. Since the matrix is positive semidefinite, it follows that $\Lambda$ is a diagonal matrix with nonnegative eigenvalues $\lambda_i \geq 0$. 

Since each eigenvalue $\lambda_i$ corresponds to an eigenvector $q_i$, we rewrite $A$ as a finite sum of outer products of the eigenvectors:
\begin{equation*}
A = \sum_{i=1}^2 \lambda_i q_i q_i^T.
\end{equation*}
The sum can be generalized to an integral, with the eigenvectors $q_i$ replaced by the corresponding vectors $v \in \mathbb{R}^2$, and the eigenvalues $\lambda_i$ replaced by the corresponding continuous nonnegative function $\lambda(v,p)$. In sum, positive semidefinite matrix $A$ can be expressed in the form:
\begin{equation*}
A = \int_{\mathbb{R}^2} \lambda(v) v v^T \text{d} v,
\end{equation*}
where $\lambda(v,p) \geq 0$.\end{proof}

Instead of considering the continuous family of multipliers $\lambda(v,p) $, we represent the constraint that the indirect utility function has to be convex as a matrix condition by introducing the Kuhn-Tucker matrix $M(p)$ for each $p \in P$. By Lemma \ref{l:convexity}, the matrix $M(p)$ is required to be positive semidefinite, or $M(p) \succeq 0$, for all $p \in P$. The Kuhn-Tucker matrix substitutes the term $\int_{\mathbb{R}^2} \lambda(v, p) v v^T \text{d} v$ in the objective function (\ref{eq:objective_app1}), where $\lambda(v, p) \geq 0$:
\begin{equation}
\min_{u} \; J(u) - \int_{P} \big\langle H(u), M(p) \big\rangle \text{d} p,
\end{equation}
where $M(p)$ is the positive semidefinite matrix that enforces the convexity of the indirect utility function.

We proceed by integrating by parts the term:
\begin{equation}
\int_{P} \big \langle H(u), M(p) \big\rangle \text{d}p = \int_{P} \sum_{i,j} \frac{\partial^2 u(p)}{\partial p_i \partial p_j} M_{ij}(p) \text{d}p .
\end{equation}
Through integration by parts, we shift the derivatives from the indirect utility function to the Kuhn-Tucker matrix to obtain:\footnote{Since the boundary terms do not affect the derivation of the optimality condition and the general optimal taxation formula, we suppress them for ease of exposition.}
\begin{equation}
\int_{P} \sum_{i,j} \frac{\partial^2 u(p)}{\partial p_i \partial p_j} M_{ij}(p) \text{d} p = \int_{P} \sum_{i,j} \frac{\partial^2 M_{ij}(p)}{\partial p_i \partial p_j} u(p)  \text{d} p = \int_{P} u(p) \Delta M(p) \text{d} p \;,
\end{equation}
where $\Delta M(p) = \sum\limits_{i, j}\frac{\partial^2 M_{ij}(p)}{\partial p_i \partial p_j}$. The resulting objective function for the problem is:
\begin{equation}
\min_{u}  \int_{P} L(p, u(p), \nabla u(p)) \, \text{d} p - \int_{P} u(p) \Delta M(p)\, \text{d}  p . \label{eq:elobjective}
\end{equation}
We can rewrite this by combining the promise keeping constraint with the convexity correction $u(p)\Delta M(p)$ as:
\begin{align}
\min_{u\in C}\int\hspace{-0.05cm} & \left(\mathcal{C}(u(p)-\nabla u(p)\cdot p) + z(p)\left(\mathcal{X}\left(-\frac{\partial u(p)}{\partial p_{c}}\right)+\mathcal{X}\left(-\frac{\partial u(p)}{\partial p_{m}}\right)\right)\right)\pi\,\text{d}p \notag \\ 
& \hspace{4.35 cm} -\lambda\left(\int\pi u(p)\left(1+\frac{\Delta M(p)}{\lambda\pi}\right)\,\text{d}p-\mathcal{U}\right)
\end{align}
This representation shows that the requirement that the indirect utility function is convex leads to the modified social welfare weight $1+\frac{\Delta M(p)}{\lambda\pi}=1+\frac{\sum \frac{\partial^{2}}{\partial p_{i}\partial p_{j}}M_{ij}(p)}{\lambda\pi}$. In other words, the main difference that convexity adds to the planning problem is through modifying the welfare function by the convexity correction. We next show how this result carries over to the optimality conditions.

\subsubsection{Optimality Conditions} \label{a:ocp2}

We derive the optimality conditions by using the objective function (\ref{eq:elobjective}), and by considering a small variation in the indirect utility function. The first variation of the objective gives:
\begin{equation}
\frac{\partial L}{\partial u} - \sum_{k = 1}^2 \frac{\partial}{\partial p_k}\left( \frac{\partial L}{\partial u_k} \right) = \Delta M(p) , \label{eq:el}
\end{equation}
where $u_k = \frac{\partial u}{\partial p_k}$. The left-hand side is the standard optimality condition, $\frac{\partial L}{\partial u} - \sum \frac{\partial}{\partial p_k}\big( \frac{\partial L}{\partial u_k} \big)$. This yields the optimality condition without bunching, $\frac{\partial L}{\partial u} - \sum \frac{\partial}{\partial p_k}\big( \frac{\partial L}{\partial u_k} \big) = 0$, when the convexity constraint does not bind.

The right-hand side gives the additional term involving the second derivatives of the Kuhn-Tucker matrix. This term arises from the integration by parts of the Kuhn-Tucker matrix, and represents the effect of the convexity constraint, which is enforced through the positive semidefinite matrix $M$.


Our derivation shows that the minimizer of our variational problem over convex indirect utility functions satisfies an optimality condition with an additional term arising from the Kuhn-Tucker multipliers associated with the convexity constraint. Our results build on \citet{Lions:1998} which analyzes variational problems over convex functions through a duality approach. \citet{Lions:1998} shows that the optimality conditions can be understood in terms of the polar cone of convex functions, where elements of the dual space are represented by measures linked to second derivatives. Our result explicitly incorporates Kuhn-Tucker multipliers into the variational problem to enforce the convexity constraint and includes them into the optimality condition. This  construction directly shows the effect of convexity by showing how additional measure terms arise in the optimality conditions.

\subsubsection{Optimal Tax Formula} \label{a:ocp3}

We next apply the optimality condition for the general Lagrangian (\ref{eq:el}) to the Lagrangian for the optimality multidimensional taxation problem (\ref{e:resources_original4}). As a result, we write the optimal taxation formula as:
\begin{equation}
     \partial_{p_c} \big( \pi (p_c \mathcal{C}'( c )+ z  \mathcal{X}' ( x_c ))) + \partial_{p_m} \hspace{-0.07 cm} \left( \pi (p_m \mathcal{C}'( c )+ z  \mathcal{X}' ( x_m )) \right) = \pi (\lambda - C'(c)) + \Delta M(p).
\end{equation}
Similar to our reformulation of the general optimal tax formula as stochastic dominance (\ref{e:general_abc}), we use the definition of the labor skill wedge (\ref{e:optimal_wedges}) to rewrite the optimal taxation formula as:
\begin{equation}
\partial_{p_c} \Big( \frac{\pi}{u'(\mathcal{C}(c))} p_c \frac{\tau_c}{1 - \tau_c}  \Big) + \partial_{p_m} \Big(  \frac{\pi}{u'(\mathcal{C}(c))} p_m \frac{\tau_m}{1 - \tau_m} \Big) = \pi \left( \frac{1}{u'(\mathcal{C}(c))} - \lambda \right) - \Delta M(p),
\end{equation}
which is the optimal taxation formula (\ref{e:general_abceq}).

\subsection{Corollary \ref{p:el}} \label{a:el}

We start with the region of strong convexity of the indirect utility function $u$ and, hence, a region without bunching. To analyze properties of optimal tax distortions, we use a perturbation function. Specifically, we construct a variation of the indirect utility function for a specific worker $p$. Consider a worker $p$ in the interior of the type space such that both the assignment function $z$ and the distribution of worker types $\pi$ are differentiable in a neighborhood around this worker. Moreover, suppose that the strongly convex utility function $u$ is twice continuously differentiable within a neighborhood of the worker $p$. 

Consider an arbitrary perturbation of the indirect utility $u$ denoted $\hat{u} = u + \varepsilon V$, where $V$ is a bump function that is concentrated in a small ball around $p$ which lies within the neighborhood around $p$, and $\varepsilon$ is small. The arbitrary perturbation function $u + \varepsilon V$ is convex for small enough values for $\varepsilon$ within the support of the bump function, $\vert \varepsilon \vert < \bar{\varepsilon}$. Intuitively, if the underlying utility function is strongly convex, a small enough additive perturbation preserves convexity.\footnote{The proof of this statement is presented below. See \textit{Convex Perturbation Function}.}

The perturbation function is convex, positive and non-increasing, and therefore implementable (\ref{e:stoch_dominance_weak_basic}). Since the implementability condition (\ref{e:stoch_dominance_weak_basic}) is linearly separable and holds with equality for an optimal utility function by Proposition \ref{p:implementable}, the implementability also has to be satisfied for $\varepsilon V$ for all $|\varepsilon| \leq \bar{\varepsilon}$. Since $\varepsilon$ can take either positive or negative values, the implementability condition holds with equality with respect to the bump function $V$:
\begin{equation}
\int \big( \mathcal{C}'( c ) \big( V - \nabla V \cdot p \big) - z \mathcal{X}' ( x ) \cdot \nabla V \big) \pi  \text{d}p = \lambda \int V \pi \text{d}p \label{e:stoch_dominance_weak_V}.
\end{equation}
Integrating the left-hand side of this equation by parts and tending the bump function $V$ to the Dirac delta function, we obtain the optimality condition equation in Corollary \ref{p:el}.

\vspace{0.4 cm}
\noindent \textbf{Convex Perturbation Function}. We establish that the perturbation function is convex. We suppose that the indirect utility function $u$ is strongly convex for interior worker type $p$ and twice continuously differentiable within its neighborhood. Specifically, we suppose that $H(u) - \alpha_I I$ is positive semidefinite for worker $p$ for some $\alpha_I > 0$, where $H$ denotes the Hessian matrix and $I$ denotes the identity matrix. 

Since worker $p$ is in the interior of the type space, the indirect utility function is strictly positive and strictly decreasing for worker $p$. By contradiction, suppose the indirect utility function equals zero for worker $p$, $u(p)=0$. Since the indirect utility function is non-increasing,  $u(p+\varepsilon) = 0$ for small enough $\varepsilon \geq 0$, implying that the gradient of the indirect utility function for worker $p$ is equal to zero, $\nabla u(p) = 0$. By implication, consider that the partial derivative of the indirect utility function with respect to cognitive type $p_c$ equals zero, $\frac{\partial}{\partial p_c} u(p) = 0$. Since we consider a partial derivative for a convex function, the partial derivative increases with $p_c$  so that $\frac{\partial}{\partial p_c} u(p_c + \varepsilon_c, p_m) = 0$ for all $\varepsilon_c \geq 0$, or $\frac{\partial^2}{\partial^2 p_c} u(p) = 0$. It hence follows that $H_{cc}(u) = 0$, and hence that $H_{cc}(u) - \alpha_I < 0$ for $\alpha_I > 0$ which contradicts that $H(u) - \alpha_I I$ is positive semidefinite by the Sylvester criterion. We conclude that the utility function is strictly positive and strictly decreasing for interior worker $p$.

Since the indirect utility function $u$ is strongly convex for worker type $p$ and twice continuously differentiable within its neighborhood, the utility function is strongly convex in this neighborhood. The restriction that $H(u) - \alpha_I I$ is positive semidefinite in a neighborhood around worker type $p$ implies $H(u) - \frac{\alpha_I}{2} I$ is positive semidefinite in the neighborhood around $p$ when the indirect utility function is twice continuously differentiable. Hence, the utility function $u$ is indeed strongly convex in this neighborhood.


We consider a perturbation of the indirect utility $u$ denoted by $u + \varepsilon V$, where $V$ is a bump function that is concentrated in a small ball around $p$ which lies within the neighborhood around $p$, and $\varepsilon$ is small. The arbitrary perturbation function $u + \varepsilon V$ is convex for small enough values for $\varepsilon$ within the support of the bump function, $\vert \varepsilon \vert < \bar{\varepsilon}$. 

While intuitive, we prove that $u + \varepsilon V$ is convex for small enough values for $\varepsilon$ within the support of the bump function in two steps. First, we observe that for some $\beta>0$, it holds that $H(V) - \beta I$ is negative semidefinite and that $H(V) + \beta I$ is positive semidefinite. In the former case, negative semidefinite is equivalent to $x_c^2 V^{\phantom{2}}_{cc}+ 2 x^{\phantom{2}}_c x^{\phantom{2}}_m V^{\phantom{2}}_{cm} + x_m^2 V^{\phantom{2}}_{mm} \leq \beta (x_c^2 + x_m^2 )$ for any $(x_c,x_m)$. To see this, we first observe $x_c^2 V^{\phantom{2}}_{cc}+ 2 x^{\phantom{2}}_c x^{\phantom{2}}_m V^{\phantom{2}}_{cm} + x_m^2 V^{\phantom{2}}_{mm} \leq |x_c|^2 |V^{\phantom{2}}_{cc}|+ 2 |x^{\phantom{2}}_c| |x^{\phantom{2}}_m| |V^{\phantom{2}}_{cm}| + |x_m|^2 |V^{\phantom{2}}_{mm}|$. Furthermore, we use that $2 |x_m| |x_c| \leq |x_c|^2 + |x_m|^2$ to write $x_c^2 V^{\phantom{2}}_{cc}+ 2 x^{\phantom{2}}_c x^{\phantom{2}}_m V^{\phantom{2}}_{cm} + x_m^2 V^{\phantom{2}}_{mm} \leq x_c^2 (|V^{\phantom{2}}_{cc}| + |V^{\phantom{2}}_{cm}|) + x_m^2 (|V^{\phantom{2}}_{cm}| + |V^{\phantom{2}}_{mm}|)$. Therefore, there indeed exists $\beta = \max (|V^{\phantom{2}}_{cc}| + |V^{\phantom{2}}_{cm}|, |V^{\phantom{2}}_{cm}| + |V^{\phantom{2}}_{mm}|) > 0$ such that $H(V) - \beta I$ is negative semidefinite. Through a similar argument $H(V) + \beta I$ is positive semidefinite. Given $\beta > 0$, it holds that $\varepsilon H(V) +|\varepsilon| \beta I$ is positive semidefinite for positive $\varepsilon$, and that $\varepsilon (H(V) - \beta I) = \varepsilon H(V) +|\varepsilon| \beta I$ is positive semidefinite for negative $\varepsilon$. 

Second, we note that the Hessian matrix for the perturbation function is additively separable, $H(u + \varepsilon V) = H(u) + \varepsilon H(V)$. Since the matrix $H(u) - \frac{\alpha_I}{2} I$ is positive definite, the matrix $H(u + \varepsilon V) - \frac{\alpha_I}{2} I - \varepsilon H(V)$ is positive definite. Finally, since the sum of positive semidefinite matrices is itself positive semidefinite, it follows that $H(u + \varepsilon V) - \left( \frac{\alpha_I}{2} - |\varepsilon| \beta \right) I$ is positive semidefinite for $\varepsilon$ small enough, which confirms that the perturbation function is indeed convex. Following analogous reasoning, the indirect utility function $u$ is also decreasing and positive in a neighborhood around worker $p$.


\vspace{0.4 cm}
\noindent \textbf{Changing Coordinates}. To connect our expression to the existing literature, we transform the optimal tax formula into the original type coordinates $\alpha$. We illustrate this transformation by focusing on the partial derivative with respect to cognitive skill in (\ref{e:euler}), 
\begin{equation}
\partial_{p_c} \left( \frac{\pi}{u'(\mathcal{C}(c))} p_c \frac{\tau_s}{1 - \tau_c}  \right) , \label{e:taus1}
\end{equation}
where $\pi$ is the probability distribution in the transformed worker space $p$. To convert this term into the original worker space, we first recall the change of coordinates $p_s = \kappa \alpha_s^{-\rho}$, or equivalently $\alpha_s = ( \kappa / p_s )^{\frac{1}{\rho}}$, implying that $\text{d}\alpha_s = - \frac{1}{\rho} \frac{\alpha_s}{p_s} \text{d} p_s = - \frac{1}{\kappa \rho} \alpha_s^{\rho+1} \text{d} p_s $.

We first explicitly formulate the relationship between the distribution function in the original worker type space $\alpha$ given by $\phi$, and the worker distribution function in transformed coordinates $p$ given by $\pi$:
\begin{align*}
\phi (\alpha) \text{d} \alpha_c \text{d} \alpha_m & = \phi (\alpha) \frac{\alpha_c^{\rho+1} \alpha_m^{\rho+1}}{(\kappa \rho)^2} \text{d} p_c  \text{d} p_m = \pi(p)  \text{d} p_c \text{d} p_m ,
\end{align*}
where the distribution function $\pi(p) := \phi (\alpha) \alpha_c^{\rho+1} \alpha_m^{\rho+1} / (\kappa \rho)^2$. As a result, we express (\ref{e:taus1}) as: 
\begin{equation}
\partial_{p_c} \left( \frac{\pi}{u'(c(\alpha))} \frac{\tau_c}{1 - \tau_c}  \kappa \alpha_c^{-\rho} \right) = \partial_{p_c} \bigg( \frac{\phi}{u'(c(\alpha))} \frac{\tau_c}{1 - \tau_c}  \frac{\alpha_c \alpha_m^{\rho+1}}{\kappa \rho^2} \bigg) , \label{e:taus2}
\end{equation}
 

Next, by the chain rule we have that $\frac{\partial z}{\partial p_c} = \frac{\partial z}{\partial \alpha_c} \frac{\partial \alpha_c}{\partial p_c}$, which gives:
\begin{align*}
\partial_{p_c} \bigg( \frac{\pi}{u'(c(\alpha))} \frac{\tau_c}{1 - \tau_c}  \kappa \alpha_c^{-\rho} \bigg) & = \partial_{\alpha_c} \bigg( \frac{\phi}{u'(c(\alpha))} \frac{\tau_c}{1 - \tau_c}  \frac{\alpha_c \alpha_m^{\rho+1}}{\kappa \rho^2} \bigg) \frac{\partial \alpha_c}{\partial p_c} = - \frac{\alpha_c^{\rho+1} \alpha_m^{\rho+1}}{\kappa^2 \rho^3} \partial_{\alpha_c} \bigg( \frac{\phi}{u'(c(\alpha))} \frac{\tau_c}{1 - \tau_c} \alpha_c \bigg) 
\end{align*}
The derivation for the manual skill term is symmetric, which allows us to summarize the previous two expressions for both tasks as:
\begin{equation}
\partial_{p_s} \bigg( \frac{\pi}{u'(c(\alpha))} \frac{\tau_s}{1 - \tau_s}  \kappa \alpha_c^{-\rho} \bigg) = - \frac{\alpha_c^{\rho+1} \alpha_m^{\rho+1}}{\kappa^2 \rho^3} \partial_{\alpha_s} \bigg( \frac{\phi}{u'(c(\alpha))} \frac{\tau_s}{1 - \tau_s} \alpha_s \bigg)  , \label{e:taus3}
\end{equation}
Finally, we rewrite the left side of equation (\ref{e:euler}), using the relation between density functions, as
\begin{equation}
\pi \bigg( \frac{1}{u'(\mathcal{C}(c))} - \lambda \bigg) = - \phi(\alpha) \frac{\alpha_c^{\rho+1} \alpha_m^{\rho+1}}{\kappa^2 \rho^2} \bigg( \lambda - \frac{1}{u'(c(\alpha))}\bigg) \label{e:lhs_el}
\end{equation}
Combining equation (\ref{e:euler}) in the worker type space $p$, with (\ref{e:taus3}) and (\ref{e:lhs_el}), we obtain (\ref{e:eulera}) in the worker space $\alpha$.

\subsection{Proposition \ref{p:bunch}} \label{pf:bunch}

By Corollary \ref{p:el} it follows that when the optimality condition does not hold, the Hessian matrix is degenerate for worker $p$. We next show that the Hessian matrix $H(u)$ is also degenerate for all workers within the neighborhood of $p$. By contradiction, suppose that in every neighborhood of point $p$ we can find a worker $\hat{p}$ such that its Hessian is non-degenerate, or equivalently, has full rank. By Corollary \ref{p:el}, the optimality condition holds for worker $\hat{p}$. We can thus construct a sequence of points $\{\hat{p}_n\}$ that converges to $p$. Since the optimality equation is continuous in $p$, the sequence converges and that the optimality equation holds for worker $p$, which is a contradiction.
 
\subsection{Planner Duality} \label{pf:planner_duality}

We prove duality between our cost minimization problem and a welfare maximization problem. The welfare maximization problem is to choose allocation $(c,x)$ to maximize utilitarian welfare:
\begin{equation}
\int \big( c - p_c x_c - p_m x_m \big) \pi \text{d} p, \label{e:linear_welfare}
\end{equation}
subject to the incentive constraints (\ref{e:linear_ic}) and the linear resource constraint:
\begin{equation}
\int \big( \mathcal{C} ( c ) + z (p) \big( \mathcal{X} ( x_c ) + \mathcal{X} ( x_m ) \big)  \big) \pi \text{d} p \leq R, \label{e:convex_resources}
\end{equation}
for some exogenous level of federal resources $R$.

\begin{proposition}
Let $( c,x )$ solve the cost minimization problem associated with maximum welfare level $\overline{\mathcal{U}}$ so that the minimum resource cost is less than government resources $R$. Then allocation $(c,x)$ solves the welfare maximization problem given government resources $R$. 

Conversely, if allocation $(c,x)$ solves the welfare maximization problem for resources $R$ and induces welfare $\overline{\mathcal{U}}$, then $(c,x)$ solves the cost minimization solves the cost minimization problem for $\mathcal{U} = \overline{\mathcal{U}}$.
\end{proposition}

\begin{proof}
First, we establish that the welfare attained by the cost minimization problem and welfare maximization problem are identical. Consider the solution to the cost minimization problem with maximum welfare level $\overline{\mathcal{U}}$ such that the resource cost is below resource level $R$. Allocation $(c,x)$ satisfies both the incentive constraints and the resource constraints of the welfare maximization problem and is thus a feasible solution to the welfare maximization problem. Welfare in the welfare maximization problem therefore exceeds $\overline{\mathcal{U}}$. 

Conversely, take the solution to the welfare maximization and let $\overline{\overline{\mathcal{U}}}$ denote maximum welfare. Consider the allocation $(c,x)$ that solves the welfare maximization problem. The allocation $(c,x)$ satisfies both the incentive constraints and the promise keeping constraint to the cost minimization problem. Further, the associated resource cost is below resource level $R$. Hence, $\overline{\mathcal{U}} \geq \overline{\overline{\mathcal{U}}}$, implying that welfare is identical for the two problems, $\overline{\mathcal{U}} = \overline{\overline{\mathcal{U}}}$.

Second, we show duality of allocations. Suppose allocation $(c,x)$ solves the cost minimization problem with maximum welfare level $\mathcal{U}$ such that the cost is below resources $R$, but the allocation does not solve the welfare maximization problem. Then, there is an alternative allocation $(\hat{c},\hat{x})$ that solves the welfare maximization problem, is feasible, and attains strictly greater welfare. This implies that there exists a welfare level $\hat{\mathcal{U}} > \mathcal{U}$ so that $(\hat{c},\hat{x})$ has a cost below resources $R$, contradicting that $\mathcal{U}$ is the maximum welfare so that the minimum resource cost is below $R$.

Conversely, suppose allocation $(\hat{c},\hat{x})$ solves the welfare maximization problem given resources $R$ inducing welfare $\hat{\mathcal{U}}$, but does not solve the cost minimization problem. Then, there exists an alternative allocation $(c,x)$ that solves the cost minimization problem for a welfare level $\mathcal{U}> \hat{\mathcal{U}}$ such that the minimum cost is below resources $R$. Allocation $(c,x)$ is feasible and attains strictly greater welfare, contradicting that $(\hat{c},\hat{x})$ solves the welfare maximization problem.\end{proof}

\subsection{Transformed Planner Problem} \label{s:transformed_planner}

In this appendix, we analyze the planner problem of choosing an allocation $(c,x)$ to minimize the resource cost of providing welfare as in Section \ref{ss:change_of_var}. Using the Legendre transforms (\ref{e:lt_c}) and (\ref{e:lt_x}) to linearize the resource costs, the planning problem is equivalent to: 
\begin{equation}
\min_{c,x} \; \max\limits_{\varphi,\psi} \; \int \Big( \big( \varphi (p) c(p) - \mathcal{C}^*(\varphi(p)) \big) + z(p) \sum_{s} \big( \psi_s(p) x_s(p) - \mathcal{X}^*(\psi_s(p)) \big) \Big) \pi (p) \text{d} p
\end{equation}
subject to the set of linear irreducible incentive constraints (\ref{e:linear_ic}) and the promise keeping constraint (\ref{e:promise_keeping_linear}).

To develop properties of the solution we formulate a Lagrangian, where $\lambda$ is the multiplier on the promise keeping constraint: 
\begin{equation}
\mathcal{L}\big( c,x,\varphi,\psi \big) = \hspace{-0.08 cm} \int \hspace{-0.08 cm} \Big( \big( \varphi c - \mathcal{C}^*(\varphi) \big) + z\sum_{s} \big( \psi_s x_s - \mathcal{X}^*(\psi_s) \big) - \lambda \Big( \int \big( c - \sum p_s x_s \big) - \mathcal{U} \Big) \Big) \pi  \text{d}p . \label{e:lagrange}
\end{equation}
The Lagrangian is a continuous function that is concave-convex. Since the Legendre transform of a convex function is itself convex, the Lagrangian is concave in the distortions $(\varphi,\psi)$ holding constant the allocations $(c,x)$, and convex in the allocations when holding constant the distortions. Further, since the set of allocations that satisfies the incentive constraints (\ref{e:linear_ic}) is convex, we can apply the minimax theorem. We use the minimax relationship, $\min\limits_{c,x \in \mathcal{I}} \; \max\limits_{\varphi \geq 0,\psi \leq 0} \mathcal{L}\big( c,x,\varphi,\psi \big) = \max\limits_{\varphi \geq 0,\psi \leq 0} \; \min\limits_{c,x \in \mathcal{I}} \; \mathcal{L}\big( c,x,\varphi,\psi \big)$, to establish Lemma \ref{lemma:sd}.


\vspace{0.45 cm}
\begin{lemma}\label{lemma:sd}
For every incentive compatible allocation $(c,x) \in \mathcal{I}$, stochastic dominance has to be satisfied:
\begin{equation}
\int \hspace{-0.05 cm} \Big( \varphi c + \sum \psi_s x_s \Big) \pi  \text{d}p  \geq \lambda \int \hspace{-0.05 cm} \big( c - \sum p_s x_s \big) \pi  \text{d}p \label{e:sd} .
\end{equation}
\end{lemma}
\vspace{0.25 cm}

\noindent The result follows by analyzing $\max\limits_{\varphi \geq 0,\psi \leq 0} \; \min\limits_{c,x \in \mathcal{I}} \; \mathcal{L}\big( c,x,\varphi,\psi \big)$. By contradiction, suppose instead that $\int ( \varphi c + \sum \psi_s x_s ) \pi  \text{d}p  < \lambda \int ( c - \sum p_s x_s ) \pi  \text{d}p $. Consider an increase in the allocation $(c,x)$ by a constant factor $\zeta > 1$. Since incentive compatible constraints are linear, the alternative allocation $\zeta (c,x)$ is feasible. By increasing the constant factor, $\zeta \rightarrow \infty$, optimization would lead to negative infinity, which is not optimal. At the solution to the planning problem, the stochastic dominance condition (\ref{e:sd}) will hold with equality.

\begin{proposition}\label{p:stoch_dominance}
Let $( c,x,\varphi,\psi )$ solve the planning problem, then stochastic dominance condition holds with equality at optimum:
\begin{equation}
\int \Big( \varphi c + \sum \psi_s x_s \Big) \pi  \text{d}p  = \lambda \int \big( c - \sum p_s x_s \big) \pi  \text{d}p \label{e:stoch_dominance}.
\end{equation}
\end{proposition}

\vspace{0.2 cm}
\begin{proof}
To establish the result, we use two problems. First, define the maximization problem: 
\begin{equation}
\max\limits_{\varphi,\psi} \;  \underline{\mathcal{L}}(\varphi,\psi,\lambda) , 
\end{equation}
where $\underline{\mathcal{L}}(\varphi,\psi,\lambda) := \min\limits_{c,x}  \mathcal{L}( c,x,\varphi,\psi , \lambda)$. Let $(\varphi^*,\psi^*)$ be a solution to this problem. Similarly, we define a minimization problem:
\begin{equation}
 \min\limits_{c,x} \; \bar{\mathcal{L}}(c,x) , 
\end{equation}
where $\bar{\mathcal{L}}(c,x) := \max\limits_{\varphi,\psi} \mathcal{L}( c,x,\varphi,\psi , \lambda)$, and let $(c^*,x^*)$ be a minimizer to this problem. 

\vspace{0.4 cm}
\begin{claim}\label{c:lagrange}
We show that for the Lagrangian (\ref{e:lagrange}) evaluated at the optimum it holds that:
\begin{equation}
\mathcal{L}\big( c^*,x^*,\varphi^*,\psi^*, \lambda \big) = \min\limits_{c,x} \; \max\limits_{\varphi, \psi} \; \mathcal{L}\big( c,x,\varphi,\psi, \lambda \big) = \max\limits_{\varphi, \psi} \; \min\limits_{c,x} \; \mathcal{L}\big( c,x,\varphi,\psi, \lambda \big) \label{e:lagrange_optimum}
\end{equation}
\end{claim}

\begin{proof} 
Necessarily it holds that $\mathcal{L}( c^*,x^*,\varphi^*,\psi^*, \lambda ) \geq \min\limits_{c,x} \mathcal{L}( c,x,\varphi^*,\psi^* , \lambda ) = \underline{\mathcal{L}}(\varphi^*,\psi^*,\lambda)$. Since $(\varphi^*,\psi^*)$ solves the optimization problem, $\underline{\mathcal{L}}(\varphi^*,\psi^*,\lambda) = \max\limits_{\varphi,\psi} \underline{\mathcal{L}}(\varphi,\psi,\lambda) = \max\limits_{\varphi,\psi} \min\limits_{c,x}  \mathcal{L}( c,x,\varphi,\psi , \lambda)$, and thus it follows $\mathcal{L}( c^*,x^*,\varphi^*,\psi^*, \lambda ) \geq  \max\limits_{\varphi,\psi} \min\limits_{c,x}  \mathcal{L}( c,x,\varphi,\psi , \lambda)$.

Similarly, note that it necessarily holds that $\mathcal{L}( c^*,x^*,\varphi^*,\psi^*, \lambda) \leq \max\limits_{\varphi,\psi } \mathcal{L}( c^*,x^*,\varphi,\psi , \lambda ) = \bar{\mathcal{L}}(c^*,x^*)$. Since the utility allocation $(c^*,x^*)$ is a solution to the minimization problem, $\bar{\mathcal{L}}(c^*,x^*) = \min\limits_{c,x} \; \bar{\mathcal{L}} (c,x) = \min\limits_{c,x} \max\limits_{\varphi,\psi} \mathcal{L}( c,x,\varphi,\psi , \lambda)$. Combining the previous two statements, we conclude $\mathcal{L}( c^*,x^*,\varphi^*,\psi^*, \lambda^*) \leq \min\limits_{c,x} \max\limits_{\varphi,\psi} \mathcal{L}( c,x,\varphi,\psi , \lambda)$, and hence: 
\begin{equation}
\min\limits_{c,x} \max\limits_{\varphi, \psi} \; \mathcal{L}( c,x,\varphi,\psi , \lambda) \geq \mathcal{L}( c^*,x^*,\varphi^*,\psi^*, \lambda ) \geq \max\limits_{\varphi, \psi} \min\limits_{c,x} \; \mathcal{L}( c,x,\varphi,\psi , \lambda) .
\end{equation}
By the minimax theorem it follows that (\ref{e:lagrange_optimum}) applies.\end{proof} 

\vspace{0.4 cm}
\noindent \textbf{Optimality Conditions}. We obtain optimality conditions analyzing the planner problem using $\mathcal{L}( c^*,x^*,\varphi^*,\psi^*, \lambda^*) = \max \; \min \; \mathcal{L}( c,x,\varphi,\psi, \lambda ) $ from Claim \ref{c:lagrange}. By reorganizing terms:
\begin{equation}
\max\limits_{\varphi, \psi} \min\limits_{c,x \in \mathcal{I}} \; \int \Big( \varphi c + \sum  \psi_s x_s - \lambda \Big( c - \sum p_s x_s \Big) - \mathcal{C}^*(\varphi) - \sum  \mathcal{X}^*(\psi_s) - \lambda \mathcal{U} \Big) \pi  \text{d}p .
\end{equation}
We observe that only the first four terms depend on the utility allocation, and observe further that these terms are necessarily jointly positive for some utility allocation to be incentive compatible following Lemma \ref{lemma:sd}. Since allocation $(c,x) = 0$ is incentive compatible and attains the minimum, the optimal utility allocation is chosen such that these terms jointly equal zero, implying: 
\begin{equation}
\int \big( \varphi^* c^* + \sum \psi^*_s x^*_s \big) \pi  \text{d}p  = \lambda \int \big( c^* - \sum p^*_s x^*_s \big) \pi  \text{d}p \tag{\ref{e:stoch_dominance}},
\end{equation}
and thus concluding the proof.\end{proof}
 
To obtain further optimality conditions to our problem, we analyze the planner problem using $\mathcal{L}( c^*,x^*,\varphi^*,\psi^*, \lambda) = \min \; \max \;  \mathcal{L}( c,x,\varphi,\psi, \lambda ) $ in Claim \ref{c:lagrange} to write:
\begin{equation}
\min\limits_{c,x} \; \max\limits_{\varphi, \psi} \int \Big( \varphi c - \mathcal{C}^*(\varphi) + \sum  \psi_s x_s - \sum \mathcal{X}^*(\psi_s) - \lambda \big( c - \sum p_s x_s  - \mathcal{U}\big) \Big) \pi  \text{d}p .
\end{equation}
We observe that only the first four terms depend on the convex conjugates, and that only the final term depends on the multiplier, in terms of the inner maximization problem. Since the promise keeping condition requires $c - \sum p_s x_s \geq \mathcal{U}$, and $\lambda \geq 0$, it has to hold that $\lambda (c^* - \sum p^*_s x^*_s - \mathcal{U}) = 0$. Similarly, it has to hold that $\varphi^* = {\mathcal{C}}'(c^*)$ and $\psi^*_s = {\mathcal{X}}'(x^*_s)$.

\subsection{Numerical Approach}\label{a:numerical_approach}

\noindent \textbf{Linearization of the Problem}. We now discuss the linearization of the problem that is central to numerical tractability. The only nonlinear part of the optimization problem that remains to be linearized is the objective:
\begin{equation}
\min \int \left( \mathcal{C} ( c (p) ) + z(p) \big( \mathcal{X} ( x_c (p) ) + \mathcal{X} ( x_m (p) ) \big) \right) \pi (p) \text{d} p .  \label{e:resources_original5}
\end{equation}
To illustrate our approach, we focus on the linearization of the convex resource cost function for consumption utility $\mathcal{C}$, and we suppose that bounds for the optimal solution are known a priori, or $\underline{c}(p) \leq c(p) \leq \bar{c}(p)$ and $\underline{x}_s(p) \leq x_s(p) \leq \bar{x}_s(p)$.

The idea is to approximate the convex  cost for consumption utility $\mathcal{C}$ from below with the tangent lines on the bounded interval. For each worker type $p$, it follows from the definition of the Legendre transform (\ref{e:lt_c}) that $\mathcal{C}(c(p)) = \max\limits_{\varphi} \varphi c(p) - \mathcal{C}^*(\varphi)$. We replace this continuous set of tangent slopes $\varphi$ in (\ref{e:lt_c}) with a finite set of tangent lines. Specifically, we consider a list of slopes $\{ \varphi_i(p) \}_{i=1}^n$ with corresponding tangent lines $l^c_{ip}(t) := \varphi_i(p) t - \mathcal{C}^*(\varphi_i(p))$ such that the inequality:
\begin{equation}
0 \leq \mathcal{C}(t)  - \max_{1 \leq i \leq n} \; l^c_{ip}(t) \leq \varepsilon_c
\end{equation}
holds for all $t$ in the bounded interval $[\underline{c}(p), \bar{c}(p)]$. Analogously, to linearize the resource cost of labor disutility $\mathcal{X}$, we consider a list of slopes $\{ \psi^s_i(p) \}_{i=1}^n$ with corresponding tangent lines $l^s_{ip}(t) :=\psi^s_i(p) t - \mathcal{X}^*(\psi^s_i(p))$ such that the inequality:
\begin{equation}
0 \leq \mathcal{X}(t)  - \max_{1 \leq i \leq n} \; l^s_{ip}(t) \leq \varepsilon_s
\end{equation}
holds for each skill $s \in \mathcal{S}$ and for all $t$ in the interval $[\underline{x}_s(p), \bar{x}_s(p)]$.

As a key step, we next introduce independent auxiliary variables $r(p)$ for each worker $p$ satisfying the following set of linear inequalities for all $i$:
\begin{equation}
r(p) \ge \varphi_i(p) c(p) - \mathcal{C}^*(\varphi_i(p)). \label{e:auxiliary_c}
\end{equation}
It follows from the discussion above that $r(p) \gtrsim \mathcal{C}(c(p))$ for each worker $p$. For the resource cost of disutility from working, we similarly define independent auxiliary variables $r_s(p)$ satisfying the linear inequalities for all $i$:
\begin{equation}
r_s(p) \ge \psi^s_i(p) x_s(p) - \mathcal{X}^*(\psi^s_i(p)). \label{e:auxiliary_x}
\end{equation}

We substitute the auxiliary variables $r(p)$ and $r_s(p)$ for $\mathcal{C}(c(p))$ and $\mathcal{X}(x_s(p))$ into our nonlinear objective to define the approximate planner problem. The approximate planner problem chooses $(c, x_s, r, r_s)$ to solve:
\begin{equation}
\min \int \left(r(p) + z(p) \big( r_c(p) + r_m(p) \big) \right)\text{d}\pi ,
\end{equation}
subject to the incentive constraints (\ref{e:linear_ic}), the promise keeping constraint (\ref{e:promise_keeping_linear}), constraints on the auxiliary variables (\ref{e:auxiliary_c}) and (\ref{e:auxiliary_x}), and the approximation bounds for consumption utility $\underline{c}(p)\leq c(p) \leq\bar{c}(p)$ and task outputs $\underline{x}_{s}(p)\leq x_{s}(p)\leq\bar{x}_{s}(p)$. 

\vspace{0.4 cm}
\noindent \textbf{Accuracy}. We next describe the accuracy of the approximate planner problem and provide the algorithm that we use to characterize its solution. The precision of the solution to the approximate planner's problem naturally depends on the accuracy of the prior location of the solution. The criterion we evaluate to ensure that the location is accurate is the absence of binding boundary constraints at the optimal solution. In line with this criterion, we define a solution is proper when no boundary constraints binds.

\vspace{0.1 cm}
\begin{definition*} The solution to the approximate problem is \underline{proper} if the solution is strictly interior, that is $\underline{c}(p) <  c(p) < \bar{c}(p)$, $\underline{x}_{s}(p)< x_{s}(p) <\bar{x}_{s}(p)$ if $\underline{x}_s(p) \neq 0$ and $x_{s}(p)\geq \underline{x}_{s}(p)$ when $\underline{x}_{s}(p)=0$.
\end{definition*}

\vspace{0.1 cm}
\noindent Proposition \ref{p:approximation} shows that one can readily verify that a proper solution approximates well the optimal solution to the initial planner problem.

\vspace{0.05 cm} 
\begin{proposition}\label{p:approximation}
For the approximate problem, introduce the maximal approximation errors:
\begin{equation*}
 \varepsilon := \max_p \max_{\underline{c} \leq t \leq \overline{c}} \left[ \mathcal{C}(t) - \max_{i} \; l^c_{ip}(t) \right] \hspace{0.7 cm}\text{and} \hspace{0.7 cm} \varepsilon_s := \max_p \max_{\underline{x}_s \leq t \leq \overline{x}_s} \left[ z(p)\mathcal{X}(t) - z(p)\max_{i} l^s_{ip}(t) \right].
\end{equation*}
If the solution to the approximate planner problem is proper, then the overall approximation error is bounded from above by the sum of maximal approximation errors:
\begin{equation*}
0 \le \int \Big(\mathcal{C}(c(p)) + z(p) \big( \mathcal{X}(x_c(p)) + \mathcal{X}(x_m(p)) \big) \Big)\text{d}\pi - \Omega \leq \varepsilon + \varepsilon_c + \varepsilon_m,
\end{equation*}
where  $\Omega$ is the minimum value for the original problem.

\end{proposition}

\begin{proof}
We show that if the solution $(c,x_s, r, r_s)$ to the approximation problem is proper, then the overall approximation error is bounded from above by the sum of maximal approximation errors. We next prove that both inequalities are satisfied. 

The first inequality is satisfied since the approximate allocation $(c,x)$ is feasible. Since the approximate planner's problem produces a feasible solution, we clearly have
\begin{equation}
\Omega \leq \int \left(\mathcal{C}(c(p)) + z(p)\big( \mathcal{X}(x_c(p)) + \mathcal{X}(x_m(p))\big) \right)\text{d}\pi.
\end{equation}

To prove the second inequality, we use that the definition of the approximation error $\varepsilon_c$ and the approximation constraints (\ref{e:auxiliary_c}) implies $\mathcal{C}(c(p)) \leq r_c(p) + \varepsilon_c$.  Denote by $(\hat{c},\hat{x})$ the allocation that attains the minimum resource cost $\Omega := \int \left(\mathcal{C} ( \hat{c} (p) ) + z(p) \big( \mathcal{X}(\hat{x}_c(p)) + \mathcal{X}(\hat{x}_m(p))\big) \right) \text{d}\pi $. Since the solution to the approximate problem is proper, there exist a weight $\lambda \in (0,1)$ such that the convex combination given by $\tilde{c} (p) = \lambda c (p) + (1 - \lambda) \hat{c} (p)$ and $\tilde{x}_s (p) = \lambda x_s (p) + (1 - \lambda) \hat{x}_s (p)$ is a proper allocation. To construct an alternative allocation that is feasible under the approximate problem, we can set:
\begin{align}
\tilde{r}_c(p) & = \max_{i} \; l_{ip}^c(\tilde{c}(p)) \\
\tilde{r}_s(p) & = \max_{i} \; l_{ip}^s(\tilde{x}_s(p))
\end{align}


Since $(c,x_s,r,r_s)$ solves the approximate problem, and since $(\tilde{c},\tilde{x}_s,\tilde{r},\tilde{r}_s)$ is feasible, we know that the cost under the alternative allocation exceeds the cost under the approximate solution. Since the pointwise maximum of convex functions is convex, we have that 
\begin{align}
\tilde{r}(p) & \leq \lambda \max_{i}  l^c_{ip} (c (p) ) + (1 - \lambda) \max_{i}  l^c_{ip} ( \hat{c} (p) ) \leq \lambda r(p) + (1 - \lambda) \mathcal{C} ( \hat{c} (p) ) \\
\tilde{r}_s(p) & \leq \lambda \max_{i} l^s_{ip} (x_s (p) ) + (1 - \lambda) \max_{i} l^s_{ip} ( \hat{x}_s (p) ) \leq \lambda r_s(p) + (1 - \lambda) \mathcal{X}( \hat{x}_s (p) )
\end{align}
where the final inequalities follows from the definition of the approximation constraints (\ref{e:auxiliary_c}), and from the observation that approximations are from below. By combining the two previous claims we write that 
\begin{align*}
& \int \big(r(p) + z(p)\big( r_c(p) + r_m(p) \big)\big)\text{d}\pi  \leq \int \big(\tilde{r}(p) + z(p) \big( \tilde{r}_c(p) + \tilde{r}_m(p) \big) \big)\text{d}\pi \\
\leq \lambda &  \int \big(r_c(p) + z(p)\big( r_c(p) + r_m(p) \big) \big)\text{d}\pi + ( 1 - \lambda ) \int \big(\mathcal{C} ( \hat{c} (p) ) + z(p) \big( \mathcal{X}( \hat{x}_c (p) ) + \mathcal{X} ( \hat{x}_m (p) ) \big) \big)\text{d}\pi ,
\end{align*}
which implies $\Omega \geq \int \big(r_c(p) + z(p)\sum r_s(p)\big)\text{d}\pi$. Finally, we use the definition of the approximation errors to write: 
\begin{align*}
\Omega \geq \int \big(r(p) + z(p)\big( r_c(p) + r_m(p) \big) \big)\text{d}\pi \geq \int \big( \mathcal{C} ( c (p) ) + z(p)\big( \mathcal{X}( x_c (p) ) + \mathcal{X} ( x_m (p) ) \big) \big) \text{d}\pi - \sum \varepsilon
\end{align*}
which concludes the proof.\end{proof}

\begin{algorithm}[!t]

\RestyleAlgo{ruled}

\setstretch{1.65}%

\textbf{Algorithm 1}. Iterative Algorithm for Planner's Problem with Fixed Assignment. 

Set initial location boundaries $\{ \underline{c}, \overline{c}\}$ and $\{\underline{x}_s, \overline{x}_s\}$, define initial accuracy levels $\varepsilon_c$ and $\varepsilon_s$

\While{$\varepsilon_c + \sum\varepsilon_s > \zeta$}{  \vspace{0.2 cm}

for each $p$, construct piecewise linear approximations of $\mathcal{C}$ and $\mathcal{X}$ on bounded intervals $[\underline{c}, \overline{c}]$ and $[\underline{x}_s, \overline{x}_s]$ with  precisions $\varepsilon_c$ and $\varepsilon_s$\\

solve the approximate planner's problem\\

\eIf{\text{the approximate solution is proper}}{ \vspace{0.2 cm}

update precision levels $ \varepsilon_c \to \alpha\varepsilon_c$ and $\varepsilon_s \to \alpha \varepsilon_s$ for some $\alpha < 1$\\

update location boundaries $[\underline{c}, \overline{c}]$ and $[\underline{x}_s, \overline{x}_s]$

}{

relax location boundaries

}

\Return{solution $(c, x_s, r_c, r_s)$ to the final approximate planner's problem.}
\label{alg:optimal_allocation}
}

\end{algorithm}

\vspace{0.2 cm}
\noindent \textbf{Algorithm}. We use an iterative algorithm to solve the approximate planner's problem for a given precision level.\footnote{See \citet{Ekeland:2010} for a discussion on numerically solving optimization problems subject to a convexity constraint on the function $u$, and \citet{Oberman:2013} for a practical approach of dealing with global incentive constraints.} We solve the planner problem for a worker type space with 200 types in both the cognitive and the manual dimension, equivalently, for a total of 40 thousand types. We display the structure of our numerical approach in Algorithm 1.

Having described how to characterize the planner problem given an arbitrary assignment, we next describe how to update the assignment to obtain a jointly optimal assignment and allocation. Given the optimality of positive sorting between workers and firms in Proposition \ref{prop:p_assignment}, we update our assignment after each step by positively sorting the distribution of project values with the effective worker skill index $\mathcal{X}(x_c(p)) + \mathcal{X}(x_m(p))$. By doing so, we reassign projects across workers which yields a new assignment. We then solve the planner's problem for the new assignment function using Algorithm 1. We proceed until the assignment converges. To the best of our knowledge, there is no proof of unique convergence for this iterative procedure. In practice, however, we find that our algorithm always converges to the same assignment function for distinct initial assignments.

\vspace{0.35 cm}
\noindent \textbf{Literature}. Our numerical approach relates to outer linearization of a separable convex objective function. This approach is well-established, see for example, \citet{Bertsekas:2011}. 

Outer linearization of a separable convex objective is part of the outer linearization approach for general problems. For example, \citet{Geoffrion:1970} present the idea of approximating a convex function with supporting hyperplanes, which is outer linearization. \citet{Duran:1986} applies outer linearization to general convex mixed-integer optimization problems. In our case, the objective is separable in variables, which leads to faster and more efficient algorithms for constructing the supporting hyperplanes. \citet{Bertsekas:2011} also have an objective function that is separable in variables and discuss that this problem has been explored under the framework of extended monotropic programming, which builds on monotropic programming \citep{Rockafellar:1984}. Our approach extends beyond outer linearization of a separable convex objective. Whereas the literature routinely focuses on problems with linear equality constraints, our approach also addresses inequality constraints, making it applicable to a broader class of problems.

\vspace{0.4 cm}

\subsection{Characterizing Equilibrium using Transport Problems} \label{p:equilibrium_transport}

To characterize an equilibrium, we relate our positive economy to optimal transport problems. 


\vspace{0.4 cm}
\noindent \textbf{Primal Problem}. The primal problem is to choose an assignment to maximize production:
\begin{equation}
\max_{\gamma \in \Gamma}\; \int y (x_1, x_2, z ) \text{d} \gamma. \label{e:assignment2}
\end{equation}
The choice of an assignment is restricted by the feasibility constraint, $\gamma \in \Gamma(F_x,F_x,F_z)$, where $\Gamma$ denotes the set of probability measures on the product space $\mathbf{X} \times \mathbf{X} \times Z$ such that the marginal distributions of $\gamma$ onto $\mathbf{X}$ and $Z$ are $F_x$ and $F_z$ respectively.   

\vspace{0.4cm}
\noindent \textbf{Dual Problem}. The dual transport problem is to choose functions $w$ and $\Omega$ that solve:
\begin{equation}
\min\limits_{w,\Omega} \; \int w(x_1) \text{d} F_x + \int w(x_2) \text{d} F_x + \int \Omega(z) \text{d} F_z \label{e:pp_dual},
\end{equation}
subject to the constraint that the surplus is weakly negative for any triplet $(x_1,x_2,z)$, that is, $S(x_1,x_2,z) \leq 0$.

\vspace{0.4cm}
\noindent We connect the primal problem and the dual problem to equilibrium in Lemma \ref{c:equilibrium_transport}. 

\vspace{0.15cm}
\begin{lemma}\label{c:equilibrium_transport}
The equilibrium assignment $\gamma$ solves the primal problem (\ref{e:assignment2}), equilibrium wages $w$ and firm values $\Omega$ solve the dual problem (\ref{e:pp_dual}). 
\end{lemma}

\vspace{0.05cm}
\noindent We use Lemma \ref{c:equilibrium_transport} to characterize the equilibrium.\footnote{We note that a transport problem with two identical worker distributions $F_x$ with unit mass for each role is equivalent to a transport problem with a single worker distribution $\Phi_x$ with mass two (Appendix \ref{a:symmetrize}).} We solve the primal problem (\ref{e:assignment2}) to characterize the equilibrium assignment and the dual problem (\ref{e:pp_dual}) to characterize wages $w$ and firm values $\Omega$. To prove Lemma \ref{c:equilibrium_transport}, we use Lemma \ref{l:duality}.

\vspace{0.15cm}
\begin{lemma}\label{l:duality}
Suppose the objectives of the primal problem (\ref{e:assignment}) and the dual problem (\ref{e:pp_dual}) coincide $\int y(x_1,x_2,z) \text{d} \gamma = \int w(x_1) \text{d} F_x + \int w(x_2) \text{d} F_x + \int \Omega(z) \text{d} F_z$. Then $\gamma$ solves the primal problem, and the functions $w$ and $ \Omega$ solve the dual transport problem.
\end{lemma}

\noindent The proof to \Cref{l:duality} only uses of a notion of weak duality. 

\vspace{0.3 cm}
\noindent \textbf{Weak Duality}. Let $\gamma \in \Gamma (F_{x_1}, F_{x_2}, F_z)$ be a joint probability measure, and $(f,g,h)$ be functions such that $y(x_1,x_2,z) \leq f(x_1) + g(x_2) + h(z)$ for all $(x_1,x_2,z)$. Then
\begin{equation}
\min_{f,g,h} \; \int f(x) \text{d}F_{x_1} + \int g(x) \text{d}F_{x_2} + \int h(z) \text{d}F_{z} \; \geq \; \max_{\gamma \in \Gamma} \; \int y(x_1,x_2,z) \text{d}\gamma .  \label{e:weak_duality}
\end{equation}

\vspace{0.4 cm}
\noindent \textit{Proof}. For any functions $(f,g,h)$ so that $y(x_1,x_2,z) \leq f(x_1) + g(x_2) + h(z)$ we have:
\begin{equation*}
\max_{\gamma \in \Gamma} \int y(x_1,x_2,z) \text{d}\gamma \leq \int \big( f(x_1) + g(x_2) + h(z) \big)\text{d}\gamma = \int f(x) \text{d}F_{x_1} + \int g(x) \text{d}F_{x_2} + \int h(z) \text{d}F_{z} ,
\end{equation*}
where the equality follows as $\gamma \in \Gamma (F_{x_1}, F_{x_2}, F_z)$. Since the above inequality holds for any $(f,g,h)$ it holds for $(f,g,h)$ that minimize the right-hand side.

\vspace{0.4 cm}
\noindent We use weak duality to establish \Cref{l:duality} by contradiction. 

\vspace{0.4 cm}
\noindent \textbf{Proof of \Cref{l:duality}}. Suppose by contradiction that $\gamma$ does not solve the planning problem, then
\begin{align}
\max_{\pi \in \Gamma} \int y(x_1,x_2,z) \text{d}\pi & > \int y(x_1,x_2,z) \text{d}\gamma = \int w(x) \text{d}F_{x} + \int w(x) \text{d}F_{x} + \int \Omega(z) \text{d}F_{z} \notag \\
& \hspace{2.71 cm} \geq \min_{f,g,h} \int f(x) \text{d}F_{x} + \int g(x) \text{d}F_{x} + \int h(z) \text{d}F_{z} ,
\end{align}
where the equality follows by assumption. This contradicts weak duality (\ref{e:weak_duality}).

Suppose by contradiction that the functions $\hat{f},\hat{g}$, and $\hat{h}$ do not solve the dual problem. Then there exists functions $f,g$, and $h$ such that
\begin{align}
\min_{f,g,h} \int f(x) \text{d}F_{x} + \int g(x) \text{d}F_{x} + \int h(z) \text{d}F_{z} & < \int w(x) \text{d}F_{x} + \int w(x) \text{d}F_{x} + \int \Omega(z) \text{d}F_{z} \notag \\ 
& =  \int y(x_1,x_2,z) \text{d}\gamma \leq \max_{\pi \in \Gamma} \int y(x_1,x_2,z) \text{d}\pi ,
 \end{align}
where the equality follows by the assumption. This inequality contradicts weak duality (\ref{e:weak_duality}). 

\vspace{0.4 cm}
\noindent We now use \Cref{l:duality} to show that equilibrium assignment $\gamma$ solves the primal transport problem, and that the wage and firm value function solve the dual transport problem. 

In equilibrium, the surplus is negative for any triplet $(x_1,x_2,z)$, which implies that $y (x_1,x_2,z) \leq w(x_1) + w(x_2) +  \Omega(z)$. By substituting the household budget constraints $c=(1-\tau)w(x)$, and the government budget constraint $G = \tau \int w(x) \text{d} \Phi(\alpha)$, into the aggregate resource constraint (\ref{e:resource_constraint}), we write:
\begin{equation}
\int y(x_1,x_2,z) \text{d} \mu = \int w(x_1) \text{d} F_x(x_1) + \int w(x_2) \text{d} F_x(x_2) + \int \Omega(z) \text{d} F_z(z) . \label{e:equilibrium_discipline}
\end{equation}
By \Cref{l:duality} it thus follows that $\mu$ solves the primal problem and $w$ and $\Omega$ solve the dual problem.

\subsection{Symmetric Equilibrium} \label{a:symmetrize}

We prove that we can restrict our attention to symmetric equilibria without loss of generality.
\vspace{0.25cm}
\begin{lemma}\label{p:symmetrize}
For any equilibrium with wages $w$ and assignment function $\gamma \in \Gamma(F_{x_1},F_{x_2},F_z)$, there exists an equilibrium with wages $w$ and a symmetric assignment $\hat{\gamma} = \frac{\gamma + \gamma'}{2}$.
\end{lemma}

\vspace{0.05 cm}
\begin{proof} 
Lemma \ref{p:symmetrize} states that for any competitive equilibrium with wages $w$, firm value $\Omega$, and assignment $\gamma \in \Gamma (F_{x_1},F_{x_2},F_{z})$, there is an equilibrium with identical wages $w$, firm value function $\Omega$ with a symmetric assignment function $\hat{\gamma} := \frac{\gamma + \gamma'}{2} \in \Gamma (F_{x},F_{x},F_{z})$, where $F_{x}:= \frac{1}{2} ( F_{x_1},F_{x_2} )$. 

\vspace{0.4 cm}
\noindent To prove this result, we first define $\gamma'$ as a pushforward measure of the assignment function $\gamma$. We then show that the symmetric assignment function $\hat{\gamma}$ indeed solves the primal transport problem. Using Lemma \ref{c:equilibrium_transport}, this establishes the result.

\vspace{0.25 cm}
\begin{definition*}
Given spaces $M_1$ and $M_2$, a measure $\gamma$ concentrated on $M_1$, and a map $T: M_1 \rightarrow M_2$, the \underline{pushforward measure} of $\gamma$ through $T$, which we denote by $T_{\#} \gamma$, is defined so that:
\begin{equation}
\int f(y) \text{d} T_{\#}\gamma = \int f(T(x)) \text{d} \gamma .
\end{equation}
\end{definition*}

\vspace{0.4 cm}
\noindent We define $\gamma'$ as the pushforward measure of $\gamma$ through a mapping $T$, or $\gamma' := T_{\#}\gamma$. Our mapping $T$ maps from the matching set onto itself interchanging the position of the worker and the coworker, that is, $T:M \rightarrow M$ so that $(x_1,x_2,z) \rightarrow (x_2,x_1,z)$. If $\gamma$ is a feasible assignment, $\gamma'$ is a feasible assignment, that is, for $\gamma \in \Gamma(F_{x_1},F_{x_2},F_{z})$ we have $\gamma' \in \Gamma(F_{x_2},F_{x_1},F_{z})$.

Using the definition of $\gamma'$, we construct symmetric assignment function $\hat{\gamma} := \frac{\gamma + \gamma'}{2}$, and observe that the symmetric assignment function is feasible given $F_{x}$, that is, $\hat{\gamma} \in \Gamma(F_{x},F_{x},F_z)$. Moreover, we observe that: 
\begin{equation}
\int y(x_1,x_2,z) \text{d} \hat{\gamma}(x_1,x_2,z) = \int w(x_1) \text{d} F_x(x_1) + \int w(x_2) \text{d} F_x(x_2) + \int \Omega(z) \text{d} F_z(z) .
\end{equation}
The left-hand side is unchanged as the production of equilibrium pairings does not change, while the right-hand side is unchanged as the skill distribution is unchanged. By \Cref{l:duality},  $\hat{\gamma}$ solves the primal transport problem, and functions $w$ and $\Omega$ solve the dual transport problem. By Lemma \ref{c:equilibrium_transport}, this shows that symmetric assignment $\hat{\gamma}$ is an equilibrium assignment, $w$ are equilibrium wages, and $\Omega$ are equilibrium firm values.

\vspace{0.4 cm}
\noindent Finally, we remark that the equilibrium assignment need not be $\hat{\gamma}$. Specifically, we can replace $\hat{\gamma}$ with any other optimal primal solution. Suppose that there exists another solution to the primal problem $\tilde{\gamma}$, then 
\begin{equation}
\int y(x_1,x_2,z) \text{d} \tilde{\gamma}(x_1,x_2,z) = \int w(x_1) \text{d} F_x(x_1) + \int w(x_2) \text{d} F_x(x_2) + \int \Omega(z) \text{d} F_z(z) ,
\end{equation}
and hence $S(x_1,x_2,z) = 0$ for $\tilde{\gamma}$ almost everywhere.\end{proof}

\subsection{Wages and Firm Values} \label{a:wage_effective}

To see why only effective skill matters, consider two workers $(x_c,x_m)$ and $(\hat{x}_c,\hat{x}_m)$ with identical effective skill $X = \hat{X}$. Since the surplus is zero almost everywhere under equilibrium assignment $\gamma$, and using production technology (\ref{e:firm_tech}), $2 w(x) + \Omega(z) = z X$ and $2 w(\hat{x}) + \Omega(\hat{z}) = \hat{z} \hat{X}$. By the constraints to the dual problem (\ref{e:pp_dual}), $2 w(x) + \Omega(\hat{z}) \geq \hat{z} X = \hat{z} \hat{X}$ and $2 w(\hat{x}) + \Omega(z) \geq z \hat{X} = z X$, where the equalities follow since the workers' effective skills are identical. Combining these expressions, $w(\hat{x}) \geq w(x)$ and $w(x) \geq w(\hat{x})$, so that $w(\hat{x}) = w(x)$. It is useful to define $h(X)$, the firm's total wage bill, as $h(X):= 2w(x)$.  

Wages are convex in effective skill $X$, so small differences in effective worker skill $X$ translate into increasingly large differences in worker earnings. The dual constraints imply $h(X) \geq zX - \Omega(z)$ for any $z$. Since the surplus is zero almost everywhere with respect to the equilibrium assignment $h(X) := \sup\limits_z (  zX - \Omega(z) )$ implying that $h = \Omega^*$, the firm's wage bill is the Legendre transform of the firm value function. Since $h(X)$ is the supremum of linear functions in $X$, the wage function is convex. 

The firm value function is the Legendre transform of the wage bill. The dual constraints also imply that for any $x$ it holds that $\Omega(z) \geq zX - h(X)$ and therefore $\Omega(z) := \sup\limits_X (  zX - h(X) )$. This implies that the firm value function is convex and indeed the Legendre transform of the wage bill $\Omega = h^*$. As a result, $h(X) + h^*(z) = z X$.

\subsection{Lemma \ref{prop:equilibrium_positive}} \label{proof:equilibrium_positive}

We show there exists a firm distribution $F_z$ such that given wage schedule $w$, workers and firms both optimize in a self-sorting equilibrium, where the distribution of worker skills $F_x$ is determined by the worker problems given a talent distribution $\Phi$. We verify this claim by studying the firm and worker problem given the postulated wage schedule (\ref{e:eq_wages}).

\vspace{0.4 cm}
\noindent \textbf{Firm}. Taking the wage schedule $w$ as given, the firm problem of choosing two workers to employ can be written as:
\begin{equation*}
\max_{x_{1},x_{2}} \hspace{0.08 cm} y (x_1,x_2,z) - w(x_1) - w(x_2) = \max_{x_{1},x_{2}} \hspace{0.08 cm} z \left( x_{1c} x_{2c} + x_{1m} x_{2m} \right) - \frac{1}{2} ( x_{1c}^2 + x_{1m}^2 )^\eta - \frac{1}{2} ( x_{2c}^2 + x_{2m}^2 )^\eta
\end{equation*}
where the equality follows from substituting in the production technology (\ref{e:firm_tech}) and wage schedule (\ref{e:eq_wages}). The solution to this problem is that firm $z$ wants to hire two identical workers.

To establish that each firm wants to hire two identical workers, we show that a firm that hires different workers $(x_1,x_2)$ such that $x_1 \neq x_2$ can increase its profits by hiring two identical workers $\big( \frac{1}{2}(x_1 + x_2), \frac{1}{2}(x_1 + x_2)\big)$. By hiring two identical workers, the firms increases its production and decreases its wage bill. Production increases since the worker production technology is concave, that is, $\big( \frac{x_{1c} + x_{2c}}{2} \big)^2 + \big( \frac{x_{1m} + x_{2m}}{2} \big)^2 \geq x_{1c} x_{2c} + x_{1m} x_{2m}$ is implied by $(x_{1c} - x_{2c})^2 + (x_{1m} - x_{2m})^2 \geq 0$. To see that the firm decreases its wage bill by hiring two identical workers, observe that $(x_{1c} - x_{2c})^2 + (x_{1m} - x_{2m})^2 \geq 0$ implies $\big( \frac{x_{1c} + x_{2c}}{2} \big)^2 + \big( \frac{x_{1m} + x_{2m}}{2} \big)^2 \leq \frac{1}{2} \left(x_{1c}^2 + x_{1m}^2 \right) + \frac{1}{2} \left(x_{2c}^2 + x_{2m}^2 \right)$, which implies that for $\eta \geq 1$ we have $\big( \big( \frac{x_{1c} + x_{2c}}{2} \big)^2 + \big( \frac{x_{1m} + x_{2m}}{2} \big)^2 \big)^\eta \leq \big( \frac{1}{2} \left(x_{1c}^2 + x_{1m}^2 \right) + \frac{1}{2} \left(x_{2c}^2 + x_{2m}^2 \right) \big)^\eta$. Since the function $\varsigma^\eta$ is convex $\big(\frac{1}{2} \varsigma + \frac{1}{2} \hat{\varsigma}\big)^\eta \leq \frac{1}{2} \varsigma^\eta + \frac{1}{2} \hat{\varsigma}^\eta$, we obtain $\big( \big( \frac{x_{1c} + x_{2c}}{2} \big)^2 + \big( \frac{x_{1m} + x_{2m}}{2} \big)^2 \big)^\eta \leq \frac{1}{2} \left(x_{1c}^2 + x_{1m}^2 \right)^\eta + \frac{1}{2} \left(x_{2c}^2 + x_{2m}^2 \right)^\eta$ by applying this inequality to the right-hand side. Equivalently, by hiring two identical workers $\big( \frac{1}{2}(x_1 + x_2), \frac{1}{2}(x_1 + x_2)\big)$ a firm lowers their wage bill relative to hiring two different workers $(x_1,x_2)$, $w (\frac{1}{2}(x_1 + x_2) ) + w (\frac{1}{2}(x_1 + x_2) ) \leq w(x_1) + w(x_2)$. Finally, since firms hire identical workers, the firm's optimality condition hiring effective worker skill $X$ gives $z = h'(X)$. Since the wage schedule is convex, equilibrium sorting is positive. 

\vspace{0.4 cm}
\noindent \textbf{Worker}. The distribution of worker skills is uniquely induced by the worker problems. Given the wage schedule, a worker's problem is: 
\begin{equation}
\max_{x_c,x_m} \; u \big( (1 - \tau) w(x) \big) - \kappa \Big( \frac{x_c}{\alpha_c} \Big)^\rho - \kappa \Big( \frac{x_m}{\alpha_m} \Big)^\rho .
\end{equation}
Using a transformation $\tilde{x}_s := x_s^\rho$, the problem is:
\begin{equation}
\max_{\tilde{x}_c,\tilde{x}_m} \; u \big( (1 - \tau) \tilde{w}(\tilde{x})  ) \big) - \frac{\tilde{x}_c}{p_c}  - \frac{\tilde{x}_m}{p_m} ,
\end{equation}
where $\tilde{w}(\tilde{x}) := \Big( \frac{1}{2} \sum \tilde{x}_s^{\frac{2}{\rho}} \Big)^\eta$. 

We prove strict concavity of the objective by examining each of the terms in the objective. The second and third term are linear and thus concave. We remain to verify that the first term, $u ( (1 - \tau) \tilde{w}(\tilde{x}) )$, is strictly concave. First, since the consumption utility $u$ is strictly concave:
\begin{align*}
\lambda & u ((1 - \tau) \tilde{w}(x) ) + ( 1 - \lambda ) u ((1 - \tau) \tilde{w}(\tilde{x}) ) < u ( (1 - \tau) (\lambda  \tilde{w}(x) + ( 1 - \lambda )  \tilde{w}(\tilde{x}) )
\end{align*}
Since $u$ is an increasing and concave function, the first term is strictly concave if $\tilde{w}(\tilde{x})$ is strictly concave. To establish this, we note that the transformed wage equation is a composite function of a concave CES aggregate with an increasing concave function as long as $\eta \leq \frac{\rho}{2}$. The worker problem is strictly concave and thus has a unique solution, which implies the distribution of worker skills is uniquely induced by the worker's problem. 

To complete the proof we show there exists a distribution of firm projects such that the wage equation (\ref{e:eq_wages}) is an equilibrium wage function. Since the wage bill is continuously differentiable, $z = h'(X)$. We can use this expression to uniquely pin down a distribution of firm project values that rationalizes the wage equation. Since the wage bill is convex, the inferred distribution indeed implies positive sorting between effective worker skills and firm project values.\footnote{For the parametric specification, $\Omega(z) := \sup (  zX - h(X) )$ can be characterized as $\Omega(z) = \mathcal{C}_z z^{\frac{\eta}{\eta-1}} - 2 \zeta$, where $\mathcal{C}_z$ is a multiplicative constant independent of $z$. Together with the dual constraint, this closed-form expression for the firm value allows us to characterize project value $z$ without relying on the derivative of the firm's wage bill in the quantitative section.}

\subsection{Frisch Elasticity} \label{a:worker_problem}

We next show how to derive the expression for the Frisch elasticity of labor supply within our model. Adding the optimality conditions across tasks gives $\ell_c^{\rho} + \ell_m^{\rho} = \mathcal{C} w(x) \lambda$, with constant $\mathcal{C}:= (1-\tau) \frac{\eta}{\kappa \rho}$. Using the constant effort shares implied by (\ref{e:skill_ratio}), and multiplying and dividing by $\ell^\rho := (\ell_c + \ell_m)^\rho$ we write $\ell^\rho = \mathbb{C} w(x) \lambda$, where $\mathbb{C} := \mathcal{C} \big/ \big( \big( \frac{\ell_c}{\ell} \big)^{\rho} + \big( \frac{\ell_m}{\ell} \big)^{\rho} \big)$ is constant across workers. To obtain the Frisch elasticity implied by our model, we relate a worker's total efforts to earnings per hour $z(x) = w(x) / (\ell_c + \ell_m)$ as $\ell^{\rho - 1} = \mathbb{C} z(x) \lambda$ to obtain (\ref{e:frisch_calibration}):
\begin{equation}
\varepsilon = \frac{\partial \log (\ell_c + \ell_m)}{\partial \log z(x)} \bigg\vert_{\lambda} = \frac{\partial \log (\ell_c + \ell_m)}{\partial \log (1 - \tau)} \bigg\vert_{\lambda}  = \frac{1}{\rho - 1} \tag{\ref{e:frisch_calibration}}.
\end{equation}

\subsection{Implementation} \label{s:implementationtax}

In this appendix, we describe an implementation of the optimum through an income tax system. This argument follows the analysis in \citet{Kocherlakota:2010} for an environment with multidimensional private information. To illustrate this argument, let $\{ c(\alpha), x(\alpha) \}$ denote the solution to the planning problem for a finite type space $A$.

We first note that the consumption allocation depends on type $\alpha$ only through the allocation of tasks $x(\alpha)$. If $x(\alpha) = x(\hat{\alpha})$ while $c(\alpha) > c(\hat{\alpha})$, worker $\hat{\alpha}$ would pretend being type $\alpha$ to attain more consumption for identical task inputs and the planner allocation is not incentive compatible. The consumption allocation can thus be written as a function of the task inputs $c(\alpha) = \tilde{c}(x(\alpha))$, where the domain of the consumption function is finite and given by $\textbf{X} := \cup_{\alpha} \{ x(\alpha) \}$. 

The consumption function is increasing in both cognitive and manual task disutility as the tax system has to reward workers for providing higher levels of task inputs. Suppose $x_c(\alpha) =  x_c(\hat{\alpha})$, $x_m(\alpha) > x_m(\hat{\alpha})$ while $c(\alpha) \leq c(\hat{\alpha})$. In this case, worker $\alpha$ reports $\hat{\alpha}$, so the planning allocation is not incentive compatible. 

While the consumption allocation function $\tilde{c}$ has a finite domain, we next extend to all $x$ such that $x_c \geq \min\limits_{\alpha \in A} x_c(\alpha)$ and $x_m \geq \min\limits_{\alpha \in A} x_m(\alpha)$. For all such task inputs, define
\begin{equation}
\hat{c}(x) := \max\limits_{x' \leq x} \; \tilde{c}(x')
\end{equation}
subject to $x' \in \textbf{X}$ if the maximizer exists and $\hat{c}(x) = 0$ otherwise. We also set $\hat{c}(x) = -\infty$ for all $x$ outside the domain. We define the tax system $T(x)$ over the same values of $x$ as:
\begin{equation}
T(x) = w(x) - \hat{c}(x) . 
\end{equation}

\begin{proposition}
Tax system $T$ implements the planner allocation.
\end{proposition}

\begin{proof}
In order to establish the result, we show that workers of skill $\alpha$ indeed choose the allocation for worker $\alpha$ under the planner allocation.

The problem of worker $\alpha$ given the tax system $T$ is to choose the level of cognitive and manual task inputs $x$ to solve:
\begin{equation}
\max\limits_x \; u(c) - v \left( \frac{x_c}{\alpha_c} \right) - v \left( \frac{x_m}{\alpha_m} \right)
\end{equation}
subject to the constraint $c \leq w(x) - T(x) = \hat{c}(x)$. The problem can thus be written as:
\begin{equation}
\max\limits_x \; u(\hat{c}(x)) - v \left( \frac{x_c}{\alpha_c} \right) - v \left( \frac{x_m}{\alpha_m} \right) .
\end{equation}
Choosing cognitive and manual task inputs below their minimum would be suboptimal as the worker pay infinite taxes. We thus restrict the attention to task inputs $x \geq \big(\min\limits_{\alpha} x_c(\alpha), \min\limits_{\alpha} x_m(\alpha) \big)$. 

We next show that workers choose a task allocation $x \in \textbf{X}$. By contradiction, suppose that the worker instead chooses an allocation $\hat{x} \notin \textbf{X}$. Hence, the worker attains consumption $\hat{c}(\hat{x})$, which by the definition of the consumption allocation function is given by:
\begin{equation}
\hat{c}(\hat{x}) = \max\limits_{x' \leq \hat{x}} \; \tilde{c}(x')
\end{equation}
subject to $x' \in \textbf{X}$. Since $\hat{x} \notin \textbf{X}$, the worker can do better by reducing their work effort choosing $x \in \textbf{X}$ that delivers the same consumption level. The worker consumes the same, exerting less effort. As a result, we restrict the choice to $x \in \textbf{X}$ without loss of generality. 

Worker $\alpha$ chooses the optimal bundle of task inputs $x \in \textbf{X}$, which boils down to choosing $x$ such that for all $x' \in \textbf{X}$
\begin{equation}
u \big( \tilde{c}(x) \big) - v \left( \frac{x_c}{\alpha_c} \right) - v \left( \frac{x_m}{\alpha_m} \right) \geq u \big(\tilde{c}(x') \big) - v \left( \frac{x'_c}{\alpha_c} \right) - v \left( \frac{x'_m}{\alpha_m} \right)  
\end{equation}
By definition of the consumption allocation function over the domain $\textbf{X}$ it follows that this is equivalent to:
\begin{equation}
u \big( c(\alpha) \big) - v \left( \frac{x_c(\alpha)}{\alpha_c} \right) - v \left( \frac{x_m(\alpha)}{\alpha_m} \right) \geq u\big( c(\alpha') \big) - v \left( \frac{x_c(\alpha')}{\alpha_c} \right) - v \left( \frac{x_m(\alpha')}{\alpha_m} \right)  
\end{equation}
for all $\alpha' \in A$. That worker $\alpha$ optimally chooses allocation $(c(\alpha), x(\alpha))$, implementing the planner allocation, then follows from the incentive constraints.\end{proof}

\end{document}